\documentclass[onefignum,onetabnum]{siamart190516}

\usepackage{psfrag}
\usepackage{amssymb}
\usepackage{lipsum}
\usepackage{amsfonts}
\usepackage{graphicx}
\usepackage{epstopdf}
\usepackage{algorithmic}
\ifpdf
  \DeclareGraphicsExtensions{.eps,.pdf,.png,.jpg}
\else
  \DeclareGraphicsExtensions{.eps}
\fi

\newcommand{\fR}{f_R}
\newcommand{\fL}{f_L}
\newcommand{\Pe}{\text{Pe}}
\newcommand{\p}{{\bf p}}
\newcommand{\q}{{\bf q}}
\newcommand{\e}{{\bf e}}
\newcommand{\x}{{\bf x}}

\newcommand{\X}{{\bf X}}

\newcommand{\ud}{\mathrm{d}}

\newcommand{\txi}{\tilde \xi}
\newcommand{\tx}{\tilde \x}
\newcommand{\tr}{\tilde r}
\renewcommand{\tt}{\tilde \theta}
\newcommand{\rhop}{\tilde \rho}
\newcommand{\fp}{\tilde f}
\newcommand{\pp}{\tilde \p}
\newcommand{\bfb}{{\bf b}}

\newcommand{\Dphi}{D_e(\phi)}
\newcommand{\weak}{\rightharpoonup}

\newsiamremark{remark}{Remark}
\newsiamremark{hypothesis}{Hypothesis}
\crefname{hypothesis}{Hypothesis}{Hypotheses}
\newsiamthm{claim}{Claim}

\headers{Active Brownian particles}{M. Bruna, M. Burger, A. Esposito and S. M. Schulz}

\title{Phase Separation in Systems of Interacting Active Brownian Particles\thanks{
Submitted to the editors \today.
\funding{This work of the first author was partially funded by a Royal Society University Research Fellowship (grant number URF/R1/180040) and a Humboldt Research Fellowship from the Alexander von Humboldt Foundation. The second and the third author gratefully acknowledge support from the German Science Foundation (DFG) through CRC TR 154  ``Mathematical Modelling, Simulation and Optimization Using the Example of Gas Networks". The third author was supported by the Advanced Grant Nonlocal-CPD (Nonlocal PDEs for Complex Particle Dynamics: Phase Transitions, Patterns and Synchronization) of the European Research Council Executive Agency (ERC) under the European Union’s Horizon 2020 research and innovation programme (grant agreement No. 883363).
The fourth author was funded by the Royal Society Award (RGF/EA/181043).}}}

\author{Maria Bruna\thanks{Department of Applied Mathematics and Theoretical Physics, University of Cambridge, Cambridge CB3 0WA, UK 
  (\email{bruna@maths.cam.ac.uk)}.}
\and Martin Burger\thanks{Department Mathematik, Friedrich-Alexander Universit\"at Erlangen-N\"urnberg, Cauerstr. 11, D 91058 Erlangen, Germany 
  (\email{martin.burger@fau.de}).}
\and Antonio Esposito\thanks{Mathematical Institute, University of Oxford, Oxford OX2 6GG, UK (\email{antonio.esposito@maths.ox.ac.uk}).}
\and Simon M. Schulz\thanks{Department of Mathematics, University of Wisconsin-Madison, Van Vleck Hall, 480 Lincoln Dr, Madison, WI 53706, USA, \email{(smschulz2@wisc.edu)}.}}

\usepackage{amsopn}

\ifpdf
\hypersetup{
  pdftitle={Active Brownian particles},
  pdfauthor={M. Bruna, M. Burger, A. Esposito and S. M. Schulz}
}
\fi

\begin{document}

\maketitle

\begin{abstract}
The aim of this paper is to discuss the mathematical modeling of Brownian active particle systems, a recently popular paradigmatic system for self-propelled particles. We present four microscopic models with different types of repulsive interactions between particles and their associated macroscopic models, which are formally obtained using different coarse-graining methods. The macroscopic limits are integro-differential equations for the density in phase space (positions and orientations) of the particles and may include nonlinearities in both the diffusive and advective components. In contrast to passive particles, systems of active particles can undergo phase separation without any attractive interactions, a mechanism known as motility-induced phase separation (MIPS). We explore the onset of such a transition for each model in the parameter space of occupied volume fraction and P\'eclet number via a linear stability analysis and numerical simulations at both the microscopic and macroscopic levels. We establish that one of the models, namely the mean-field model which assumes long-range repulsive interactions, cannot explain the emergence of MIPS. In contrast, MIPS is observed for the remaining three models that assume short-range interactions that localize the interaction terms in space. 
\end{abstract}

\begin{keywords}
  self-propelled particles, phase separation, excluded-volume interactions, stability analysis
\end{keywords}

\begin{AMS}
  35Q84, 35R09, 35B35, 60J70, 35C20
\end{AMS}

\section{Introduction}

Active matter systems consisting of many interacting self-propelled particles occur in many applications ranging from synthetic self-propelled colloids \cite{Buttinoni:2013de}, microtubules \cite{Sumino.2012}, and bacterial suspensions \cite{Berg:1993ug}, to large-scale systems such as fish schools, bird flocks, and collective robotics \cite{Vicsek:2012gp}. 

Individual-based models of active matter can be broadly divided into velocity-jump processes and active Brownian walks. 
Velocity-jump processes consist of a sequence of runs and reorientations at randomly distributed times, when a new velocity is chosen \cite{Othmer:1988kx}. One of the best-known examples in nature is the run-and-tumble of E. coli bacteria, which are flagellated bacteria that move in roughly straight lines and constant velocity interrupted by sudden reorientations \cite{Berg:1993ug}. In contrast, active Brownian particles (ABPs) are used to model particles that change their orientation gradually by Brownian motion. The prototypical experimental systems are synthetic colloids (often with energy transfer based on photoactivity) with an extensive range of prospective applications (cf. \cite{shields2017evolution,zottl2016emergent}).
However, ABPs can also be seen as a model system for a broader class of self-propelled particles including velocity-jump processes, as the latter can be shown to reduce to ABPs after a suitable time- and space-rescaling in some cases \cite{Cates:2013ia}.
As a result, ABPs are natural model systems studied frequently in condensed matter physics, which incorporates typical pattern formation effects \cite{Bialke:2013gw,Buttinoni:2013de,Cates:2014tr,Redner.2013,Romanczuk:2012iz,Solon:2015hz,Speck:2015um,Stenhammar:2015ex,Wittkowski:2014dt}. A (single) ABP in a two-dimensional space evolves according
\begin{align} \label{sde_single}
	\ud \X = \sqrt{2 D_T} \ud {\bf W} + v_0 \e(\Theta) \ud t,\qquad 
	\ud \Theta = \sqrt{2 D_R} \ud W,
\end{align}
where $\X \in \mathbb R^2$ and $\Theta \in [0, 2\pi)$ denote the position and orientation respectively of the particle, $D_T$ and $D_R$ are the translational and rotational diffusion coefficients, respectively, $v_0$ is the \emph{self-propulsion} speed, $\e(\theta) = (\cos\theta, \sin \theta)$ is the orientation vector, and $\bf W$ and $W$ are independent two- and one-dimensional Wiener processes, respectively. 

A central difference between active Brownian particles and their passive counterparts is that, since ABPs consume energy to move, active matter systems are inherently out-of-equilibrium.
Because of this, ABPs display much richer behavior. One of the most intriguing ones is the presence of phase separation despite only having purely repulsive interactions, and the emergence of macroscopic behaviour traditionally associated with passive systems with attractive or attractive-repulsive interactions. This mechanism is known as \emph{Motility-induced phase-separation} (MIPS), as it is due to the interplay between self-propulsion and excluded-volume interactions \cite{Cates:2014tr}.

Active systems have been widely studied by simulation and phenomenological models in the recent years, in particular in the (theoretical) physics community. Most work is based on extensive simulations at the microscopic particle level. However, those are computationally very expensive, particularly in the regime where MIPS appears (which corresponds to high concentrations of particles, for which interactions must be computed, and high self-propulsion speeds). Macroscopic models help to explore the phase plane better and get a deeper understanding of the mechanisms. Some of such models derived in the literature are \cite{Bialke:2013gw,Cates:2014tr, Speck:2015um,Tailleur:2008kd,Wittkowski:2014dt}.
Let us also mention the kinetic model of the interface between dilute and dense regions with absorption-evaporation dynamics in \cite{Redner.2013}.   

Here we consider different approaches to obtain macroscopic models, review existing ones from the physics literature, and also provide new systematic derivations. We bring them all under the same mathematical framework and study the MIPS via the stability of associated differential operators.
This allows us to compare and contrast what is the effect of specific microscopic properties but also crucially, what the implications of different assumptions made in the upscaling/coarse-graining are. In particular, we present two derivations of a Brownian active model \cref{sde_single} coupled with volume-exclusion interactions resulting in novel macroscopic models for active particles. To our knowledge, these models are the first ones to incorporate nonlinear cross-diffusion terms. These reduce to an existing model in the literature when the diffusive operator is linearised around the homogeneous equilibrium. 

Mathematically speaking, the resulting models are $2d -1$-dimensional, where $d$ is the physical dimensional. Throughout this work, we will use the term one- or two-dimensional to refer to models in one or two \emph{physical} dimensions, respectively. In the one-dimensional case, the resulting model is like a two-species reaction-diffusion-advection system, with two species of left- and right-biased particles. In the two-dimensional case, the resulting models are integro-differential equations for the density in phase space (space of position and orientation). These can also be seen as infinite-dimensional reaction-diffusion-advection models, with species parameterized by their orientation.

The structure of the paper is as follows. In the remainder of this section, we summarize the microscopic models considered and the resulting macroscopic models. \cref{sec:derivation} is concerned with the derivation of the four models using a variety of methods: a mean-field approximation, a phenomenological near-equilibrium approximation, and the method of matched asymptotic expansions. Rescaled models in terms of only two non-dimensional parameters, the occupied volume fraction $\phi$ and the P\'eclet number $\Pe$ are given in \cref{sec:rescale}. The presence of MIPS is considered via a stability analysis of the homogeneous stationary state in \cref{sec:stability}, again using different methods: the entropy dissipation, analytic linear stability analysis of the symmetrized differential operators, and numerical stability analysis of the original full operators.  Finally, in \cref{sec:numerics} we present some two-dimensional numerical examples of the patterns associated with MIPS, obtained from both the stochastic microscopic models and the macroscopic models.
The rigorous analysis of one of the local macroscopic models, displaying a nonlinearity in the advection term but linear diffusion in space and orientation is addressed in \cite{BruBurEspSch-an-phen21}. The analysis of the two macroscopic models that, in addition to the nonlinearity in the advection term, display nonlinear cross-diffusion-like terms will be the subject of future work.

\subsection{Microscopic models} \label{sec:micro_models}
We consider two alternative microscopic descriptions for particles in two spatial dimensions. The first is a set of $N$ ABPs as in \cref{sde_single} coupled with a short-range repulsion between particles:
\begin{subequations}\label{sde_all}
\begin{align}
\label{sde_x}
	\ud \X_i &= \sqrt{2 D_T} \ud {\bf W}_i - \chi\sum_{j\ne i} \nabla u((\X_i-\X_j)/\ell)\ud t + v_0 \e(\Theta_i) \ud t,\\
	\label{sde_angle}
	\ud \Theta_i &= \sqrt{2 D_R} \ud W_i,
\end{align}
\end{subequations}
where ${\bf W}_i$ and $W_i$ are independent two- and one-dimensional Wiener processes, respectively, for $1\le i \le N$. In \cref{sde_x}  $\chi$ and $\ell$ represent the strength and the range of the interaction potential $u$, respectively. 
Here $D_T$ and $D_R$ are the translational and rotational diffusion coefficients respectively, and $\e(\theta) = (\cos \theta, \sin \theta)$ is the self-propulsion direction. Angles can take values in $[0, 2\pi)$ and we consider periodic boundary conditions in both positions and angles.

We consider isotropic particles (so, in particular, $u$ is radially symmetric and a function of $\|\X_i-\X_j\|$), and the angle $\Theta_i$ only represents an internal variable that determines the direction of the drift term. Of course, there are possible extensions considering elongated particles and for $\Theta_i$ to represent the angle of the main axis of the particle. One particular model we will consider is hard-core interactions, for which $u$ is a hard-core potential, $u_\text{HS} (r) = + \infty$ for $r<1$, and 0 otherwise.

The second microscopic model is based on a simple exclusion process on a two-dimensional lattice. The angles $\Theta_i$ are still defined continuously and undergo a Brownian motion \cref{sde_angle}, while a discrete process, namely an asymmetric simple exclusion process, governs the positions $\X_i$. Jumps in position are only carried out in a finite set of orientations $\e_j$ with rates depending on $\e(\Theta_i)\cdot \e_j$ (the angle biases the random walk towards the direction $\e(\Theta_i)$). In this case, the interactions between particles are modeled as follows: if the destination lattice site is already occupied, the jump is aborted. If we consider a regular lattice with distance $\ell$ between sites, this model can be seen as the lattice-based version of Brownian hard-core particles ($u_\text{HS}$) of diameter $\ell$. 

\subsection{Macroscopic models} \label{sec:macro}

Let us  denote by $f(\x,\theta,t)$ the macroscopic density of particles with angle $\theta$, and the corresponding \emph{space density}
\begin{align}\label{local_rho}
	\rho(\x,t) &:= \int_0^{2\pi} f(\x,\theta,t) ~\ud\theta,
\end{align}
at position $\x$ in the two-dimensional torus $\Omega := \mathbb T^2 = [0,1]^2$.
Another relevant quantity is the \emph{polarisation} (also known as orientational order parameter)
\begin{align}\label{polarisation}
	\p(\x,t) := \int_0^{2\pi} \e(\theta) f(\x,\theta,t) ~\ud\theta.
\end{align}

In the next section we will sketch the derivation of the following models:
\begin{description}
\item[Model 1:] Nonlocal model for soft repulsive particles via a mean-field approximation
\begin{equation*} \label{model1}
\partial_t f + v_0\nabla \cdot ( f \e(\theta)) = D_T \Delta f + D_R \partial_{\theta}^2 f  + \nabla \cdot ( f \nabla (u \ast \rho)).
\end{equation*}

\item[Model 2:] Local model for soft repulsive particles from \cite{Bialke:2013gw, Speck:2015um}
\begin{align*} 
\partial_t f + v_0\nabla \cdot (  f (1-\phi \rho) \e(\theta)) &= D_{T}(1-\phi)^2 \Delta f + D_R \partial_{\theta}^2 f,  
\label{model2}
\end{align*}
where $\phi\in[0,1)$ is an effective occupied fraction for soft interacting particles.
\item[Model 3:] Hard-sphere particles model derived via matched asymptotic expansions
\begin{align*} 
\partial_t f + v_0\nabla \cdot \left[ f (1-\phi \rho) \e(\theta) + \phi \p f\right] &\!=\! D_T \nabla \cdot \left[ (1- \phi \rho) \nabla f + 3 \phi f \nabla  \rho \right]  + D_R \partial_{\theta}^2 f , 
\label{model3}
\end{align*}
where $\phi$, given in \cref{phi_hs}, is again a measure of the occupied fraction in the system.
\item[Model 4:] Active simple exclusion process via a mean-field approximation
\begin{align*}
\partial_t f + v_0\nabla \cdot [ f (1- \phi \rho) \e(\theta)] &= D_T \nabla \cdot ( (1- \phi \rho) \nabla f + \phi f \nabla \rho)  + D_R \partial_{\theta}^2 f,
\label{model4}
\end{align*}
where $\phi \in[0,1)$ is the occupied fraction in the system.
\end{description}  

Above and throughout the paper, $\nabla$, $\nabla \cdot$, and $\Delta$ denote the gradient, divergence, and Laplacian in the spatial variables $\x$, respectively. When the operators are taken in other variables, these are explicitly included in the subscript. All models are considered in the three-dimensional torus $\Upsilon := \Omega \times [0, 2\pi)$ with $|\Omega| = 1$. We impose the normalization conditions
$$
\int_{\Upsilon} f(\x,\theta, t)  \ud \theta \ud \x = 1, \qquad  \int_\Omega \rho(\x, t) \ud \x = 1.
$$

\section{Derivation of the macroscopic models} \label{sec:derivation}

Our starting point for all the continuous or off-lattice models in position (Models 1,2 and 3 above) is to define the joint probability density for $N$ particles evolving according to \cref{sde_x}-\cref{sde_angle}, $F_N(\vec \xi, t)$ with $\vec \xi = (\xi_1, \dots, \xi_N)$ and $\xi_i = (\x_i, \theta_i)$. By using the Chapman--Kolmogorov equation, see e.g., \cite[Chapter 3]{erban_chapman_2020}, this is given by
\begin{equation} \label{N_eq}
	\partial_t F_N(\vec \xi, t) \! = \! \sum_{i=1}^N \nabla_{\x_i} \cdot \left[ D_T \nabla_{\x_i} F_N -v_0 \e(\theta_i) F_N + \chi\nabla_{\x_i} U_\ell(\x_1, \dots, \x_N) F_N \right] + D_R \partial^2_{\theta_i} F_N,
\end{equation}
where 
$$U_\ell(\x_1,\dots, \x_N) = \sum_{1\le i < j \le N} u((\x_i -\x_j)/\ell).
$$

The domain of definition of one particle's coordinates is $\xi = (\x, \theta) \in \Upsilon$, with periodic boundary conditions in both angle and space.
The goal is to obtain a macroscopic model for the one-particle density $f(\xi, t)$, which we can describe by picking the first particle since all particles are identical, i.e.
\begin{align}
	f(\xi_1,t) = \int_{\Upsilon^{N-1}} F_N(\vec \xi) \ud \xi_2 \dots \ud \xi_N.
\end{align}
To this end, keeping in mind all the particles are indistinguishable, we integrate \cref{N_eq} with respect to $\xi_2, \dots, \xi_N$. Using periodicity, all the terms for $i\ge 2$ in the right-hand side of \cref{N_eq} vanish, and we are left with:
\begin{equation} \label{1_eq}
	\partial_t f(\xi_1, t) = \nabla_{\x_1} \cdot \left[ D_T \nabla_{\x_1} f -v_0 \e(\theta_1) f + {\bf G}(\xi_1,t) \right] + D_R  \partial^2_{\theta_1 }f,
\end{equation}
with 
\begin{align} \label{interaction_G}
\begin{aligned}
	{\bf G}(\xi_1,t) &= \chi \int_{\Upsilon^{N-1}} F_N(\vec \xi, t) \sum_{i=2}^N \nabla_{\x_1} u((\x_1-\x_i)/\ell) \ud \xi_2, \dots, \ud \xi_N \\
	&= \chi (N-1) \int_{\Upsilon} F_2(\xi_1,\xi_2,t) \nabla_{\x_1} u((\x_1-\x_2)/\ell) \ud \xi_2,
	\end{aligned}
\end{align}
where $F_2$ is the two-particle density
$$
F_2(\xi_1,\xi_2,t) = \int_{\Upsilon^{N-2}} F_N(\vec \xi,t) \ud \xi_3 \dots \ud \xi_N.
$$
\subsection{Mean-field model} \label{sec:model1}

The mean-field scaling corresponds to $\chi = 1/N$ and $\ell = O(1)$ so that we have a weak and long range interaction in the limit of $N\to \infty$. Here we only give a heuristic derivation of such limit, which has been proven rigorously for passive Brownian particles \cite{Jabin:2017fb} (corresponding to setting $D_R = 0$ in model \cref{sde_all}). The result is a propagation of chaos for suitable conditions on the potential $u$, which implies the factorisation of the second marginal as $N\to \infty$, $P_2(\x_1,\x_2, t) \weak p(\x_1, t) p(\x_2, t)$.
For our active Brownian model, since the interaction in \cref{interaction_G} is independent of the angles, one expects an analogous result 
\begin{equation}
	\label{mfa}
	F_2 (\xi_1, \xi_2, t) = f(\xi_1,t) f(\xi_2, t).
\end{equation}
Inserting \cref{mfa} into \cref{interaction_G}, equation \cref{1_eq} reads
\begin{subequations}
	\label{mfa_final}
\begin{equation}
	\label{1_mfa}
	\partial_t f(\xi_1, t) = \nabla_{\x_1} \cdot \left[ D_T \nabla_{\x_1} f -v_0 \e(\theta_1) f +  f \nabla_{\x_1} \mathcal U \right] + D_R \partial^2_{\theta_1} f,
\end{equation}
with the interaction term
\begin{equation} \label{U_functional}
	\mathcal U(f) = \frac{N-1}{N} \int_\Upsilon f(\xi_2,t) u((\x_1-\x_2)/\ell) \ud \xi_2 \to  \int_\Omega \rho(\x_2,t) u((\x_1-\x_2)/\ell) \ud \x_2,
\end{equation}
\end{subequations}
as $N \to \infty$.

We assume $u$ to be a purely repulsive potential, being monotonically decreasing with respect to the radius $r$.  Moreover, we also require integrability at the origin, that is, $u(r) = O(r^{-2+\delta})$ for $\delta>0$ as $r\to 0$ (in two dimensions).   
Below we shall assume the condition
\begin{equation}\label{condition_u_1}
	 \iint  \nabla_\x g \cdot \nabla_\x {\cal U}(g)~\ud \x ~\ud\theta \geq 0,
\end{equation}
for arbitrary sufficiently smooth nonnegative functions $g$, 
which can be used to show exponential decay of the solution $f$ to the homogeneous steady state $f_*$ by means of the logarithmic Sobolev inequality (see \cref{sec:stabmodel1}). The condition is satisfied if $u$ is a concave potential, but also in the formal asymptotic limit of $\ell \rightarrow 0$, where we find ${\cal U}(f) = c f$, thus
\begin{equation} \label{nonlinear-term_mfa}
	\int \nabla_{\x_1} \rho(\x_1) \cdot \nabla_{\x_1} {\cal U}(f)~\ud \x_1 = c \int |\nabla_{\x_1} \rho(\x_1) |^2 \ud \x_1 \int u(\|\tx\|)  \ud \tx \ge 0.
\end{equation}

\subsection{Phenomenological model} \label{sec:model2}

In \cite{Bialke:2013gw, Speck:2015um} they derive an equation similar to \cref{mfa_final} but with a nonlinearity in the advection term. As we will see in \cref{sec:stability}, this additional term makes a big difference in the stability of the models. Here we give the main idea behind this model, and leave the details for \cref{sec:model2_SM}. The starting point is \cref{1_eq}-\cref{interaction_G} but with short-range and strong repulsive interactions, namely, $\chi = 1$ and $\ell = \epsilon \ll 1$. As a result, the interaction term $\bf G$ reads
\begin{align}
	{\bf G}(\xi_1,t) = (N-1) \int_{\Omega} f_2(\xi_1,\x_2,t) \nabla_{\x_1} u((\x_1-\x_2)/\epsilon) \ud \x_2,
\end{align}
where $f_2(\xi_1,\x_2,t) := \int F_2 ~\ud \theta_2$ is the two-body probability density to find another particle at $\x_2$ (with arbitrary orientation) together with the tagged particle at $\x_1$ with orientation $\theta_1$. They proceed with a decomposition of the force along $\e(\theta_1)$ and its perpendicular direction, which they assume to be parallel to $\nabla_{\x_1} f$ at leading order
\begin{align}
\mathbf{G} (\xi_1,t) =G_\e \e(\theta_1) +\delta \mathbf{G} \approx G_\e \e(\theta_1)  + G_{\|} \nabla_{\x_1} f.
	\label{interaction_model2}
\end{align}
The coefficient $G_\e$ is approximated as $G_\e = (N-1) \zeta f(\xi_1,t) \rho(\x_1,t)$ with $\zeta$ constant assuming that the system is homogeneous and neglecting the time-dependence of the pair correlation function, while $G_{\|}$ is taken to be a function of $\phi$ only \cite{Bialke:2013gw}. 

Inserting \cref{interaction_model2} into \cref{1_eq} one obtains 
\begin{subequations}\label{mod2}
\begin{align} \label{eq_physics}
	\partial_t f = \nabla_{\x_1} \cdot \left[D_e \nabla_{\x_1} f - v_e \e(\theta_1) f  \right] + D_R \partial_{\theta_1}^2 f,
\end{align}
with effective diffusion $D_e = D_T + G_{\|}$ and effective speed $v_e = v_0 - (N-1) \rho \zeta$. 
In \cite{Bialke:2013gw} they assume that both $D_e$ and $v_e$ are constants, taking $\rho \equiv 1$ uniform. They note that $D_e$ corresponds to the self-diffusion coefficient in a passive suspension ($v_0 = 0$). In \cite{Speck:2015um} they consider $v_e = v_e (\rho)$. 

The theory in \cite{Bialke:2013gw, Speck:2015um} does not provide expressions for $G_{\|}$ and $\zeta$; these are instead measured from simulations. Here we define $\phi$ such that $v_0\phi = (N-1)\zeta$; the constant $\phi$ can be seen as an effective occupied area, noting also that $\zeta$ scales like the area of the interaction $\epsilon^2$ and $v_0$ (assuming that the correlation function grows linearly with $v_0$, which is reasonable as the larger the speed, the larger correlation lengths). Therefore, we can write the effective diffusion coefficient and speed as 
\begin{equation}\label{model2_params}
D_e = \Dphi, \qquad 
v_e = v_0(1 - \phi \rho), \qquad 0 \le \phi < 1.
\end{equation}
For our subsequent analysis, we will use $\Dphi =D_T(1-\phi)^2$. 
The form of $\Dphi$ is chosen such that $D_e'(0) = -2 D_T$ as the self-diffusion coefficient for hard-spheres \cite{Hanna:1982gi} and $D_e(1) = 0$ (for the maximum packing density). Combining  \cref{eq_physics} and \cref{model2_params} and integrating in $\theta$ one finds
\begin{equation}
	\partial_t \rho + \nabla \cdot ( v_e \p) = \Dphi \Delta \rho.
\end{equation}
\end{subequations}

The one-dimensional version of \cref{mod2} coincides with a crowded version of the Goldstein--Taylor model \cite{goldstein1951diffusion,taylor1922diffusion} for the densities of  left- and right-moving particles, $\fL(x,t)$ and $\fR(x,t)$ respectively,
\begin{align} \label{GT_model}
\begin{aligned}
\partial_t \fR + v_0\partial_x [\fR (1-\phi \rho)] &= D_T\partial_{xx} \fR + k (\fL - \fR), \\
\partial_t \fL - v_0\partial_x [\fL (1-\phi \rho)] &= D_T\partial_{xx} \fL + k (\fR - \fL).
\end{aligned}
\end{align}
Comparing this with the two-dimensional version \cref{mod2}, we find that here $\Dphi \equiv D_T$ and $k = 2D_R/\pi^2$.\footnote{This can be seen from discretising $\partial_\theta^2 f$ using centred differences with a grid of spacing $\pi$: $D_R\partial_\theta^2 \fR \approx D_R(\fL + \fL - 2\fR)/\pi^2 = (2D_R/\pi^2)(\fL - \fR)$.} We outline the derivation of \cref{GT_model} in \cref{sec:GTmodel}.
A rigorous well-posedness theory for \cref{mod2,GT_model} is provided in \cite{BruBurEspSch-an-phen21}.

\subsection{Active Brownian particles model}
\label{sec:model4}

In this section, we consider a systematic derivation of a macroscopic model for active Brownian particles in the case of very short-range repulsive interactions. In particular, instead of the mean-field scaling ($\chi=1/N, \ell =1$) we consider strong ($\chi=1$) but short-ranged ($\ell = \epsilon \ll 1$) interactions, and an asymptotic expansion for $N \epsilon^2$ small. This approach works well for singular potentials for which the mean-field approach breaks down (see \cref{nonlinear-term_mfa}), the most extreme of cases being the hard-sphere potential.  

To this end, in this section we consider the following hard-core interaction potential: $u(r) = +\infty$ for $r<1$, 0 otherwise, so that particles are hard spheres of diameter $\epsilon$. Therefore, we will have to consider the boundary conditions $|\X_i(t)-\X_j(t)|\ge \epsilon$ for all $j\ne i$ at all times.

When dealing with a hard-core potential of range $\epsilon$, the domain of definition of \cref{N_eq} depends on $\epsilon$:
\begin{equation}
	\Upsilon_\epsilon^N = \Omega_\epsilon^N \times [0,2\pi)^N,
\qquad
\Omega_\epsilon^N = \left \{ \vec x \in \Omega^N \ : \   |\x_i - \x_j| \ge \epsilon, \forall i \ne j \right \}.
\end{equation}
Then in addition to the periodic boundary conditions on the external boundaries $\partial \Upsilon^N$, we have to take into account no-flux boundary conditions on the internal boundaries of $\partial \Upsilon_\epsilon^N$. For $N=2$, the problem looks as follows:
\begin{align} \label{F2_eq}
	\partial_t F_2 (\xi_1,\xi_2,t) &= \nabla_{\xi_1} \cdot \left[ \bar D \nabla_{\xi_1} F_2 - s(\theta_1) F_2 \right] + \nabla_{\xi_2} \cdot \left[ \bar D \nabla_{\xi_2} F_2 - s(\theta_2) F_2 \right],
\end{align}
where $\bar D = \text{diag}(D_T, D_T, D_R)$ and $s(\theta) = v_0(\cos \theta, \sin \theta, 0)$, 
and boundary condition
\begin{align} \label{bc_2}
	\left[ D_T \nabla_{\x_1} F_2 - v_0 \e(\theta_1) F_2 \right]\cdot {\bf n}_1 + \left[ D_T \nabla_{\x_2} F_2 - v_0 \e(\theta_2) F_2 \right]\cdot {\bf n}_2 = 0
	\end{align}
on $|\x_1 - \x_2| = \epsilon$, where ${\bf n}_2 = - {\bf n}_1$ are the normal vectors on the collision surface. 

We want to obtain the equation for the one-particle density $f(\xi_1,t)$, which for hard spheres is given by
$$
f(\xi_1,t) = \int_0^{2\pi} \ud \theta_2 \int_{\Omega(\x_1) } F_2(\xi_1,\xi_2,t) \ud \x_2,
$$
where $\Omega(\x_1)$ is the spatial domain available to the second particle, $\x_2$, when the first particle is at $\x_1$ (perforated domain). 
 Integrating \cref{F2_eq} with respect to $\xi_2$ we obtain
 \begin{align} \label{integrated_eq_hs}
 \begin{aligned}
 	\partial_t f & = \nabla_{\xi_1} \cdot \left[ \bar D \nabla_{\xi_1} f - s(\theta_1) f \right] 
 	+ \int_0^{2\pi} \ud \theta_2 \int_{\partial B_\epsilon(\x_1)} \left[D_T \nabla_{\x_2} F_2 - v_0 \e(\theta_2) F_2 \right] \cdot {\bf n}_2 \ud S_{\x_2} \\
 	&\qquad + \int_0^{2\pi} \ud \theta_2 \int_{\partial B_\epsilon(\x_1)} \left[v_0 \e(\theta_1) F_2 - 2D_T \nabla_{\x_1} F_2 - D_T \nabla_{\x_2} F_2 \right] \cdot {\bf n}_2 \ud S_{\x_2}\\
 	& = \nabla_{\xi_1} \cdot \left[ \bar D \nabla_{\xi_1} f - s(\theta_1) f \right] 
 	- D_T \int_0^{2\pi} \ud \theta_2 \int_{\partial B_\epsilon(\x_1)} \left[\nabla_{\x_1} F_2 + \nabla_{\x_2} F_2  \right] \cdot {\bf n}_2 \ud S_{\x_2},
 	\end{aligned}
\end{align}
where $\partial B_\epsilon(\x_1) = \{ |\x_1-\x_2| = \epsilon \}$. The last term in the first line comes from applying the divergence theorem on the $\nabla_{\xi_2}\cdot $ term in \cref{F2_eq}, while the term in the second line comes from applying the Reynolds theorem on the $\nabla_{\xi_1}\cdot $ term in \cref{F2_eq}. The equality in the last line comes using the boundary condition \cref{bc_2}. 

We then need an expression for the two-particle density $F_2$ when two hard-sphere particles are in contact. Clearly, in this case, the approximation \eqref{mfa} is not suitable. We seek an approximation via matched asymptotic expansions instead.

Given the hard sphere potential, we suppose that when two particles are far apart ($|\x_1 - \x_2| \gg \epsilon$) they are independent, whereas when they are close to each other ($|\x_1 - \x_2| \sim \epsilon$) they are correlated. We designate these two regions of configuration space the outer region and inner region, respectively. In the outer region we define $F_\text{out}(\xi_1,\xi_2, t) = F_2(\xi_1,\xi_2, t)$. By independence we have that\footnote{Independence only tells us that $F_\text{out}(\xi_1,\xi_2,t) \sim q(\xi_1,t)q(\xi_2,t)$ for some function $q$, but the normalization condition on $F_2$ implies $f = q + O(\epsilon)$.}
\begin{equation} \label{outer}
F_\text{out} (\xi_1,\xi_2, t) = f(\xi_1,t) f(\xi_2,t) + \epsilon F_\text{out}^{(1)} (\xi_1,\xi_2, t) + O(\epsilon^2),
\end{equation}
for some function $F_\text{out}^{(1)}$. In the inner region, we set $\xi_1 = \txi_1$, and $\xi_2 = \txi_1 + \text{diag}(\epsilon, \epsilon, 1) \txi$, or $\x_2 = \tx_1 + \epsilon \tx$, $\theta_2 = \tilde \theta_1 + \tilde \theta$ and define $\tilde F(\txi_1,\txi, t) = F_2(\xi_1,\xi_2, t) $. 
Rewriting \cref{F2_eq} in terms of the inner coordinates gives 
\begin{multline}
\label{P2_inner} 
\epsilon^2 \partial_t \tilde F
=   2 D_T \Delta_{\tx} \tilde F + \epsilon v_0  \nabla_{\tx} \cdot \left[ \widehat \e(\tt_1,\tt) \tilde F \right ] -  2\epsilon D_T \nabla_{\tx_1}  \cdot   \nabla_{\tx}  \tilde F \\ 
 + \epsilon^2 D_T\Delta_{\tx_1} \, \tilde F - \epsilon ^2  v_0 \nabla_{\tx_1} \cdot \left[ \e(\tt_1)  \tilde F \right ]
+ \epsilon^2 D_R \left( \partial_{\tt_1}^2 \tilde F - 2\partial_{\tt_1} \partial_{\tt} \tilde F + 2 \partial_{\tt}^2 \tilde F \right),
\end{multline}
where $\widehat \e(\tt_1,\tt) = \e(\tt_1) -  \e(\tt_1 + \tt) $. The boundary condition \cref{bc_2} when two particles are in contact becomes in inner coordinates
\begin{equation}\label{bc2_inner}
	2 D_T ~ \tx \cdot  \nabla_{\tx} \tilde F = \epsilon \tx \cdot \left [ D_T \nabla_{\tx_1} \tilde F - v_0 \widehat \e(\tt_1, \tt) \tilde F \right], \qquad \text{on} \qquad |\tx| = 1.
\end{equation}
The inner solution $\tilde F$ must match with the outer solution $F_\text{out}$ as $\|\tx \| \to \infty$.  Expanding $F_\text{out}$ in terms of the inner variables gives (omitting the time variable for ease of notation)
\begin{align}
\label{P2bc_match}
\begin{aligned}
F_\text{out} ( \xi_1, \xi_2) &\sim  f(\txi_1) f(\tx_1 + \epsilon \tx, \tt_1 + \tt) + \epsilon F_\text{out}^{(1)}(\tx_1, \tt_1, \tx_1 + \epsilon \tx, \tt_1 + \tt)  \\
& \sim  f(\txi_1) f(\tx_1, \tt_1 + \tt) + \epsilon f(\txi_1) \tx \cdot \nabla_{\tx_1} f(\tx_1, \tt_1 + \tt)\\
&\quad + \epsilon  F_\text{out}^{(1)}(\tx_1, \tt_1, \tx_1, \tt_1 + \tt) \\
& \sim  f f^+ + \epsilon f \tx \cdot \nabla_{\tx_1} f^+ + \epsilon  F_\text{out}^{(1)}(\tx_1, \tt_1, \tx_1, \tt_1 + \tt),
\end{aligned}
\end{align}
where $f \equiv f(\txi_1,t)$ and $f^+ = f(\tx_1, \tt_1 + \tt,t)$.

We look for a solution of \cref{P2_inner,bc2_inner} matching with \cref{P2bc_match}  as $\|\tilde {\bf x} \|\rightarrow \infty$ of the form
$\tilde F  \sim \tilde F^{(0)} + \epsilon\tilde F^{(1)} +  \cdots$.
The leading-order inner problem is 
\begin{align} \label{hs_so1}
\begin{aligned}
0 &= 2 D_T\Delta_{\tx}   \tilde F^{(0)},  \\
0&= 2 D_T ~ \tx \cdot  \nabla_{\tx} \tilde F^{(0)}, \quad |\tx| = 1\\
\tilde F^{(0)} &\sim f f^+ \quad \textrm{as} \quad |\tx|\to \infty,
\end{aligned}
\end{align}
with solution $F^{(0)} =  f f^+$.

The $O(\epsilon)$ problem reads
\begin{align} \label{hs_order1}
\begin{aligned}
0&= \Delta_{\tx}  \tilde F^{(1)}\\
\tx \cdot  \nabla_{\tx} \tilde F^{(1)} &= \tx \cdot {\bf A}(\tx_1, \tt_1, \tt), \qquad |\tx| = 1\\
\tilde F^{(1)} &\sim  \tx \cdot {\bf B}(\tx_1, \tt_1, \tt)  \qquad \textrm{as} \qquad |\tx|\sim \infty,
\end{aligned}
\end{align}
where
\begin{align*}
	{\bf A}(\tx_1, \tt_1, \tt) &= \frac{1}{2D_T} \left[ D_T \nabla_{\tx_1}(f f^+) - v_0 \widehat \e(\tt_1,\tt) f f^+ \right],\\
	{\bf B}(\tx_1, \tt_1, \tt) &= f\nabla_{\tx_1} f^+. 
\end{align*}
The solution of \cref{hs_order1} is
\begin{equation} 
\label{hs_order1_sol}
\tilde F^{(1)} = \tx \cdot {\bf A} + \left( \tx + \frac{\tx}{|\tx|^2} \right) ({\bf B} - {\bf A}).
\end{equation}

Now we consider the integral term in \cref{integrated_eq_hs}, 
\begin{equation}
	\label{integral_hs}
	\mathcal I = - D_T \int_0^{2\pi} \ud \theta_2 \int_{\partial B_\epsilon(\x_1)} \left[\nabla_{\x_1} F_2 + \nabla_{\x_2} F_2  \right] \cdot {\bf n}_2 \ud S_{\x_2}.
\end{equation}
Since the domain of integration of $\mathcal I$ is in the inner region, we use the inner region solution to compute it. Rewriting $\mathcal I$ in terms of the inner variables and using the factorization at leading order 
with \cref{hs_order1_sol} we find
\begin{align} \label{integral_result}
\begin{aligned}
	\mathcal I &= \frac{\pi}{2} \epsilon^2 \nabla_{\tx_1} \cdot \int_0^{2\pi} \left[ 3 D_T f \nabla_{\tx_1} f^+ - D_T f^+ \nabla_{\tx_1} f + v_0\widehat \e(\tt_1, \tt) f f^+ \right] \ud \tt \\
	&= \frac{\pi}{2} \epsilon^2 \nabla_{\tx_1} \cdot  \left \{ 3 D_T f \nabla_{\tx_1} \rho - D_T \rho \nabla_{\tx_1} f + v_0 f [\e(\tt_1) \rho - \p ]\right \},
	\end{aligned}
\end{align}
where $\rho(\tx_1,t)$ and $\p(\tx_1,t)$ are given in \cref{local_rho} and \cref{polarisation} respectively.  Combining \cref{integrated_eq_hs} and \cref{integral_result} we arrive at
  \begin{align} \label{eq_hs_final}
 	\partial_t f(\xi_1,t) & = \nabla_{\x_1} \cdot \left[ D_T \nabla_{\xi_1} f - v_0 \e(\theta_1) f \right] 
 	+ \mathcal I + D_R \partial_{\theta_1}^2 f.
\end{align}
This is the equation for $N=2$. For a general $N$, the tagged particle (first particle) will have $N-1$ inner regions, so the term $\mathcal I$ in \cref{eq_hs_final} will be premultiplied by $N-1$, leading to 
\begin{align}
    \begin{aligned}
	\partial_t f = &\ \nabla_{\x_1} \cdot \left[ D_T \nabla_{\x_1} f - v_0 \e(\theta_1) f  \right]  + D_R \partial_{\theta_1}^2 f \\
	& + (N-1)\frac{\epsilon^2\pi}{2} \nabla_{\x_1} \cdot  \left \{ 3 D_T f \nabla_{\x_1} \rho - D_T \rho \nabla_{\x_1} f + v_0 f [\e(\theta_1) \rho - \p ]\right \}  
    \end{aligned}
\end{align}
We introduce the dimensionless quantity 
\begin{equation}\label{phi_hs}
	\phi = (N-1)\frac{\epsilon^2 \pi}{2},
\end{equation}
noting that for large $N$ it is approximately twice the occupied area of $N$ hard spheres of diameter $\epsilon$. 
Therefore we can think of $\phi$ as a rescaled occupied area. Using $\phi$ we can rewrite the resulting hard-spheres model as
\begin{equation}\label{model42}
\partial_t f + v_0 \nabla \cdot \left[  f (1-\phi \rho) \e(\theta) + \phi \p f\right] = D_T \nabla \cdot \left[ (1- \phi \rho) \nabla f + 3 \phi f \nabla \rho \right]  + D_R \partial_{\theta}^2 f.
\end{equation}
Integrating \cref{model42} with respect to $\theta$ we arrive at
\begin{equation}\label{model4_rho}
\partial_t \rho + v_0 \nabla \cdot  \p  = D_T \nabla \cdot \left[ (1+ 2\phi \rho) \nabla \rho \right],
\end{equation}
while multiplying \cref{model42} by $\e(\theta)$ and integrating we arrive at
\begin{equation}
\label{model4_p}
		\partial_t \p + v_0 \nabla \cdot \left[ (1-\phi \rho) {\bf P}  + \phi \p \otimes \p \right]  =  D_T\nabla \cdot \left[ (1-\phi \rho) \nabla \p + 3\phi \p \otimes \nabla \rho \right] - D_R \p,
\end{equation}
where $\bf P$ is the second moment
$$
{\bf P} =\int_0^{2\pi} f \e(\theta) \otimes  \e(\theta) \, \ud \theta.
$$
Therefore we do not obtain a closed model for $\rho$ and $\p$ and therefore must solve for $f$ to obtain the evolution of the first two moments. An alternative would be to write a closure for $\bf P$ or drop it altogether as done in \cite{Bialke:2013gw}. We note that the effective diffusion coefficient in \cref{model4_rho} is consistent with the collective diffusion coefficient for hard spheres, $D_c(\rho)  = 1 + 2\phi \rho + O(\phi^2)$ \cite{bruna2012excluded}.

\subsection{Active simple exclusion lattice model}
\label{sec:model3}

We now present our model of an active exclusion process on a lattice. We consider a jump process on a two-dimensional lattice, where jumps are only carried out in a finite set of discrete orientations $\e_j$, $j=1,\ldots,m$. The paradigmatic example is, of course a rectangular lattice, but essentially we only need the symmetry conditions 
\begin{equation}
    \sum_{j=1}^m \e_j = 0,   \qquad  \sum_{j=1}^m \e_j \otimes \e_j = c I,
\end{equation}
for some constant $c  > 0$. Here we consider a two-dimensional regular square lattice with spacing $\epsilon$, and $|\e_j|=1$ such that $c = 2$ (this is independent of the dimension). Therefore, given a lattice site $\x$, its neighbouring sites are given by $\x+\epsilon \e_j$, $j=1,\ldots,4$. Note our choice of lattice spacing $\epsilon$, which coincides with the diameter in our hard-spheres model in the previous section. This emphasises that, in both cases, the interactions are local and of range equal to the distance between two particles at contact. 

The set of orientations for the $N$ particles in the two-dimensional model is continuous and described, as in the other models, by an angle $\theta$ evolving according to Brownian motion \cref{sde_angle}. The jump rates to neighbouring sites are depending on the local relative orientations $\e(\theta) \cdot \e_j$, namely $$ \pi_j(\theta) = \alpha_\epsilon \exp(\beta_\epsilon \e(\theta) \cdot \e_j).$$
We assume that $\beta_\epsilon \sim \epsilon$ and $\alpha_\epsilon \sim \epsilon^{-2}$, the latter being mainly time rescaling. As usual in simple exclusion models, the jump is only executed if the target site $\x+\epsilon\e_j$ is empty. Note that this corresponds to an asymmetric simple exclusion process (ASEP), with the bias being a function of orientation instead of position as in the standard ASEP for passive particles. A similar microscopic model with slightly different rates $\pi_j(\theta)$ is considered by Erignoux in \cite{Erignoux:2016un}. Instead of a continuous Brownian motion in angle, Erignoux considers a Glauber jump process that models alignment in angle depending on its neighbors. 

Denoting by $f_\epsilon(\x,\theta,t)$ the probability density for finding a particle at site $\x$ with orientation $\theta$ at time $t$, it evolves according to the master equation
\begin{subequations} \label{lattice_der}
\begin{equation}
	\partial_t f_\epsilon(\x,\theta,t) = \sum_j \pi_j(\theta) \left[ Q_\epsilon(\x-\epsilon\e_j,\theta,\x,t) -  Q_\epsilon(\x,\theta,\x+\epsilon\e_j,t)\right] + D_R \partial_{\theta\theta} f_\epsilon(\x,\theta,t),
\end{equation}
where $Q_\epsilon(\x,\theta,\x+\epsilon\e_j,t)$ is the probability density to find a particle at $\x$ with orientation $\theta$ and no particle at all at $\x+\epsilon\e_j$. 
Using a simple mean-field closure assumption
\begin{equation} \label{MFA_lattice}
	Q_\epsilon(\x,\theta,\x+\epsilon\e_j,t) = f_\epsilon(\x,\theta,t) (1-\rho_\epsilon(\x + \epsilon \e_j,t)), \quad \rho_\epsilon(\x,t) = \int_0^{2
\pi} f_\epsilon(\x,\theta,t)~\ud\theta
\end{equation}
we then obtain a closed system. 
\end{subequations}

The last step is to take the hydrodynamic limit $N\to \infty, \epsilon \to 0$ while keeping the occupied fraction $\phi := N \epsilon^2$ finite. This converts the master equation for the  discrete number density $f_\epsilon$ into a PDE for a continuous probability density  $f(\x,\theta,t)$ defined for all $\x \in \Omega$, which can be approximated as 
$$
f(\x_i,\theta,t) \approx \frac{1}{N} \frac{f_\epsilon(\x_i,\theta,t)}{\epsilon^2},
$$
where $\x_i$ is the centre of the $i$th compartment.
Using a standard asymptotic expansion as $\epsilon \rightarrow 0$ while keeping $N \epsilon^d = \phi$ finite in \cref{lattice_der}, we obtain
\begin{equation} \label{lattice_final}
    \partial_t f + \nabla \cdot \left[ v_0 \e(\theta) (1- \phi \rho)f \right] = \nabla \cdot \left[ D_T ((1- \phi \rho) \nabla f + \phi f \nabla \rho) \right]
    + D_R \partial_{\theta}^2 f,
\end{equation}
with 
$$v_0 = c \lim_{\epsilon \downarrow 0} \epsilon \alpha_\epsilon \beta_\epsilon, \qquad D_T = \frac{c}2
\lim_{\epsilon \downarrow 0} \epsilon^2 \alpha_\epsilon.$$
For a square lattice, we then have $\alpha_\epsilon = D_T /\epsilon^2$ and $\beta_\epsilon = v_0 \epsilon/(2 D_T)$. Integrating \cref{lattice_final} with respect to $\theta$ yields
\begin{equation}
\partial_t \rho + v_0\nabla \cdot \left[  (1- \phi \rho) \p \right] = D_T \Delta \rho.
\end{equation}

We briefly outline the differences between \cref{lattice_final} and the hydrodynamic limit obtained by \cite{Erignoux:2016un}. The result in \cite{Erignoux:2016un} is rigorous, building on the multi-type exclusion model introduced by Quastel \cite{Quastel:1992iv}; in particular it does not rely on the mean-field approximation  \cref{MFA_lattice}. A slight disadvantage is that, like in Quastel's result, the resulting PDE depends on a non-explicit density-dependent self-diffusion coefficient. A related active lattice gas model was considered in \cite{kourbane2018exact}, which takes a more straightforward explicit form thanks to modifying the exclusion rule by allowing neighboring particles to diffuse by swapping their positions. The result is linear diffusion in their hydrodynamic limit (see Eqs. (3,4) in \cite{kourbane2018exact}), which agrees with the crowded Goldstein--Taylor model \cref{GT_model} presented above.

\subsection{Rescaled version of the macroscopic models} \label{sec:rescale}
In what follows, it is convenient to rescale the macroscopic models to reduce the number of parameters. 
We rescale the densities in each model with the parameter $\phi$, which measures how crowded the system is, although due to the different nature of the models, it varies slightly for each of them. For soft-interacting particles (models 1 and 2), \cref{sde_all} leads to an area fraction of the order
$$
\phi \sim N \chi  \ell^2,
$$
where we recall that $|\Omega|=1$ and $\ell$ is the range of the interaction potential. In the mean-field limit (Model 1), $\chi = 1/N$, $\ell= 1$ such as $\phi = 1$. In Model 2, we have $\phi \sim N \zeta/v_0$ (we expect it to be of order $N\epsilon^2$ using that the range of the potential is $\ell = \epsilon$). This is also the case in the hard-exclusion Models 3 and 4, where $\phi \sim N \epsilon^2$ (in Model 3, it is $\phi  = (N-1)\epsilon^2 \pi/2$, in Model 4 it is exactly $\phi = N\epsilon^2$.

We rescale time $t = T \hat t$ and space $\x = L \hat \x$ with $T = D_R^{-1}$ and $L = \sqrt{D_T/D_R}$ such that in the rescaled system the diffusion coefficients are both equal to one. Note that $\phi$ remains unchanged as it is dimensionless. We define the rescaled velocity or P\'eclet number $\Pe = v_0/\sqrt{D_R D_T}$, as well as $\hat u(\hat \x) =  u(\x)/\sqrt{D_R D_T}$, $\hat D_e(\phi) = \Dphi/ D_T$. Finally, we introduce the mass density $\hat f(\hat \x, \theta, \hat t) = \phi f(\x,\theta, t)$, and similarly for $\rho$ and $\p$. Note that this implies a mass rescaling $\int_\Upsilon \hat f \ud \xi = \int_\Omega \hat \rho \ud \x = \phi$.  Inserting these in the four models of the previous section, and dropping the hats, we obtain:
\begin{align}
 \label{model1_c} \tag{M1}
\partial_t f + \Pe\nabla \cdot ( f \e(\theta) ) &= \Delta f + \partial_{\theta}^2 f  + \nabla \cdot ( f \nabla (u \ast \rho )),\\
\label{model2_c} \tag{M2}
\partial_t f + \Pe\nabla \cdot (  f (1- \rho) \e(\theta) ) &= \Dphi \Delta f +  \partial_{\theta}^2 f,  \\
\label{model3_c} \tag{M3}
\partial_t f + \Pe\nabla \cdot \left[ f (1- \rho) \e(\theta) + \p f\right] &= \nabla \cdot \left[ (1- \rho) \nabla f + 3 f \nabla  \rho \right]  + \partial_{\theta}^2 f,  \\
\label{model4_c} \tag{M4}
\partial_t f + \Pe\nabla \cdot [ f (1- \rho) \e(\theta) ] &= \nabla \cdot ( (1- \rho) \nabla f + f \nabla \rho)  + \partial_{\theta}^2 f.
\end{align}

\section{Stability and Instability of Homogeneous Stationary States} \label{sec:stability}

In the following, we investigate the stability, respectively instability, of homogeneous stationary states to understand the possible onset of phase separation. The obvious first step is to verify the existence of homogeneous stationary states, which is possible due to the periodic boundary conditions:
\begin{lemma}
The homogeneous state $f_*(\x,\theta) = \frac{\phi}{2\pi}$ with mass $\phi \in [0,1]$ is a stationary state of the two-dimensional models \cref{model1_c}, \cref{model2_c}, \cref{model3_c}, and \cref{model4_c}.
\end{lemma} 

As we shall see, the homogeneous state for Model 1 is fully stable under generic conditions, which somehow rules out this mean-field model for a possible description of phase separation effects. We study the remaining models via a linear stability analysis. In \cref{sec:1dinstability} we consider the one-dimensional version of Model 2, for which we can solve the linear stability explicitly by Fourier series. Below we tackle the two-dimensional Models 2 to 4 by solving explicitly the eigenvalue problems associated with their symmetric operators and numerical simulations. 

Throughout this section, we shall only consider perturbations with a mean value of zero, since the homogeneous stationary states are a one-parameter family of the mass (mean value) of the solution. Thus they cannot be fully stable in terms of constant perturbations.

\subsection{Stability of  Model 1} 
\label{sec:stabmodel1}

In the case of the mean-field model \cref{model1_c} we can establish nonlinear stability of the homogeneous solution for arbitrary mass. We note that
$$ \frac{d}{dt} \int_\Upsilon  f \log \frac{f}{f_*} ~\ud \xi  = - \int_\Upsilon  \left[ - \Pe\nabla f \cdot \e(\theta) + \frac{|\nabla f|^2}f + \frac{|\partial_\theta f|^2}f + \nabla f \cdot \nabla {\cal U}(f) \right]~\ud \xi, $$
where recall that $\mathcal U = u\ast \rho$.
Using the periodic boundary conditions, 
$$ \int_\Upsilon  \nabla f \cdot \e(\theta) ~\ud \xi  =  \int_\Upsilon  \nabla \cdot( f \e(\theta) ) ~\ud \xi = 0, $$
and the standard assumptions on repulsive forces (see \cref{sec:model1})
$$ \int_\Upsilon \nabla f \cdot \nabla {\cal U}(f))~\ud \xi  \geq 0,$$ 
we arrive at
$$ \frac{d}{dt} \int_\Upsilon  f \log \frac{f}{f_*} ~\ud \xi  \leq  - \int_\Upsilon \frac{|\nabla_\xi f|^2}f~\ud \xi. $$
The logarithmic Sobolev inequality, cf.~\cite[Corollary 1.1]{DolEstKowLoss_log_Sob}, implies exponential decay of $f$ to $f_*$ as it gives
\[
\frac{d}{dt} \int_\Upsilon  f \log \frac{f}{f_*} ~\ud \xi\le -c \int_\Upsilon  f \log \frac{f}{f_*} ~\ud \xi,
\]
where $c>0$ is a constant. Convergence to the stationary state is then a consequence of the Gronwall lemma and the Csis\'{z}ar--Kullback inequality $\int f\log(f/f_*)\ud\xi\ge(1/2)\|f-f_*\|_{L^1}^2$. We refer the reader to \cite{MarVil_EquilFP} for more details.

\subsection{Linear Stability of the two-dimensional Models 2, 3, and 4}

We consider the linear stability of the two-dimensional Models 2 to 4, \cref{model2_c}-\cref{model4_c}. The remark below will be useful to this end.
\begin{remark}[Antisymmetric part] \label{rem:antisymmetric} Stability of the symmetric operator implies stability of the full operator, but the converse is not true. To see this, suppose that $\partial_t f =  L(f)$, with $L =  L^S +  L^{AS}$ the symmetric and antisymmetric parts, respectively. If all the eigenvalues of $L^S$ are negative, then (in the $L^2$ scalar product with formal adjoints)
$$
0 \ge \langle  L^S (f), f \rangle  =  \langle  L^S (f), f \rangle +  \langle L^{AS} (f), f \rangle = \langle \partial_t f, f \rangle =  \frac{d}{dt} \frac{1}{2} |f|^2,
$$
using that $\langle L^{AS} (f), f \rangle = \langle f,  (L^{AS})^* (f) \rangle = - \langle f, L^{AS} (f) \rangle$, so any perturbation will decay over time. In other words, $L^{AS}$ cannot change the sign of the full operator as its eigenvalues are purely imaginary. 
Conversely, pure imaginary eigenvalues of $L^{AS}$ can move the spectrum of $L$ towards the negative real part.
\end{remark}

To this end, we consider the linearisation of the equations around the homogeneous state $f_*=\frac{\phi}{2\pi}$, $\rho_* = \phi$. 
We insert $f = f_* + \delta \fp$, $\rho = \rho^* + \delta \rhop$, and $\p = {\bf 0} + \delta \pp$ with $\rhop = \int_0^{2 \pi} \fp \ud\theta$, $\pp = \int_0^{2 \pi} \e(\theta) \fp \ud\theta$, and $\delta \ll 1$ into the equations. The resulting linearised problems are
all of the form, for $i = 2, 3, 4$,
\begin{equation}
	\label{general_form_stability}
	\partial_t \fp + (1-\phi) \Pe \nabla \cdot ( \fp \e(\theta) ) = \alpha_i(\phi)  \Delta \fp + \partial_{\theta}^2 \fp + \phi \nabla \cdot ({\cal L}_i(\fp)),
\end{equation}
with $\alpha_2(\phi) = \Dphi$, $\alpha_3(\phi) = \alpha_4(\phi) = 1-\phi$ and ${\cal L}:H^1(\Upsilon) \rightarrow L^2(\Omega)$ a linear and bounded operator, given by
\begin{equation} \label{operator_L_models}
{\cal L}_2 =  \frac{\Pe}{2\pi}\rhop \e(\theta) ,\qquad {\cal L}_3 =   \frac{\Pe}{2\pi}\left(\rhop \e(\theta)  -  \pp \right) + \frac{3}{2\pi} \nabla \rhop, \qquad {\cal L}_4 = \frac{\Pe}{2\pi}\rhop \e(\theta) + \frac{1}{2\pi} \nabla \rhop.
\end{equation}

\begin{lemma}\label{lem:skew-symmetric-operator}
The operator 
$$ {\cal T} :H^1(\Upsilon) \rightarrow L^2(\Upsilon), \quad f \mapsto  \nabla \cdot (f \e(\theta) ) = \e(\theta) \cdot \nabla f, $$
is skew-symmetric.
\end{lemma}
\begin{proof}
	\begin{align*}
\int_\Upsilon \nabla \cdot (   f \e(\theta) ) h ~\ud \x ~\ud \theta  = - \int_\Upsilon  f \e(\theta) \cdot \nabla h ~\ud \x ~\ud \theta = - \int_\Upsilon  f   \nabla \cdot (h \e(\theta) ) ~\ud \x ~\ud \theta.
\end{align*}
\end{proof}

Using \cref{rem:antisymmetric}, we obtain the following result. 	
\begin{theorem} \label{th:stable}
For $\phi$ sufficiently small, the homogeneous stationary state is linearly stable.
\end{theorem}
\begin{proof}
For $\phi = 0$, $\alpha_i(0) =1$ for $i = 2, 3, 4$ and the linearized equation \cref{general_form_stability} reduces to
$$\partial_t \fp +   \Pe \nabla \cdot ( \fp \e(\theta) ) = \Delta \fp + \partial_{\theta}^2 \fp.$$ Hence, for $\phi = 0$, the symmetric part of \cref{general_form_stability} reduces to the three-dimensional heat equation, which is linearly stable. Therefore, following \cref{rem:antisymmetric} and \cref{lem:skew-symmetric-operator}, \cref{general_form_stability} is linearly stable for $\phi= 0$. For $\phi$ small enough, any unstable contribution from the symmetric part of $\mathcal L_i$ will still be controlled by the Laplacian part due to the mapping properties of $\mathcal L_i$ and the coercivity of the negative Laplacian on the subspace of functions with zero mean value. 
\end{proof}

In \cref{sec:lin_inst_symm,sec:num_instability}, we determine what ``small enough'' $\phi$ in \cref{th:stable} means in terms of $\Pe$ for each model.

\subsubsection{Linear Instability of the symmetric operator} \label{sec:lin_inst_symm}

We consider the general form \cref{general_form_stability} of the linearised Models 2, 3, and 4 and investigate the instability of the symmetrized operator in a unified way for $\Pe$ or $\phi$ sufficiently large. 
Inserting a  perturbation of the form $\tilde f (\x, \theta, t) = e^{\lambda t} \hat f(\x,\theta)$ into \cref{general_form_stability} we obtain:
\begin{equation} \label{def_L}
	\lambda \hat f = L_i (\hat f):= - (1-\phi) \Pe \nabla \cdot ( \hat f \e(\theta) )+ \alpha_i(\phi)  \Delta \hat f + \partial_{\theta}^2 \hat f + \phi \nabla \cdot ({\cal L}_i(\hat f)),
\end{equation}
so $\lambda$ corresponds to the eigenvalues of the linear operator $L_i$. 
Following the previous discussion, any region of instability of $L_i$ in parameter space must be contained in the region of instability of its associated symmetric operator $L^S_i = (L_i + L_i^*)/2$, where $L_i^*$ is the formal adjoint operator. Using \cref{lem:skew-symmetric-operator} and \cref{operator_L_models}, all the terms in \cref{def_L} except $\hat \rho \e(\theta) $ (contained in ${\cal L}_i$) are either self-adjoint or skew-symmetric. Therefore, in order to obtain $L^S_i$ we only need to compute the adjoint of this term. 
\begin{lemma}
	The adjoint of the operator $f \mapsto \nabla \cdot (\rho \e(\theta) )$ is $ f \mapsto -\nabla \cdot \p$.
\end{lemma}
\begin{proof}
	\begin{align*}
\int_\Upsilon \nabla\cdot (   \rho \e(\theta) ) h ~\ud \x \ud \theta
&= \int_\Upsilon \nabla \cdot \left[  \int_0^{2\pi} f(\x,\theta') \ud\theta' \e(\theta) \right] h ~\ud \x \ud\theta \\
&= - \int_\Upsilon \left[\int_0^{2\pi}  \nabla \cdot (h \e(\theta) ) \ud \theta\right] f(\x,\theta')  ~\ud \x \ud \theta'.
\end{align*}
\end{proof}

Thus, we arrive at the following eigenvalue problem for $i = 2, 3, 4$
\begin{equation}\label{op_comptact}
	L_{i}^S \hat f = \nabla \cdot \left[ \alpha_i \nabla \hat f + \frac{\beta_i}{2 \pi} ( \hat \rho \e(\theta) - \hat \p ) + \frac{\gamma_i}{2 \pi} \nabla \hat \rho \right] + \partial_\theta^2 \hat f = \lambda \hat f,
\end{equation}
where $\alpha_i$ are as in \cref{general_form_stability} and $\beta_2 = \beta_4 = \phi \Pe/2$, $\beta_3 = \phi \Pe$, $\gamma_2 = 0, \gamma_3 = 3\phi$, and $\gamma_4 = \phi$.
In order to proceed, we first ignore the term $\Delta \hat f$ in \cref{op_comptact} and consider the auxiliary eigenvalue problem
\begin{equation}
	\label{eig_reduced}
	 \nabla \cdot \left[\beta_i ( \hat \rho \e(\theta) - \hat \p )/(2 \pi) + \gamma_i \nabla \hat \rho/(2 \pi) \right] + \partial_\theta^2 \hat f = \Lambda \hat f.
\end{equation}
We look for a solution of the form $\hat f(\x,\theta) = a(\x) + \bfb(\x)\cdot \e(\theta) $. Inserting this in \cref{eig_reduced} we find
$$
\nabla \cdot [\beta_i (a \e(\theta)  -  \bfb/2) +  \gamma_i \nabla a] - \bfb \cdot \e(\theta) = \Lambda (a + \bfb \cdot \e(\theta) ),
$$
which implies
$$
\gamma_i \Delta a -  \frac{\beta_i}{2} \nabla \cdot \bfb = \Lambda a,\qquad \beta_i  \nabla a  = (1+\Lambda) \bfb.
$$
Therefore
\begin{equation} \label{Lambda_def}
\Delta a =  \mu a \qquad \text{with} \qquad \mu := \frac{2 \Lambda(1+ \Lambda)}{2(1+\Lambda)\gamma_i -  \beta_i^2}.	
\end{equation}
Since $\bfb$ is proportional to $\nabla a$, it is an eigenfunction of the Laplacian as well, so overall
$ \Delta \hat f =  \mu \hat f.$
Imposing periodicity on $\Omega = [0,1]^2$ (it is already periodic in angle by construction) we have that $\mu = -4\pi^2 (m^2 + n^2)$ for $m,n \in \mathbb N$.

Finally, we take $\hat f$ such that it is an eigenfunction of \cref{eig_reduced} and consider \cref{op_comptact}.
Then it follows that
\begin{equation*}
	L_i^S \hat f = \alpha_i \mu \hat f + \Lambda \hat f,
\end{equation*}
so the eigenvalue of \cref{op_comptact} is given by
$	\lambda =  \alpha_i \mu + \Lambda.$
Imposing $\lambda >0$ 
leads to the condition
$$
\frac{\beta_i^2}{2} > (\alpha_i + \gamma_i)(1 -  \mu \alpha_i) \ge (\alpha_i + \gamma_i)(1 +  4\pi^2 \alpha_i),
$$
where in the last inequality we have already discarded the case $\mu = 0$ as this gives $\Lambda = 0, -1$ (therefore leading to $\lambda \le 0$). Substituting in the values of the constants $\alpha_i, \beta_i, \gamma_i$, we arrive at $\lambda>0$ iff $\Pe > \Pe_i(\phi)$ for $i = 2, 3, 4$ with
\begin{align} \label{boundary}
	\begin{aligned}
		\Pe_2(\phi) &= 2 (1-\phi) \sqrt{2 + 8\pi^2 (1-\phi)^2  }/\phi,\\
		\Pe_3(\phi) &= \sqrt{2 (1+2 \phi)(1+4\pi^2(1-\phi))}/\phi,\\
		\Pe_4(\phi) &= 2 \sqrt{2 + 8\pi^2 (1-\phi)}/\phi.
	\end{aligned}
\end{align}
In \cref{fig:dispersion} we plot the three curves for the onset of linear instability of the symmetrized operator for each of the three models. For comparison, we also plot the boundary of stability for the symmetric part of the one-dimensional Model 2 (see \cref{sec:1dinstability}).
\def \scl {1.0}
\begin{figure}[htb]
\begin{center}
\psfrag{Pe}[][][\scl][-90]{$\Pe$}
\psfrag{phi}[][][\scl]{$\phi$}
\psfrag{M2}[l][][\scl]{Model 2}
\psfrag{M3}[l][][\scl]{Model 3}
\psfrag{M4}[l][][\scl]{Model 4}
\psfrag{M21}[l][][\scl]{Model 2 (1D)}
\includegraphics[width = 0.5\textwidth]{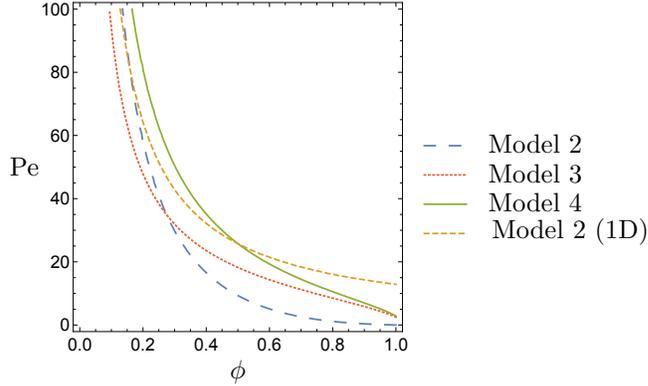}
\end{center}
  \caption{Curves $\Pe(\phi)$ in \cref{boundary} indicate the boundary between linear stability and instability of the symmetrised operators $L_i^S, i = 2, 3, 4$, given by \cref{boundary}. We use $\Dphi = (1-\phi)^2$ in Model 2. The area above the curves corresponds to the operator $L_i^S$ having a positive eigenvalue. The curve $\Pe(\phi)$ corresponding to the one-dimensional version of Model 2 is given in \cref{eig_1dsym}.
  }
 \label{fig:dispersion}
\end{figure}

\cref{fig:dispersion} gives regions in $\Pe-\phi$ space, below the curves, for which the homogeneous stationary state of Models 2, 3, and 4 is stable, given the stability of the associated symmetric problems, and therefore no phase separation can occur. In this section, we use numerical simulations to establish the converse, that is, the regions where unstable modes leading to motility-induced phase separation exist.

\subsubsection{Numerical scheme for the macroscopic PDE models}

In order to investigate the instability of the models by numerical simulations, we consider a first-order finite-volume scheme for the two-dimensional models \cref{model2_c}, \cref{model3_c}, and \cref{model4_c} based on \cite{Carrillo:2017uq, Schmidtchen:2020wy}. We note that these are integro-differential equations due to the terms involving $\rho$ and $\p$. 

For the numerical scheme, it is convenient to rewrite the models in the following form
\begin{equation}
	\partial_t f = - \nabla_\xi \cdot \left (  M {\bf U} \right), 
\end{equation}
where $ M$ is a $3\times 3$ mobility matrix and ${\bf U} = (U^x, U^y, U^\theta)$ is the vector of velocities in each of the coordinate directions. 
We note that equations \cref{model2_c}-\cref{model4_c} are not of gradient-flow form with respect to the usual $2$-Wasserstein distance and therefore $\bf U$ cannot be written as the gradient in $\xi$ of a potential function. This is due to the self-propulsion term $\e(\theta) $ being $\theta$-dependent. However, as in \cite{Kruk.2021} we still make use of the gradient-flow structure of the angular subflow (and the spatial subflow in the case of Model 4) and employ structure-preserving schemes developed for gradient flow structures in our numerical scheme. In particular, the angular velocity in all cases is $U^\theta = - \partial_\theta \log f$. For Model 2 \cref{model2_c} we have $M = f I_3$ and  the spatial velocity
\begin{equation}
	\label{spatial_flux_model2}
(U^x, U^y) = \Pe (1-\rho) \e(\theta) - \Dphi \nabla \log f. 
\end{equation}
For Model 3 \cref{model3_c} we choose $M = \text{diag}(f(1-\rho), f(1-\rho), f)$ and spatial velocity
$$
(U^x, U^y) = \Pe \left[ \e(\theta) + \frac{\p}{1-\rho} \right] -  \nabla \left[ \log f - 3 \log(1-\rho)\right]. 
$$
Finally, for  Model 4 \cref{model4_c} we use again $M = \text{diag}(f(1-\rho), f(1-\rho), f)$, and a spatial velocity given by a gradient 
$$
(U^x, U^y) = \nabla \left[ v_0 E_\theta - \log f + \log(1-\rho)\right] = -\nabla \zeta,
$$
where $E_\theta = \cos\theta x + \sin \theta y$ and $\zeta$ is the velocity potential. Thus we rewrite \cref{model4_c} as $\partial_t f =  \nabla \cdot [f(1-\rho) \nabla \zeta] + \partial_\theta( f \partial_\theta \log f)$, and use the finite-volume scheme for gradient flows of \cite{Carrillo:2017uq} in each of the two subflows of our equation.

We discretise the phase space $\Upsilon = [0, 1]^2 \times [0, 2\pi]$ into $N_x \times N_y \times N_\theta$ uniform finite volume cells $C_{i,j,k}$ of volume $\Delta x \Delta y \Delta \theta$, where $\Delta x = 1/N_x, \Delta y = 1/N_y$, and $\Delta \theta = 2\pi/N_\theta$. The cell centres are $(x_i, y_j, \theta_k) = (i\Delta x, j \Delta y, k \Delta \theta)$, $i = 0, \dots, N_x-1, j = 0, \dots N_y-1, k = 0, \dots, N_\theta - 1$. The periodic boundary conditions imply that $x_{N_x+i} = x_i$, $y_{N_y+j} = y_j$, and $\theta_{N_\theta+k} = \theta_k$. Finally, the time interval $[0,T]$ is discretised by $t_n=n\Delta t$, for $n= 0, \dots,\lceil T/\Delta t \rceil$. 

We define the cell averages
$$
f_{i,j,k}(t) = \frac{1}{\Delta x \Delta y \Delta \theta} \iiint_{C_{i,j,k}} f(x,y,\theta, t) \, \ud x \ud y \ud \theta
$$
and the finite-volume scheme
\begin{equation}\label{FV_semi}
	\frac{d}{dt}f_{i,j,k} = - \frac{F^x_{i+1/2,j,k}-F^x_{i-1/2,j,k}}{\Delta x} - \frac{F^y_{i, j+1/2,k}-F^y_{i, j-1/2,k}}{\Delta y} - \frac{F^\theta_{i,j,k+1/2}-F^\theta_{i,j,k-1/2}}{\Delta \theta},
\end{equation}
for $i = 0, \dots, N_x-1, j = 0, \dots N_y-1, k = 0, \dots, N_\theta - 1$. We approximate the flux $F^x$ 
at the cell interfaces by the numerical upwind flux
\begin{equation}
	F^x_{i+1/2,j,k} = (U_{i+1/2,j,k}^x)^+ f_{i,j,k} + (U_{i+1/2,j,k}^x)^- f_{i+1,j,k}, 
\end{equation}
using $(\cdot)^+ = \max(\cdot, 0)$ and $(\cdot)^- = \min(\cdot, 0)$, and similarly for $F^y$ and $F^\theta$. The velocities $U^x, U^y, U^\theta$ are approximated by centred differences, e.g., for the angular velocity (common in all models)
$$
U_{i,j,k+1/2}^\theta = -\frac{\log f_{i,j,k+1} - \log f_{i,j,k}}{\Delta \theta}.
$$
In the velocities of Models 2 and 3 that are not of gradient type, we use a first-order interpolation of the densities to evaluate them at the cell interfaces. For example, the $x$-velocity in \cref{spatial_flux_model2} is approximated as
$$
U_{i+1/2,j,k}^x = \Pe \, \e(k \Delta \theta) \left( 1- \frac{\rho_{i,j} + \rho_{i+1,j}}{2}\right) - \Dphi \frac{\log f_{i+1,j,k} - \log f_{i,j,k}}{\Delta x},
$$
where $\rho_{i,j} = \Delta \theta \sum_{k=0}^{N_\theta-1} f_{i,j,k}$.
Finally, we discretise the system of ODEs \cref{FV_semi} by the forward Euler method with an adaptive time-stepping satisfying the CFL condition \cite{Kruk.2021}
\begin{equation}\label{timestep-constraint}
	\Delta t \le \min \left \{ \frac{\Delta x}{6a}, \frac{\Delta y}{6b}, \frac{\Delta \theta}{6c} \right \}
\end{equation}
with $a = \max |U_{i+1/2,j,k}^x|$, $b = \max |U_{i,j+1/2,k}^y|$, and $c = \max |U_{i,j,k+1/2}^\theta|$. In \cite{Kruk.2021} they derive this CFL condition for a nonlocal model of active particles, and show it leads to a positivity-preserving numerical scheme. A key difference is that their scheme is second-order in phase space as they use a positivity-preserving piecewise linear reconstruction of the density at the interfaces, rather than using the values at the centre of the cells of either side as we do here. However, in our numerical tests (in which we use a maximum timestep of $\Delta t = 10^{-5}$) we observe \cref{timestep-constraint} to also be sufficient to preserve positivity in our scheme provided that the mobility matrix $M$ remains positive semi-definite. We found this to be always true for Models 2 and 4, whereas for Model 3 we modify the mobility in the numerical scheme to be $\tilde M_{ij} := \max (M_{ij}, 0)$. This is due to the equation for $\rho$ in Model 3 not having a maximum principle. Indeed, testing \cref{model3_c} with the negative part of either $\rho$ or $1 - \rho$ does not readily yield sign preservation. In contrast, equations \cref{model2_c} and \cref{model4_c} do readily give that $\rho \le 1$ and $f\ge 0$. 

\subsubsection{Numerical study of the instability of full operator} \label{sec:num_instability}

We solve the macroscopic models \cref{model2_c}-\cref{model4_c} using the finite-volume scheme described above with an initial perturbation around the homogeneous stationary state. In particular, we choose the first eigenfunction of the corresponding symmetric problem as the initial perturbation (as it is the worst case), which from \cref{sec:lin_inst_symm} is given by $\tilde f_0 := a(\x) + \bfb(\x)\cdot \e(\theta) $ with $a(\x) \equiv a(x) = \cos(2\pi x)$ and $\bfb(\x) \equiv (b(x), 0)$, $b(x) = \beta_i a'(x)/(1+ \Lambda)$ and $\Lambda$ given from \cref{Lambda_def} with $\mu = - 4 \pi^2$. The initial condition is then
$f_0(\xi) = f_* + \delta \tilde f_0,
$
with $\delta = 0.01$. We solve the three models with this initial condition (note that $\tilde f_0$ changes for each model through the coefficients $\beta_i$ and $\Lambda$) and final time $T = 0.2$, and for multiple combinations of the two remaining dimensionless parameters $\Pe$ and $\phi$. We use the following norm to study the growth of the perturbation over time and determine the stability of the homogeneous state:
\begin{equation} \label{perturbation_norm}
	e(t) := \Delta x \Delta y \Delta \theta \left(\sum_{i=0}^{N_x-1} \sum_{j=0}^{N_y-1} \sum_{k=0}^{N_\theta-1} |f_{i,j,k}(t) - f_*|^2 \right)^{1/2}.
\end{equation}

\cref{fig:perturbation} shows three sample outputs of the perturbation norm with three different scenarios for Model 4. The left panel corresponds to a set of parameters for which $f_*$ is linearly or exponentially stable since $e'(0) < 0$. In contrast, the right-most panel shows an example of instability since the perturbation is growing over time (and eventually converging to a non-trivial state, although we do not cover the long-time dynamics in this study). The middle panel shows an in-between case whereby the homogeneous state $f_*$ is linearly unstable (since $e'(0)>0$, so the perturbation initially grows in time), but eventually, the perturbation norm starts to decay, and the system goes back towards $f_*$. In the latter case, $f_*$ is Lyapunov asymptotically stable (for every $\epsilon > 0$, we can find an initial perturbation norm $e(0) = \delta$ such that $e(t) < \epsilon$ for all times $t>0$, and $\lim_{t\to \infty} e(t) = 0$). 

\def \scl {1.0}
\begin{figure}[htb]
\begin{center}
\psfrag{f}[][][\scl][-90]{$e$}
\psfrag{t}[][][\scl]{$t$}
\includegraphics[width = 0.32\textwidth]{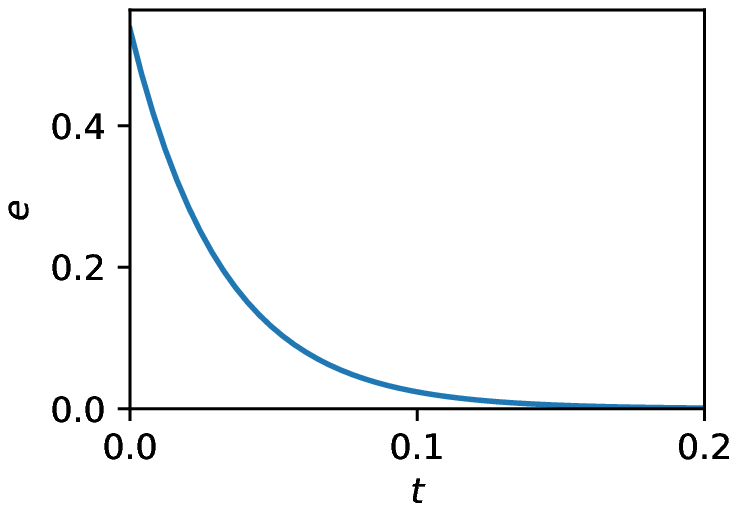}
\includegraphics[width = 0.32\textwidth]{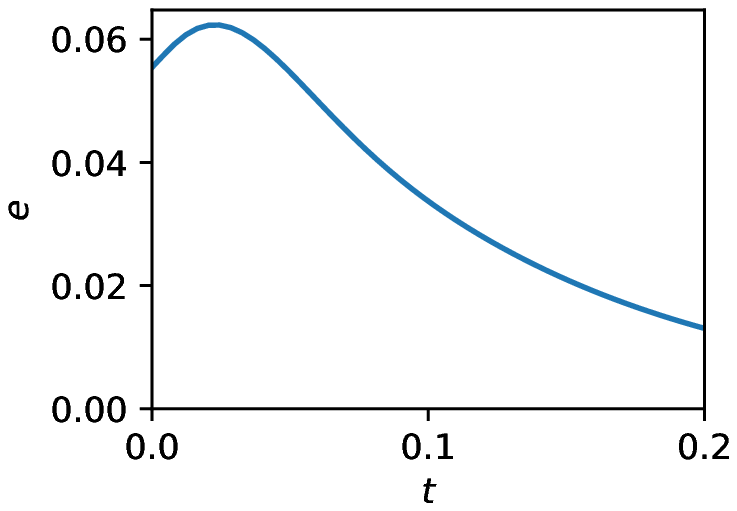}
\psfrag{f}[][][\scl][-90]{$e$}
\psfrag{t}[][][\scl]{$t$}
\includegraphics[width = 0.32\textwidth]{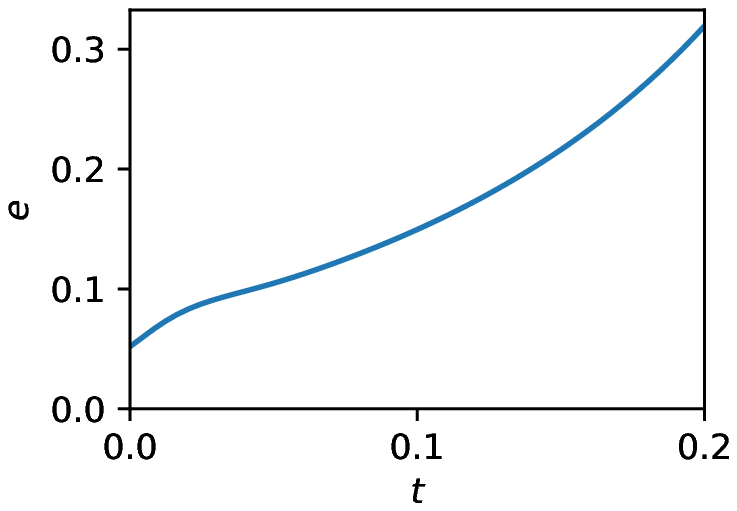}
\end{center}
\caption{Evolution of the perturbation norm $e(t)$  \cref{perturbation_norm} divided by $\Delta x \Delta y \Delta \theta$ as a function of time corresponding to Model 4 \cref{model4_c} with initial condition $f_0(\xi) = f_* + \delta \tilde f_0$. (a) Linearly stable case with $\phi = 0.25$ and $\Pe = 1$. (b) Lyapunov stable case with $\phi = 0.6$ and $\Pe = 30$. (c) Unstable case with $\phi = 0.6$ and $\Pe = 50$.}
 \label{fig:perturbation}
\end{figure}

We next do a swap in parameter space for $\phi = 0, 0.05, \dots, 0.95$ and $\Pe = 1, 5, 10, \dots, 100$, and classify each case as
\begin{itemize}
	\item linearly stable: if $e(\Delta t) - e(0) < 0$, 
	\item Lyapunov stable: if $e(\Delta t) - e(0) > 0$ but $e(T) - e(0) < 0$,
	\item unstable: if $e(T) - e(0) > 0$.
\end{itemize}
We plot the result for each of the three models in \cref{fig:fulldispersion}. For reference, for each case, we also plot the dispersion relation corresponding to the symmetric operator $L_i^S$ (solid black curves) and observe a good consistency; that is, below the curves, there are only stable points as per \cref{rem:antisymmetric} (the only exception being two green dots in Model 3 for which $e'(0) \gtrsim 0$ but within the numerical scheme set tolerance).

\begin{figure}[htb]
\begin{center}
\includegraphics[width = 0.32\textwidth]{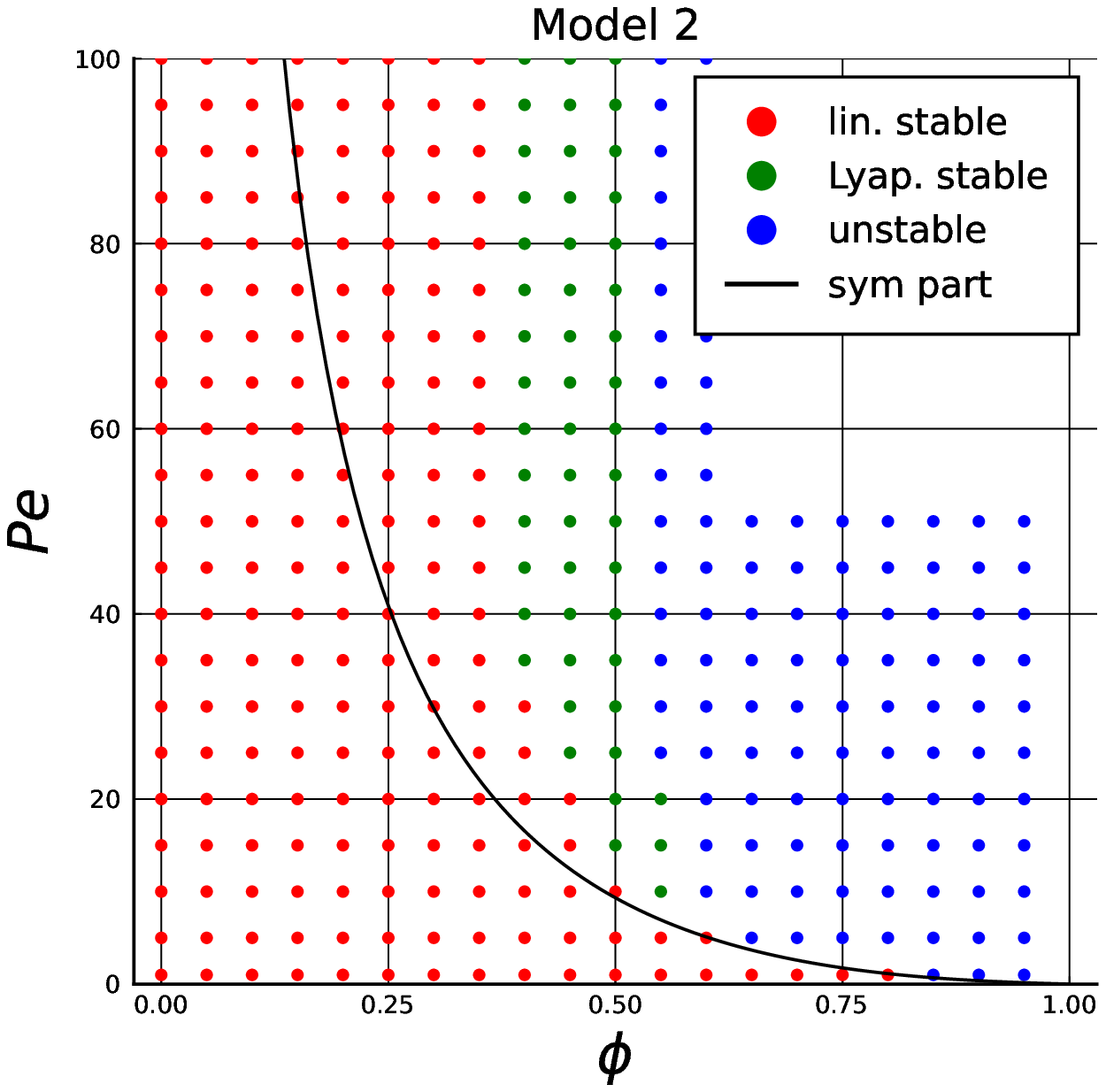}
\includegraphics[width = 0.32\textwidth]{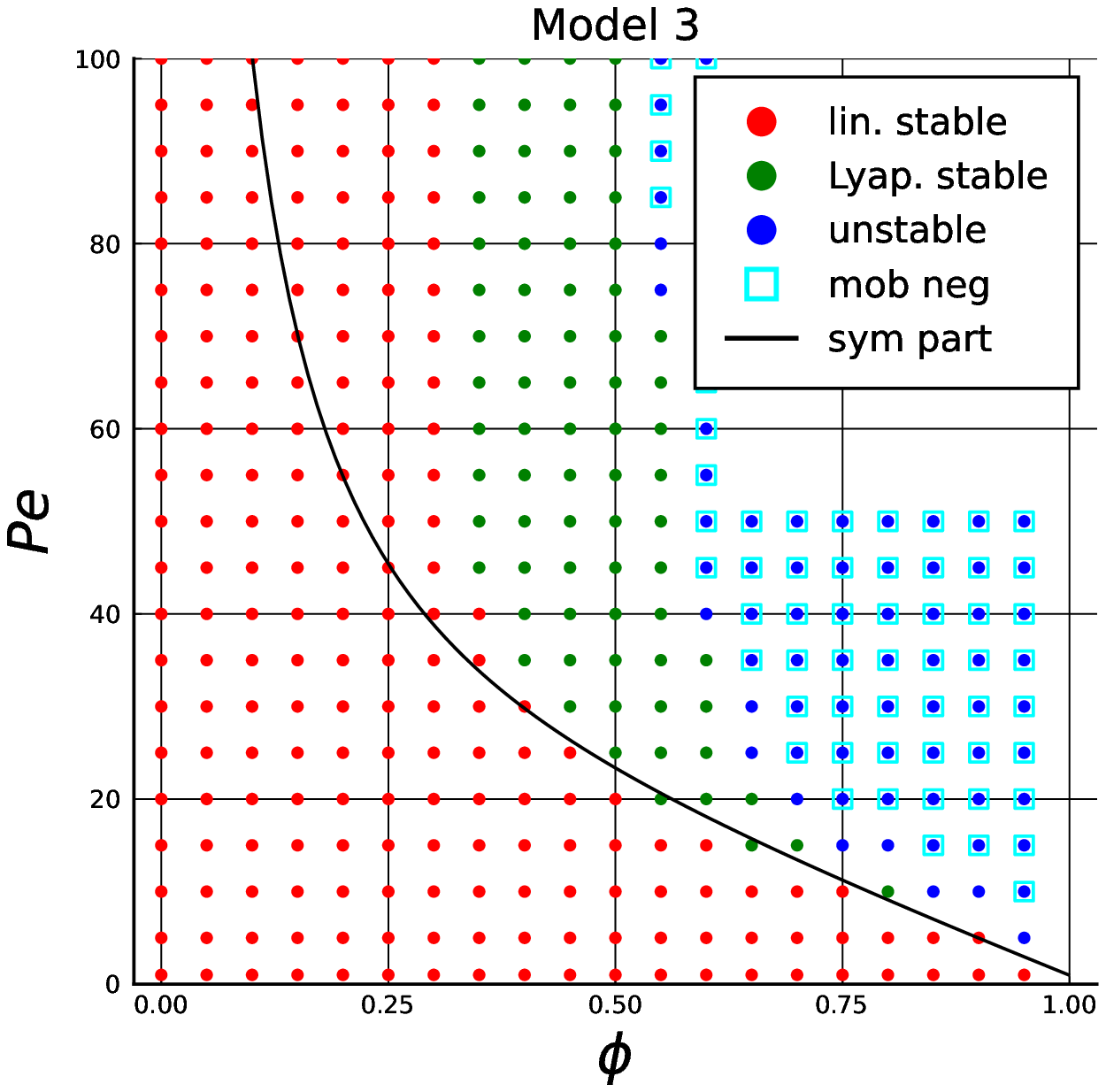}
\includegraphics[width = 0.32\textwidth]{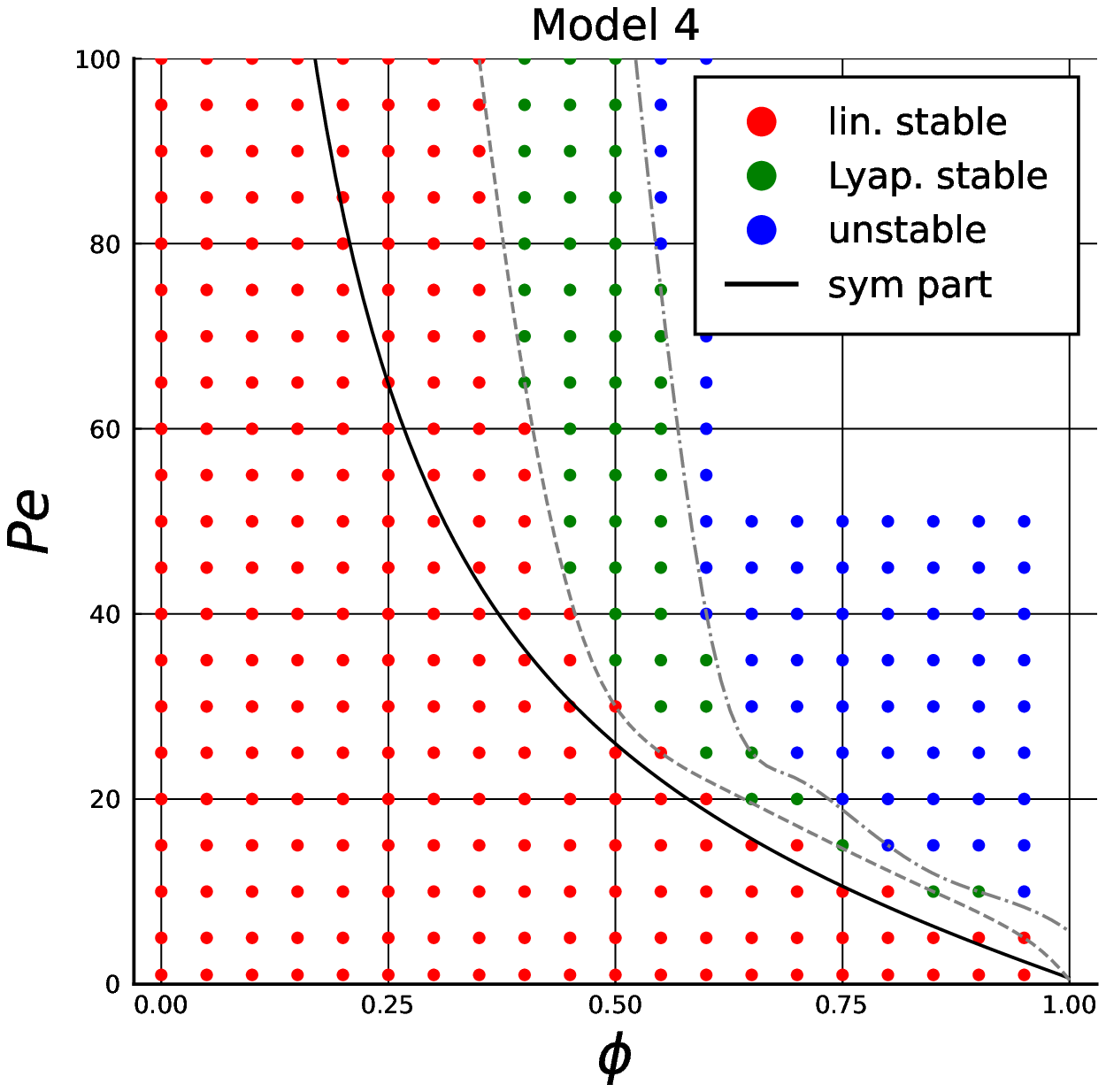}
\end{center}
  \caption{Stability of the full operator obtained via numerical simulations. Solid curves indicate boundary between linear stability and instability of the symmetrised operators $L_i^S, i = 2, 3, 4$, given by \cref{boundary} (same curves as in \cref{fig:dispersion}). In the right plot (Model 4), the grey dashed and dot-dashed lines indicate the boundaries between the three regions of stability.}
 \label{fig:fulldispersion}
\end{figure}

\section{Numerical examples} \label{sec:numerics}

This section shows numerical examples of phase separation in Models 2 to 4. In particular, using the phase diagrams in \cref{fig:fulldispersion}, we choose combinations of $\phi$ and $\Pe$ for which we know the homogeneous equilibrium is unstable and look at the convergence to non-homogeneous phase-segregated states starting from a random perturbation.

\subsection{Macroscopic phase separation}
\cref{fig:model4_2dunstable} shows an example of instability in Model 4 starting from a two-dimensional perturbation in space (a multiple of the second eigenfunction of the associated symmetric problem, corresponding to $\mu = -8 \pi ^2 $ in \cref{Lambda_def}, first column $t=0$). We observe how first the original perturbation gets amplified, so we see a two-dimensional pattern (second column $t= 0.2$ in  \cref{fig:model4_2dunstable}). However, as time progresses the symmetry is lost and the pattern evolves into a one-dimensional pattern (in space) (third and forth columns in \cref{fig:model4_2dunstable}). We also show the \emph{mean angular direction} 
\begin{equation}
\label{mean_direction}
\q (\x, t) := \frac{\p(\x,t)}{\rho(\x,t)}.	
\end{equation}
The arrows in the stream plots represent the direction of $\q$, while the color indicates the magnitude $|\q|$. We note that $|\q| \le 1$, with $|\q|$ higher the more ordered the system is ($|\q| = 0$ for randomly aligned particles and $|\q| = 1$ for perfectly aligned particles). In the middle row of \cref{fig:model4_2dunstable}, we observe that the alignment is maximal in the boundaries between the low- and high-density regions, and that particles there point towards the high-density region. This produces a ``locking in'' mechanism by which particles in the denser region cannot escape. Finally, in the third row of \cref{fig:model4_2dunstable} we show the first component of the mobility averaged over angles, $m_{11} := \int \tilde M_{11} \ud \theta$, which for Model 4 is equal to $m_{11} = \rho(1 - \rho)$. As expected, the mobility is lowest in the dense region and highest in the interface between the dilute and dense regions. 
\begin{figure}[htb]
\begin{center}
\includegraphics[width = \textwidth]{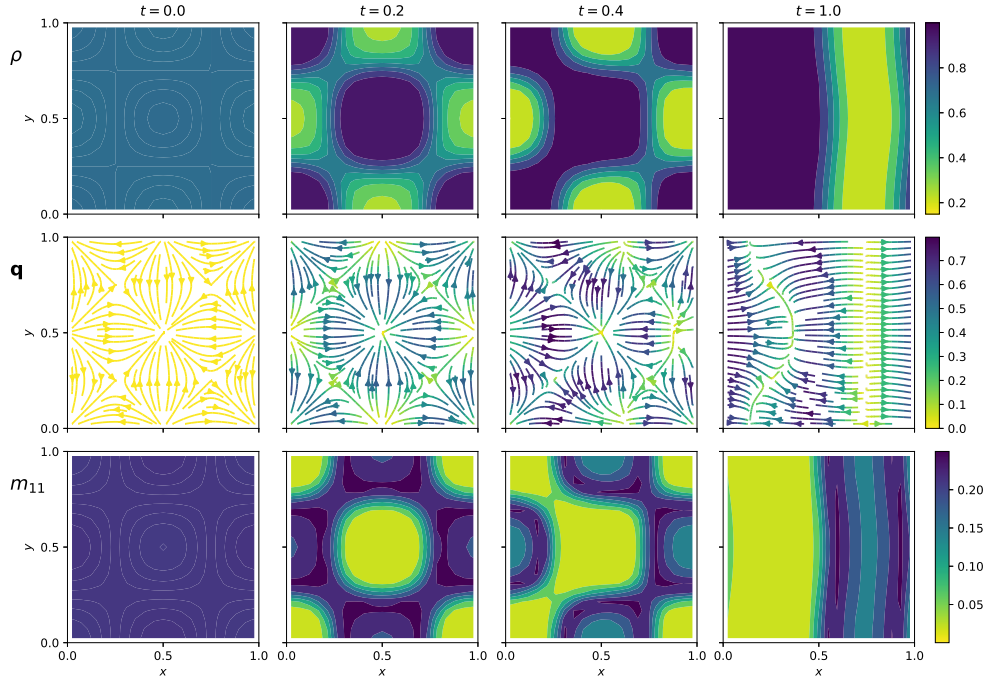}
\end{center}
  \caption{Example of a two-dimensional pattern in Model 4 \cref{model4_c} corresponding to $\phi = 0.7$ and $\Pe = 40$ starting from a perturbation symmetric in $x$ and $y$. First row corresponds to the density $\rho(\x,t)$, the second shows streamplots of the rescaled polarisation $\p(\x,t)/\rho(\x,t)$ and the last row shows the mobility $M$.}
 \label{fig:model4_2dunstable}
\end{figure}

In \cref{fig:model4_general_unstable} we consider the same model and parameters as in \cref{fig:model4_2dunstable} except that the initial condition is different: instead of taking the symmetric initial condition we consider a random perturbation of the form $\tilde f (\xi) = \sin(a_1 \pi x) \sin(a_2 \pi y) \sin (a_3 \theta/2)$ where $a_i$ are randomly picked integers between 1 and 10. We observe how the initial random orientations evolve to an ordered pattern (middle row in \cref{fig:model4_general_unstable}), showing again the locking in mechanism towards high-dense regions. In this instance, the pattern at $t=1$ is the complement of a circular blob, which stays empty (top row, right-most plot in \cref{fig:model4_general_unstable}). 
\begin{figure}[htb]
\begin{center}
\includegraphics[width = \textwidth]{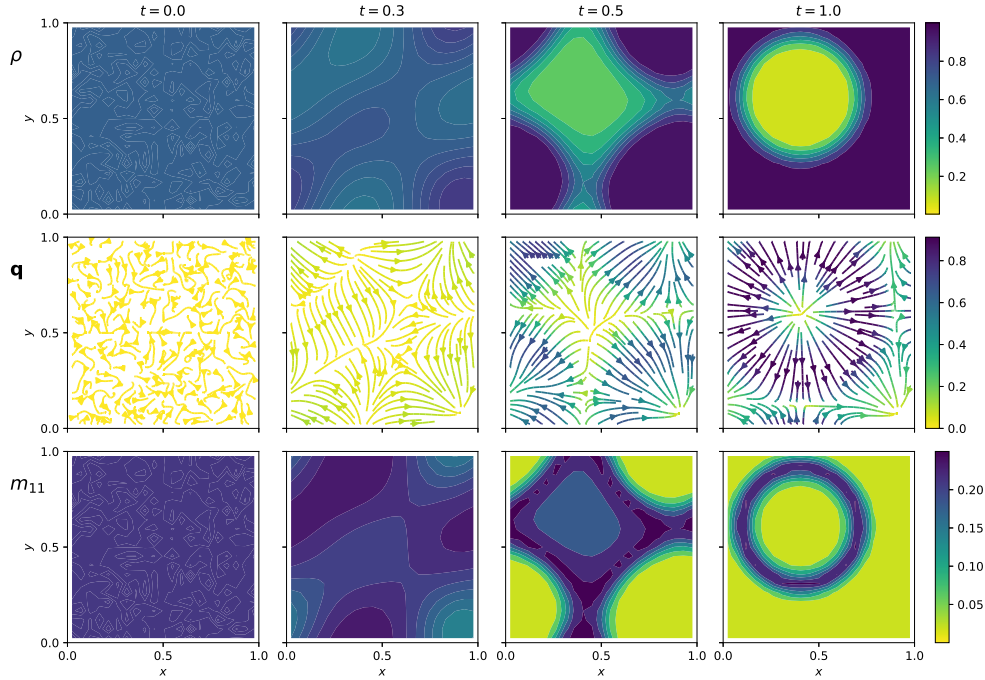}
\end{center}
  \caption{Two-dimensional pattern in Model 4 \cref{model4_c} corresponding to $\phi = 0.7$ and $\Pe = 40$ starting from a random initial perturbation.}
 \label{fig:model4_general_unstable}
\end{figure}

We observe similar patterns in Models 2 and 3, with Model 2 showing patterns for lower values of $\phi$ and $\Pe$ as expected from the dispersion diagram (see \cref{fig:fulldispersion}). We plot examples of such patterns for Models 2 and 3 in \cref{fig:model2_unstable,fig:model3_unstable} respectively. We show more of such examples in \cref{sec:A_PDEsims}.  
\begin{figure}[htb]
\begin{center}
\includegraphics[width = \textwidth]{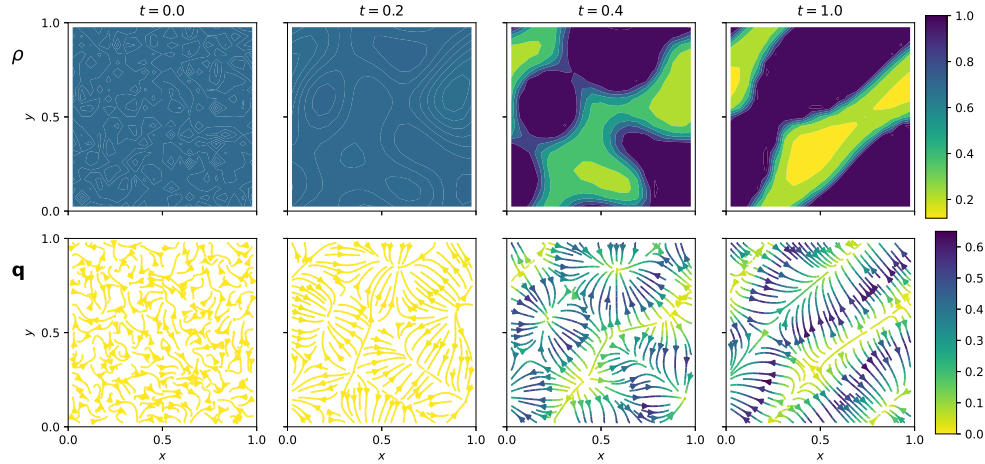}
\end{center}
  \caption{Time-evolution of Model 2 \cref{model2_c} with $\phi = 0.7$, $\Pe = 20$ and final time $T = 1.0$, starting from a random initial perturbation.}
 \label{fig:model2_unstable}
\end{figure}

\begin{figure}[htb]
\begin{center}
\includegraphics[width = \textwidth]{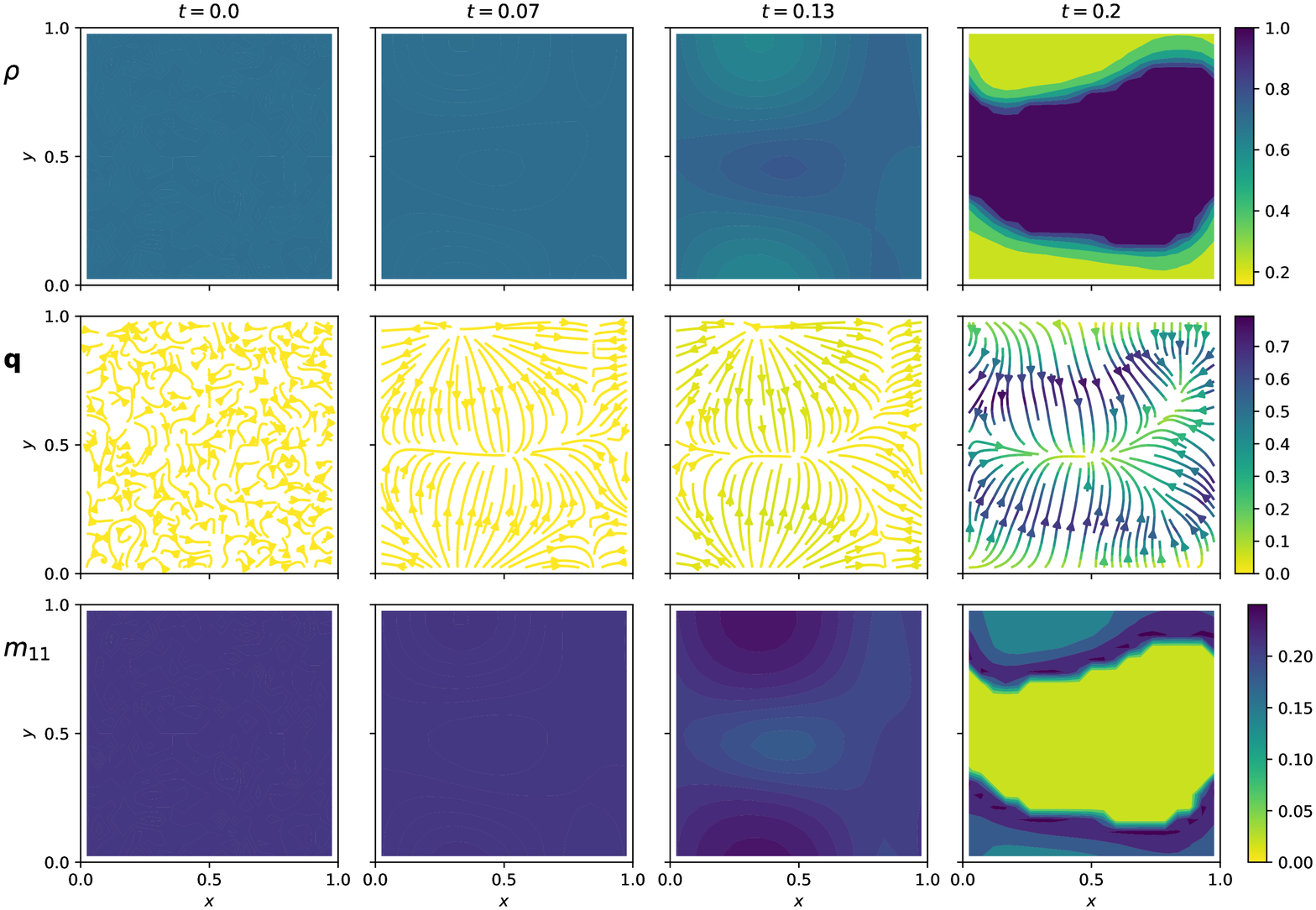}
\end{center}
  \caption{Time-evolution of Model 3 \cref{model3_c} with $\phi = 0.7$, $\Pe = 40$ and final time $T = 0.2$, starting from a random initial perturbation.}
 \label{fig:model3_unstable}
\end{figure}

\subsection{Comparison with the microscopic system}

In this section, we compare the results obtained from the PDE analysis with stochastic simulations of the microscopic models. We focus on Model 4, whose microscopic dynamics correspond to a discrete jump process in position and a continuous Brownian motion in angle (see \cref{sec:model3}). 

We use $N = 500$ particles in a square periodic lattice with spacing $\epsilon$ such that $N \epsilon^2 = \phi$ for every given $\phi$. Simulations are performed using the agent-based modelling package Agents.jl \cite{Agents.jl} in Julia. The particles are initialized at $t=0$ uniformly distributed in angle and the lattice (while satisfying the constraint of only one particle per lattice site). Then the system is evolved using a fixed timestep $\Delta t = 10^{-4}$. \cref{fig:model4_ABM} shows two examples of the system at $T=1$ for different $(\phi, \Pe)$ pairs. The particles are represented by triangles of length $\epsilon$ pointing towards $\e(\Theta_i)$. MIPS is clearly present in both cases, with two distinct regions: a dilute one with almost zero density and a dense region with particles in close packing. We observe a strong particle alignment in the boundary between two regions, with the polarisation $\p$ pointing towards the dense region. We also show the absolute value of the mean orientation $\q$ in \cref{mean_direction} in each lattice site as a colormap. Using this, it is easy to see the random orientation in the bulk of the dense region. More examples showing MIPS and its emergence for increasing $\phi$ or $\Pe$ are  shown in \cref{fig:ABM4_Tf=1.0_v0=100.0,fig:ABM4_Tf=1.0_phi=0.6}.
\begin{figure}
\begin{minipage}[c]{0.45\linewidth}
\centering	
	\includegraphics[height = \textwidth]{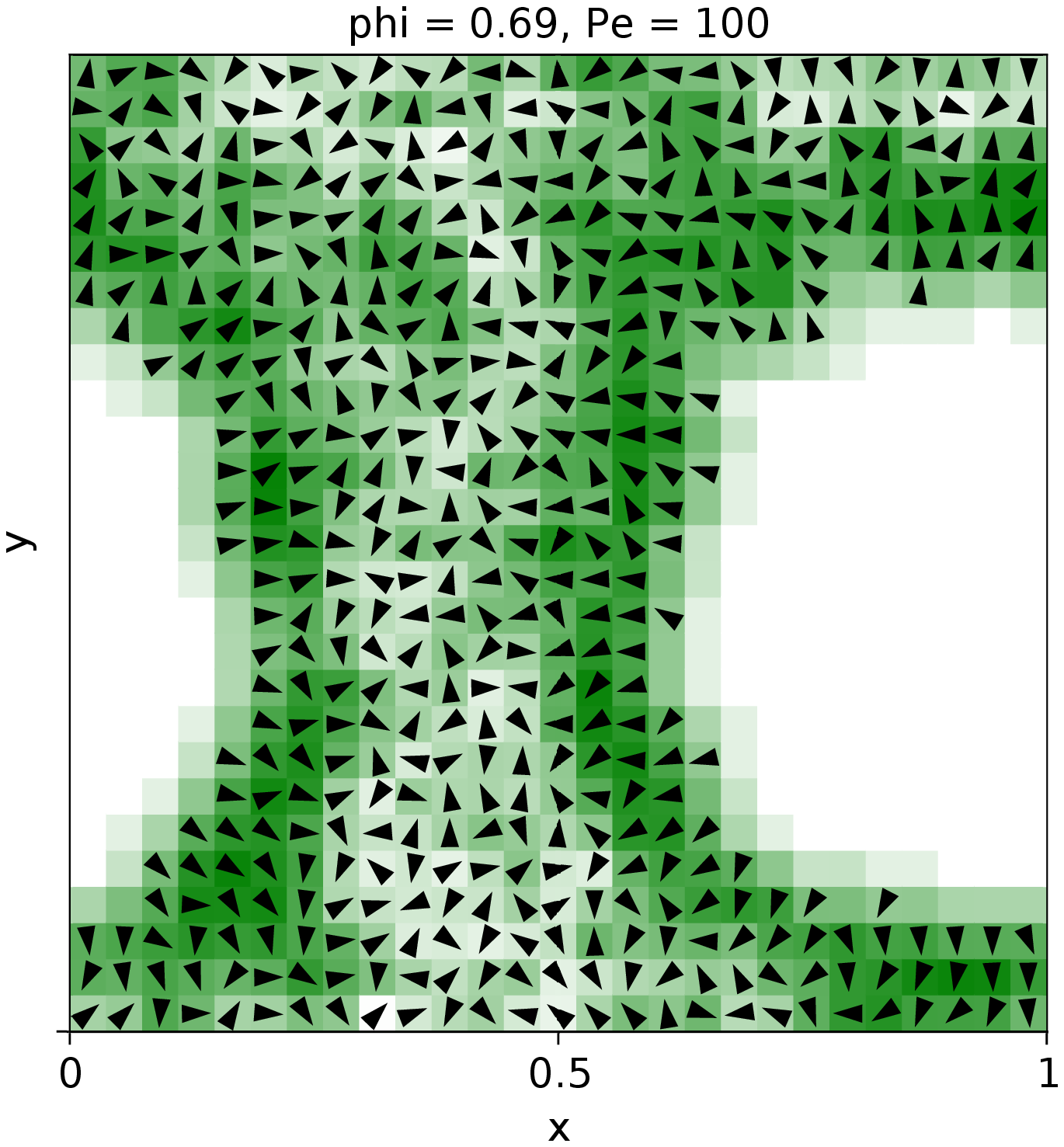}
\end{minipage}
\quad
\begin{minipage}[c]{0.475\linewidth}
\centering	
	\includegraphics[height = \textwidth]{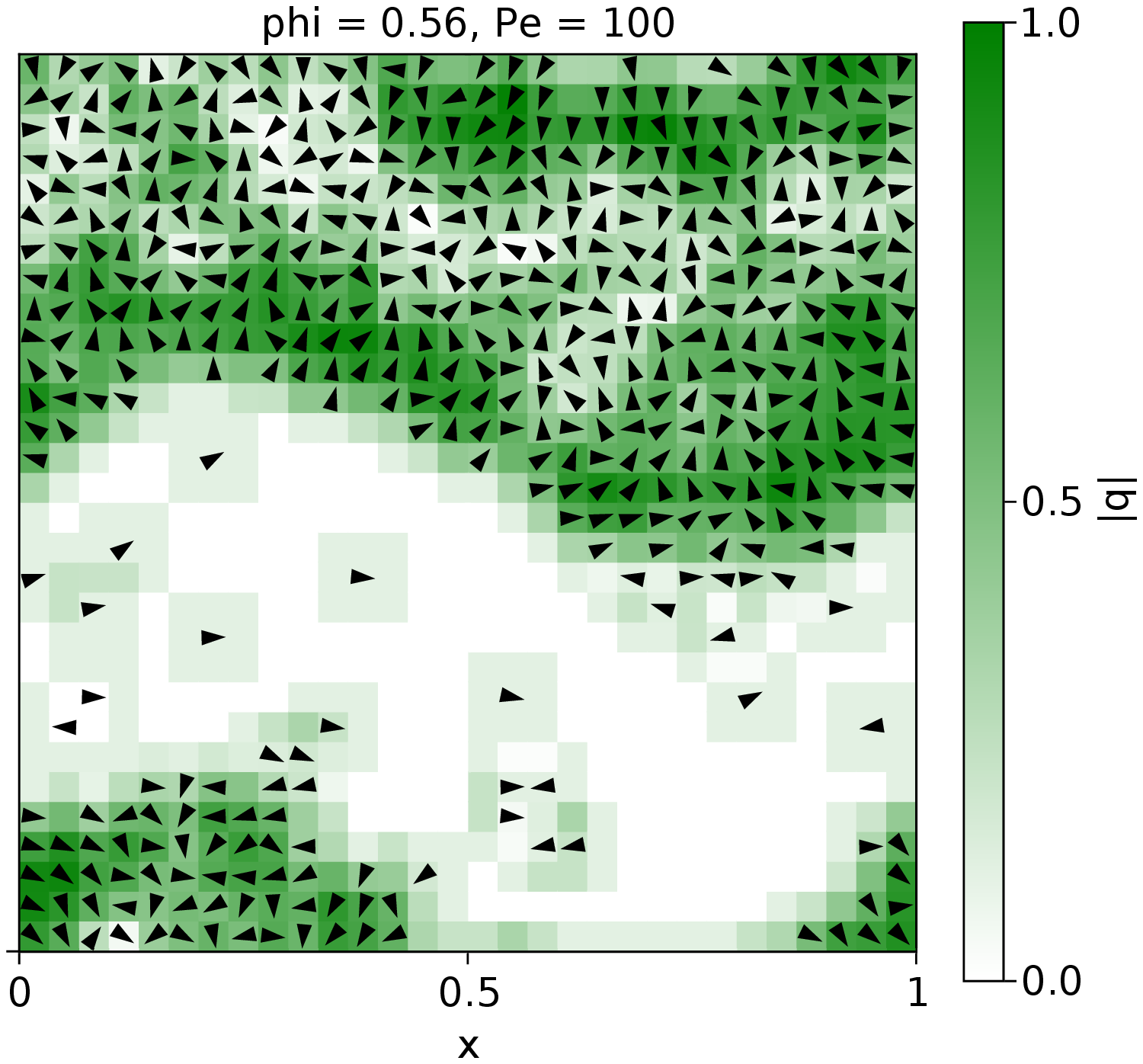}
\end{minipage}
	\caption{Snapshot of an output of Model 4 at $T = 1$ for $N = 500$ particles. (Left) $\phi = 0.69$ and $\Pe =  100$. (Right) $\phi = 0.56$ and $\Pe = 100$. Each triangle represents one particle, pointing in the direction of its orientation $\Theta_i$. The heatmap corresponds to the absolute value of mean-direction $\q$ in \cref{mean_direction}, computed using a Moore neighbourhood at each lattice point.}
 \label{fig:model4_ABM}
\end{figure}

In order to quantify the presence of MIPS in the stochastic model, we consider $P_9(\phi,\Pe)$ as the proportion of sites in the lattice whose Moore neighborhood (the central cell plus the eight cells that surround it) is full.\footnote{We note that in a finite lattice, $P_9$ also depends on the number of particles $N$, which we keep fixed at $N=500$ in all our simulations.} If the distribution of particles on the lattice is assumed to be uniform, which is the case for the unbiased symmetric simple exclusion process corresponding to $\Pe = 0$, the expected proportion $P_9(\phi, 0):= P_{9,unif}$ can be computed exactly (see \cref{eq:p9unif} and the black dashed line in \cref{fig:model4_ABM_metrics}(left)). However, we expect that for $\phi$ large enough and as $\Pe$ increases, the uniform distribution is lost in favor of clusters, leading to a sharp increase in $P_9$.
 The coloured lines in \cref{fig:model4_ABM_metrics}(left) correspond to $P_9$ measured at $T=1$ for various values of $(\phi, \Pe)$. The point of departure of each coloured line from the black line (corresponding to $P_9$ in the uniform non-segregated case) identifies the critical $\phi$ for that given $\Pe$ above which MIPS occurs. The value of $P_9$ at $T=1 $ for multiple points in space of parameters $(\phi, \Pe)$ is shown in \cref{fig:model4_ABM_metrics}(right), together with the boundaries between the three regions of stability of the PDE model \cref{model4_c} as shown in \cref{fig:fulldispersion}. We observe a good qualitative agreement between the two. The analogous plot using the difference $\Delta P_9 := P_9 - P_{9,unif}$ are shown in \cref{fig:model4_ABM_metrics_sup}. A similar cluster fraction measure was employed in \cite{Redner.2013} to establish phase separation in a model of soft Brownian active particles \cref{sde_all} with a WCA potential, resulting in a very similar phase diagram (see Fig. 3 in \cite{Redner.2013}).

\begin{figure}
\begin{center}
\begin{minipage}[c]{0.4\linewidth}
\centering
\includegraphics[width = \textwidth]{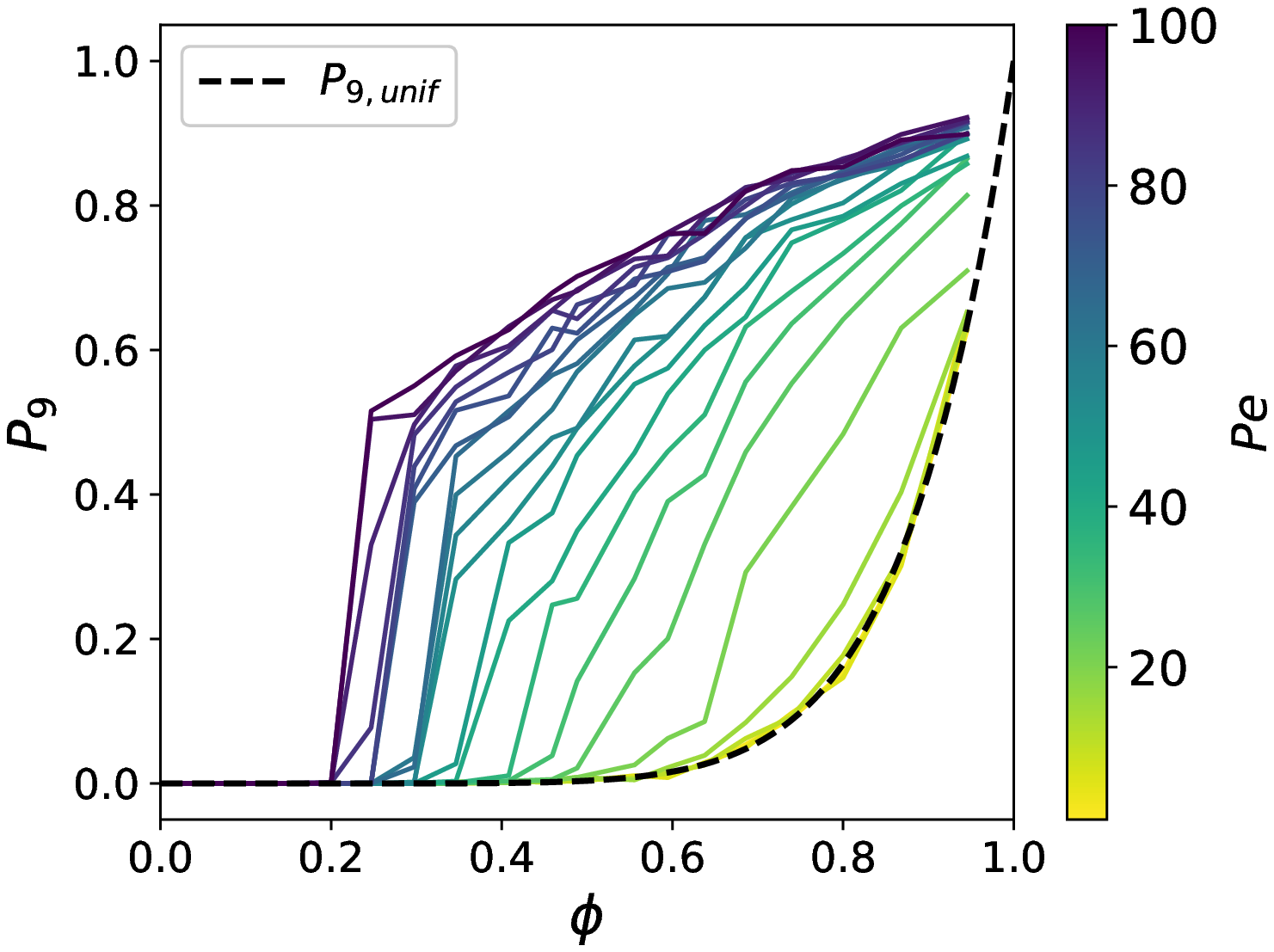}
\end{minipage}
\begin{minipage}[c]{0.4\linewidth}
\centering
\includegraphics[width =\textwidth]{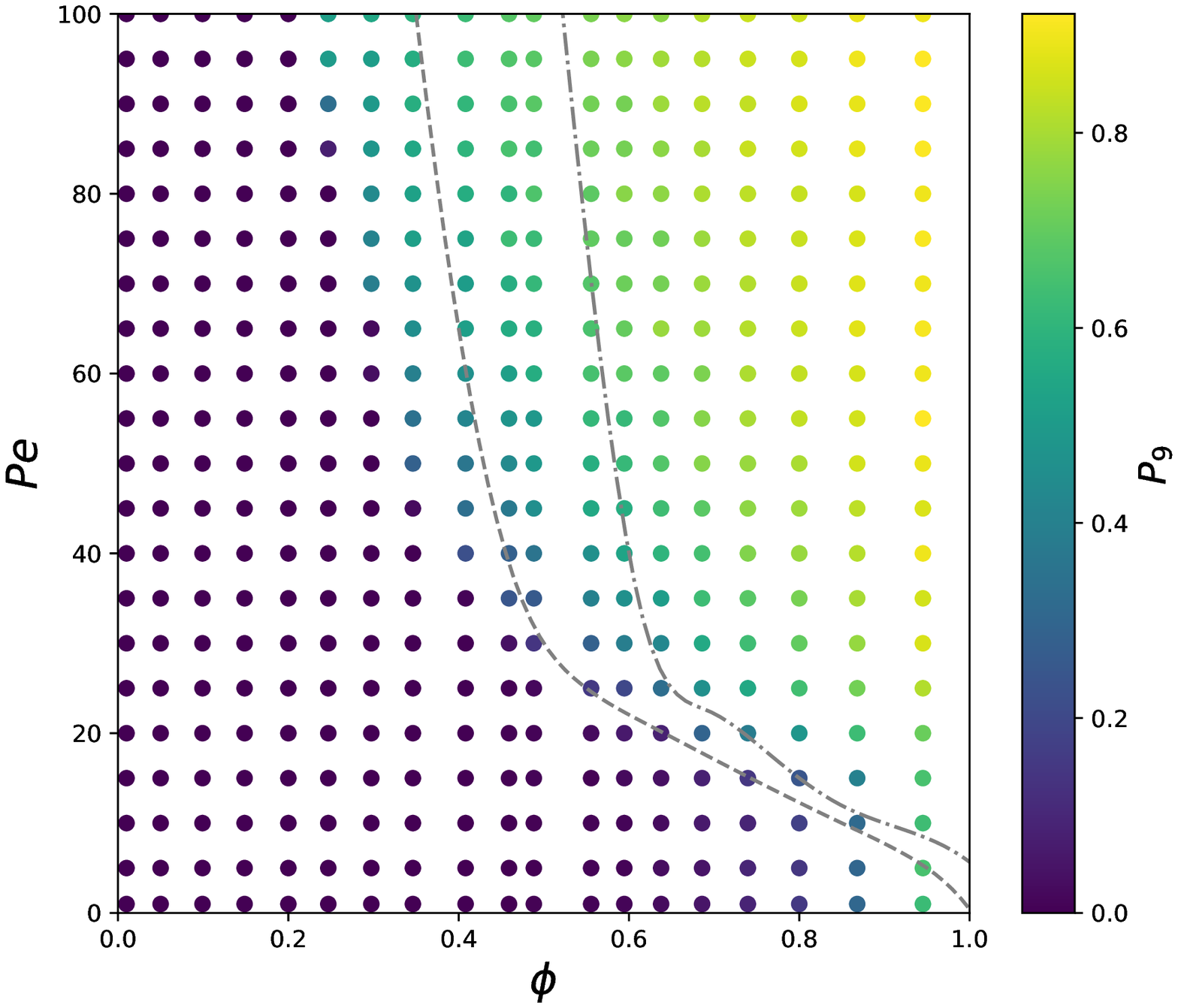}
\end{minipage}
\end{center}
\caption{Cluster fraction $P_9$ in the stochastic Model 4 for varying values of $\phi$ and $\Pe$. Left: curves $P_9(\phi)$ for various $\Pe$ computed from simulations of Model 4. The black dashed line shows $P_9(\phi)$ if particles are uniformly distributed in the lattice. Right: colormap of $P_9$. The two grey lines correspond to the boundaries between linear and Lyapunov stability and Lyapunov stability and instability from the PDE model \cref{model4_c} (as shown in \cref{fig:fulldispersion}).}
 \label{fig:model4_ABM_metrics}
\end{figure}

\section{Discussion} \label{sec:conclusion}

We have considered various macroscopic models for systems of active Brownian particles, which describe a class of self-propelled particles whereby the swimming orientation $\theta$ is governed by a Brownian motion. Despite being arguably the simplest model for active particles, the addition of purely repulsive interactions between particles can lead to striking phase separation, known as motility-induced phase separation (MIPS) \cite{Cates:2014tr}. MIPS is caused by the interplay between repulsive interactions (which are used to model size exclusion) and the active self-propulsion speed, leading to segregation into a dilute phase (where particles can swim almost freely at their desired speed) and a dense phase (where the effective speed is significantly reduced due to crowding). To investigate the nature and the strength of the interactions required to lead to MIPS, in this paper, we obtained four different macroscopic models for the density $f(\x, \theta, t)$ depending only on two non-dimensional parameters $\phi$ (the occupied volume fraction) and $\Pe$ (the P\'eclet number). These differ from slight variations of the underlying microscopic particle-based model and different coarse-graining procedures:
\begin{itemize}
	\item Model 1 \cref{model1_c}: active Brownian particles with long-range repulsive interactions. The associated PDE is obtained via a mean-field approximation.
	\item Model 2 \cref{model2_c}: active Brownian particles with short-range soft repulsive interactions. A PDE is obtained phenomenologically under close-to-equilibrium and homogeneous assumptions \cite{Bialke:2013gw}. 
	\item Model 3 \cref{model3_c}: active Brownian particles with hard-core repulsive interactions. The PDE is obtained via the method of matched asymptotic expansions in the limit of small $\phi$.
	\item Model 4 \cref{model4_c}: active asymmetric simple exclusion process (active Brownian model with the positions constrained to a lattice). The associated PDE is obtained using a mean-field approximation.
\end{itemize}
Having the four models under the same framework, we have then considered a stability analysis to establish the region in the parameter space $(\phi, \Pe)$ in which the homogeneous state $f =$ constant is unstable. We have used numerical simulations of the macroscopic models and the microscopic models to show the patterns emerging in such cases. 

We have found that the homogeneous state is always stable in Model 1, regardless of the values of $\phi$ and $\Pe$, implying that such a model cannot capture MIPS. This is to be expected that MIPS requires local interactions, whereas \cref{model1_c} has nonlocal interactions in the convolution term, and these do not change the advection speed, which remains equal to the free speed $\Pe$. In contrast, the remaining three models \cref{model2_c,model3_c,model4_c} all display MIPS for $\phi$ and $\Pe$ large enough, with qualitatively a similar phase transition boundary (see \cref{fig:dispersion}). The reason for this is an advection term of the form $\Pe \nabla \cdot [f(1-\rho) \e(\theta) ]$, where $\rho(\x,t) = \int f \ud \theta$ is the space density with mass $\phi\in [0,1)$: in dense regions, $\rho$ is large and the effective swim speed is reduced. The systematic derivation of Model 3, whose underlying microscopic model is equivalent to that of Model 2, results in additional terms missed in \cite{Bialke:2013gw, Speck:2015um}, namely nonlinear cross-diffusion terms $\rho \nabla f, f \nabla \rho$. A similar structure (but with different coefficients) is found in Model 4, corresponding to a microscopic lattice-based model. It is quite remarkable that, although the derivation of Model 3 relies on an asymptotic expansion for a small volume fraction $\phi$, the resulting PDE Model 3 can capture the phase transition boundary with values of $\phi$ above 0.5 (see \cref{fig:dispersion}). The presence of these cross-diffusion terms in Models 3 and 4, combined with the nonlinear advection term, makes their rigorous analysis very challenging. As a first step towards this goal, we have considered the well-posedness of Model 2 \cite{BruBurEspSch-an-phen21}. 

We have studied the presence of MIPS in Models 2 to 4 through an analytical linear stability analysis of the associated symmetric differential operators, as well as numerical simulations of the PDE models. In the latter, starting from a perturbation around the homogeneous state as the initial condition, we have performed relatively short time-dependent simulations to determine either the decay to the homogeneous state or the growth and transition to some other patterned state. A natural question is to ask about the nature of such patterns, that is, whether they are stable or metastable, and how many of such patterns exist. Another interesting avenue would be to study the effective collective speed of the dense phase. These problems seem highly related to the analysis in \cite{burger2008asymptotic}, where metastable patterns of a particle system under an attractive force have been analyzed asymptotically, with a similar locking-in the phenomenon of the metastable clusters. Let us also mention
\cite{dolak2005keller}, closely related to the one-dimensional crowded Goldstein--Taylor model. It thus seems reasonable to work out similar results in the asymptotic of small diffusion coefficients, additional complications being the relative scaling of spatial and angular diffusion as well as the periodic boundary conditions, which make it difficult to understand the density value inside the clusters and allow a cluster drift (while in the case of no-flux conditions as considered in  \cite{burger2008asymptotic} the clusters are always fully packed).

\section*{Acknowledgments}
This work was carried out while AE and SMS were postdoctoral researchers at FAU Erlangen-N\"{u}rnberg and the University of Cambridge, respectively. The authors would like to thank Dr. Ulrich Dobramysl (University of Oxford), Prof. Markus Schmidtchen (University of Dresden), and Prof. Joan Sol\`a-Morales (Universitat Polit\`ecnica de Catalunya) for helpful discussions. 

\bibliographystyle{siamplain}
\bibliography{active.bib}

\appendix

\section{Phenomenological models for active Brownian particles} \label{sec:model2_SM}

This section follows loosely the derivation of the macroscopic model in \cite{Bialke:2013gw, Speck:2015um}, while aiming to make explicit their assumptions and use consistent notation with the other models presented in \cref{sec:derivation}. Consider the same scenario as in \S \ref{sec:model1} but with short-range and strong repulsive interactions, namely, $\chi = 1$ and $\ell = \epsilon \ll 1$. The starting point is then \cref{1_eq}-\cref{interaction_G}, which under this subsection's assumptions read:
\begin{align} \label{1_eq_2}
\begin{aligned}
	\partial_t f(\xi_1, t) &= \nabla_{\x_1} \cdot \left[ D_T \nabla_{\x_1} f -v_0 \e(\theta_1) f + {\bf G}(\xi_1,t) \right] + D_R \partial_{\theta_1 \theta_1} f,\\
	{\bf G}(\xi_1,t) &= (N-1) \int_{\Upsilon} F_2(\xi_1,\xi_2,t) \nabla_{\x_1} u(\|\x_1-\x_2\|/\epsilon) \ud \xi_2  \\
	&= (N-1) \int_{\Omega} f_2(\xi_1,\x_2,t) \nabla_{\x_1} u(\|\x_1-\x_2\|/\epsilon) \ud \x_2,
	\end{aligned}
\end{align}
where $f_2(\xi_1,\x_2,t) := \int F_2 ~\ud \theta_2$ is the two-body probability density to find another particle at $\x_2$ (with arbitrary orientation) together with the tagged particle at $\x_1$ with orientation $\theta_1$.

In \cite{Bialke:2013gw, Speck:2015um}, the authors introduce a new coordinate system for $\x_2$ relative to $\x_1$ and $\theta_1$. In particular, let us define $\tr$ and $\tt$ such that
$
\x_2 = \x_1 +  \tr \e(\theta_1 - \tt),
$
so that $\tr$ is the distance between the particles' centres, and $\tt$ is the angle between the orientation of the first particle and the line $\x_2-\x_1$ (see Fig.~1 in \cite{Speck:2015um}; note a slight change of notation). Then they consider a decomposition of $f_2$ involving the pair correlation function $g$, which they take to be independent of $\xi_1$ and $t$ (homogeneous and stationary suspension), see \cite[Eq. $(9)$]{Bialke:2013gw}:
\begin{equation}
\label{P2_decomp}
f_2(\xi_1,\tr,\tt, t) = f_1(\xi_1,t) \rho(\x_1 + \tr \e(\theta_1 - \tt),t) g(\epsilon \tr,  \tt).
\end{equation}
With the change of variables, we have that 
$$\nabla_{\x_1} u(\|\x_1-\x_2\|/\epsilon) = -\nabla_{\tx} u(\|\tx\|/\epsilon) = -  u'(\tr/\epsilon) \e(\theta_1 - \tt).
$$
With this, the interaction term in \cref{1_eq_2} becomes 
\begin{equation}
	{\bf G}(\xi_1,t) = (N-1) f(\xi_1,t) \int_0^\infty \int_0^{2\pi} \rho(\x_1 +  \tr \e(\theta_1 - \tt),t) g(\tr, \tt) u'(\tr/\epsilon) \e(\theta_1 - \tt) \tr \ud \tr \ud \tt,
\end{equation}
The range of integration for $\tr$ is taken to $\infty$ assuming $u(r)$ decays fast enough and $\epsilon \ll 1$. 
They consider a decomposition of the force along $\e(\theta_1)$ and its perpendicular direction,
$\mathbf{G}=G_\e \e(\theta_1) +\delta \mathbf{G}$,
and a balancing of the transport and the interaction terms in the equation for $f$. In particular, inserting the decomposition of  $\mathbf{G}$ into \eqref{1_eq_2}, one obtains
\begin{equation} \label{eqf2_G}
\partial_t f(\xi_1, t) = \nabla_{\x_1} \cdot \left[ D_T \nabla_{\x_1} f + (G_\e -v_0 f)\e(\theta_1) + \delta \mathbf{G} \right] + D_R \partial_{\theta_1 \theta_1} f.
\end{equation}
In \cite{Bialke:2013gw, Speck:2015um} it is assumed that $D_T \sim \|\delta {\bf G}\| \ll G_\e \sim v_0$, such that, at leading order, we have
$\nabla_\x \cdot (\e (\theta_1) f)=0$, and therefore $\e(\theta) \cdot \nabla_\x f = 0$. Hence, in this regime, $\delta { \bf G} \sim G_{\|} \nabla_\x f$ at leading order. At leading order in $\epsilon$, the coefficients $G_\e$ and $G_{\|}$ of the $\bf G$ decomposition are, 
\begin{multline}
G_\e  =(N \! -\! 1) f(\xi_1,t)  \iint \! \rho(\x_1 + \tr \e(\theta_1 - \tt),t) g(\tr, \tt) u'(\tr/\epsilon) \cos(\tt) \tr \, \ud \tr \ud \tt \label{Ge} \\
 \approx (N\!-\!1) f(\xi_1,t) \rho(\x_1,t) \iint \!g(\tr, \tt)  u'(\tr/\epsilon) \cos(\tt) \tr \, \ud \tr \ud \tt =  (N\!-\!1) \zeta f(\xi_1,t) \rho(\x_1,t), 
\end{multline}
using that $\e(\theta_1 - \tt) \cdot \e(\theta_1) = \cos \tt$ and $\zeta := \int_0^\infty \int_0^{2\pi} g(\tr, \tt) u'(\tr/\epsilon) \cos(\tt) \tr \, \ud \tr \ud \tt$, and
\begin{equation}\label{G_par}
	G_{\|} = \frac{[\nabla_\x f - (\e \cdot \nabla f_x) \e] \cdot {\bf G}}{|\nabla f|^2}.
\end{equation}
Finally, inserting \cref{Ge} and \cref{G_par} into \cref{eqf2_G} one obtains \cref{mod2}.

\subsection{The crowded Goldstein--Taylor model}
\label{sec:GTmodel}
Here we show that the one-dimensional version of \cref{mod2} coincides with a crowded version of the Goldstein--Taylor model \cite{goldstein1951diffusion,taylor1922diffusion}. In particular, denote by $\fL(x,t)$ and $\fR(x,t)$ the densities of the left- and right-moving particles, respectively. The starting point is again the $N-$particle Fokker--Planck equation \cref{N_eq} with the only difference that the continuous diffusion in angle (modelled by the $D_R$ term) becomes a discrete jump between the two only possible orientations, corresponding to angles $\theta_i = 0, \pi$. Specifically, it is replaced by $k \sum_{\vec \xi'} F_N(\vec \xi',t) - kN F_N(\vec \xi,t)$, where the sum in $\vec \xi'$ is over configuration that are one jump in angle away from $\vec \xi$.
The right-moving density is obtained as
$$
\fR(x, t) = \int_{\Upsilon^N} F_N(\vec \xi, t) \delta(x - x_1) \delta(\theta_1 - 0) \ud \vec \xi,
$$
and similarly for $\fL$. The equation for $\fR$ is found by integrating that of $F_N$, leading to a slight modification of \cref{1_eq_2}, namely
\begin{equation} \label{eq_2_1d}
\begin{aligned}
		\partial_t \fR(x_1, t) &= \partial_{x_1}  \left[ D_T \partial_{x_1} \fR -v_0  \fR +  G_R(x_1,t) \right] + k (\fL - \fR),\\
	G_R(x_1,t)	&= (N-1) \int_{\Omega} \left[F_{RR}(x_1, x_2, t) + F_{RL}(x_1,x_2,t) \right] \partial_{x_1} u(|x_1-x_2|/\epsilon) \ud x_2,
	\end{aligned}
\end{equation}
where $F_{RR}$ and $F_{RL}$ are the two-particle densities corresponding to having the first particle moving right, and the second particle moving right and left, respectively. Following a similar derivation as for the two-dimensional model above, we introduce $\tx$ such that
$
x_2 = x_1 + \tilde x,
$
and a decomposition
$$
F_{RR}(x_1, x_2, t) + F_{RL}(x_1,x_2,t) \approx \fR(x_1,t) \rho(x_2,t) g(x_2-x_1) = \fR(x_1,t) \rho(x_1 + \tilde x, t) g(\tilde x). 
$$
Substituting this into $G_R$ and expanding in $\epsilon$, we obtain, at leading order,
$$
G_R(x_1,t) = (N-1) \fR(x_1,t)\rho(x_1,t) \int_\Omega [-u'(|\tilde x|/\epsilon)] g( \tilde x) \ud \tilde x \approx v_0\phi f_R(x_1,t) \rho(x_1,t),
$$
where $\phi = N \epsilon$, positive using that $u'(r)\le 0$ (repulsive force) and $g(r) \ge 0$. Repeating the same when the first particle is left-moving (in that case, the change of variables is $x_2 = x_1 -  \tilde x$), we obtain
$$
G_L(x_1,t) = - (N-1) \fL(x_1,t)\rho(x_1,t) \int_\Omega [-u'(|\tilde x|/\epsilon)] g(\tilde x) \ud \tilde x \approx -v_0\phi \fL(x_1,t) \rho(x_1,t).
$$ 
This leads to the system of equations
\begin{align} \label{GT_model_SM}
\begin{aligned}
\partial_t \fR + v_0\partial_x [\fR (1-\phi \rho)] &= D_T\partial_{xx} \fR + k (\fL - \fR), \\
\partial_t \fL - v_0\partial_x [\fL (1-\phi \rho)] &= D_T\partial_{xx} \fL + k (\fR - \fL).
\end{aligned}
\end{align}

\subsection{Linear stability of the crowded Goldstein--Taylor model}
\label{sec:1dinstability}
We consider the one-dimensional model \eqref{GT_model_SM} after rescaling as described in \cref{sec:rescale} (the same rescaling as in the two-dimensional cases applies by replacing $D_R$ by $k$)
\begin{equation}
\label{model21d_c} 	
\partial_t f + \Pe \partial_x[  f (1- \rho) e_\theta] =  \partial_x^2 f - e_\theta p.
\end{equation}
We linearise \cref{model21d_c} around its homogeneous state, $f = \phi/2 + \delta \tilde f, \rho = \phi + \delta \tilde \rho, p = \delta \tilde p$. Keeping terms to order $\delta$ leads to
\begin{equation} \label{eig_1dmod2}
	\partial_t \tilde f = \partial_{xx} \tilde f - e_\theta \tilde p - \Pe \partial_x [(1-\phi) \tilde f e_\theta - (\phi/2) \tilde \rho e_\theta]:= L (\tilde f). 
\end{equation}
For perturbations of the form $\tilde \fR  = \sum_{n\ge 1} \alpha_n e^{\lambda t + i 2\pi n x}$ and $\tilde \fL = \sum_{n\ge 1} \beta_n e^{\lambda t + i 2\pi n x}$, we obtain the following eigenvalue problem
\begin{align} \label{GT_eig}
\begin{aligned}
\lambda \alpha_n & =- \alpha_n (2\pi n)^2  - (\alpha_n - \beta_n) - \Pe (i 2 \pi n) \left[ \left(1- \frac{3\phi}{2} \right) \alpha_n - \frac{\phi}{2} \beta_n \right] , \\
\lambda \beta_n & = - \beta_n (2\pi n)^2   + (\alpha_n - \beta_n) - \Pe (i 2 \pi n) \left[ \left(\frac{3\phi}{2} -1 \right) \beta_n + \frac{\phi}{2} \alpha_n \right].
\end{aligned}
\end{align}
Writing $\gamma_n = \beta_n/\alpha_n$, we arrive at:
\begin{align} \label{GT_fourier2}
\begin{aligned}
\lambda & =-  (2\pi n)^2  - 1 + \gamma_n - 2 \pi  n \Pe i \left( 1- \frac{3\phi}{2} - \frac{\phi}{2} \gamma_n \right) , \\
\lambda \gamma_n & = - \gamma_n (2\pi n)^2   + 1 - \gamma_n - 2 \pi n \Pe i \left[ \left(\frac{3\phi}{2} -1 \right) \gamma_n + \frac{\phi}{2} \right].
\end{aligned}
\end{align}
The solution of \cref{GT_fourier2} is given by
\begin{equation}
\lambda = 	-1 - (2\pi n)^2 \pm \sqrt{1 + (2\pi n  \Pe)^2 r(\phi)}, \qquad r(\phi) = (1-\phi)(2 \phi-1).
\end{equation}
We have the $r(\phi)<0$ for $\phi<1/2$, so $\text{Re}(\lambda)<0$ if $\phi<1/2$, and that $r(\phi)>0$ for $\phi \in (1/2,1)$ with a maximum at $\phi = 3/4$. 
Imposing $\text{Re}(\lambda) >0$ we arrive at the condition $\Pe^2 r(\phi) > 2 + (2\pi n)^2 \ge 2 + 4 \pi^2$.\footnote{We note that this is similar to Eq. (7) in \cite{kourbane2018exact}, with a slight difference in the lower bound due our choice of a periodic domain of length one.} 
We plot the curve $\Pe^2 r(\phi) = 2 + 4 \pi^2$ in  \cref{fig:dispersion1d} (solid blue curve), which separates the stable region (low mass and/or low speed, below the curve) and the unstable region (large mass and speed, above the curve). 
We note that $n=1$ corresponds to the most unstable mode and that $\phi = 3/4$ is the most unstable volume fraction, meaning that these values require the smallest speed $\Pe$ for perturbations to grow.
\def \scl {1.0}
\begin{figure}[htb]
\begin{center}
\psfrag{Pe}[][][\scl][-90]{$\Pe$}
\psfrag{phi}{$\phi$}
\psfrag{L}{$L$}
\psfrag{Ls}{$L^S$}
\includegraphics[width = 0.5\textwidth]{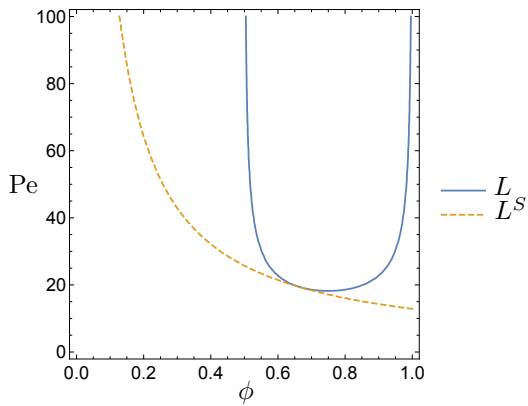}
\end{center}
  \caption{Curves for the onset of instability of the full operator $L$ and the symmetric part $L^S$ of the one-dimensional Model 2.}
 \label{fig:dispersion1d}
\end{figure}

Due to its simpler spectrum, we also consider the eigenvalue problem associated to the symmetric operator $L^S = (L + L^*)/2$, where $L$ is given in \cref{eig_1dmod2} and  $L^*$ is its adjoint. Using that 
$
L^*(\tilde f) = \partial_{xx} \tilde f - e_\theta \tilde p + \Pe \partial_x [(1-\phi) \tilde f e_\theta - (\phi/2) \tilde p],
$
we arrive at the symmetric operator
\begin{equation} \label{1dmod2_S}
L^S (\tilde f) = \partial_{xx} \tilde f - e_\theta \tilde p + \frac{\Pe \phi}{4} \partial_x [  \tilde \rho e_\theta - \tilde p].
\end{equation}
Solving the eigenvalue problem $\lambda \tilde f = L^S(\tilde f)$ we obtain
\begin{equation}
	\label{eig_1dsym}
	\lambda = -1 - (2 \pi n)^2 \pm \sqrt{1 + (n \pi \phi \Pe)^2}.
\end{equation}
Imposing that $\lambda > 0$ we obtain $(\phi \Pe)^2 > 8(1 + 2 (n \pi)^2)\ge 8 (1+ 2 \pi^2)$. We plot the curve $\Pe = \sqrt{8 (1+ 2 \pi^2)}/\phi$, separating the regions of stability instability (below and above the curve, respectively) in  \cref{fig:dispersion1d} (orange dash line) and \cref{fig:dispersion} (yellow short dash line). We observe that the symmetric operator curve is always below the full operator curve, meaning that the instability region of the full model is contained in the instability region of the symmetric part. This is to be expected, as per \cref{rem:antisymmetric}.

\section{Supplementary PDE simulations} \label{sec:A_PDEsims}
In this section we provide more examples of phase separation in Models 2, 3 and 4. As in \S\ref{sec:numerics}, we plot the space density $\rho$, the mean angular direction $\q = \p/\rho$ and the mobility in $x$ integrated over $\theta$, $m_{11} = \int \tilde M_{11} \ud \theta$. Unless otherwise stated, the simulations here use a grid with $N_x = N_y = N_\theta = 21$. The parameters used are shown in the corresponding caption. \cref{fig:model2_Tf=1.0_phi=0.6_v0=20.0,fig:model2_Tf=0.2_phi=0.6_v0=40.0,fig:model2_Tf=1.0_phi=0.6_v0=40.0,fig:model2_Tf=0.2_phi=0.6_v0=60.0} show examples of Model 2 \cref{model2_c}, Model 3 \cref{model3_c}, and Model 4 \cref{model4_c}.


\begin{figure}[htb]
\begin{center}
\includegraphics[width = \textwidth]{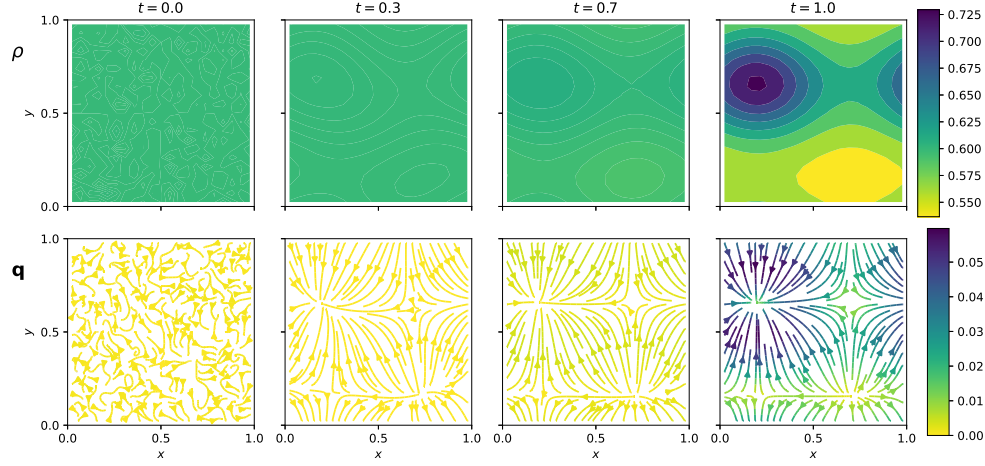}
\end{center}
  \caption{Time-evolution of Model 2 \eqref{model2_c} with $\phi = 0.6$, $\Pe = 20$ and final time $T = 1$.}
 \label{fig:model2_Tf=1.0_phi=0.6_v0=20.0}
\end{figure}

\begin{figure}[htb]
\begin{center}
\includegraphics[width = \textwidth]{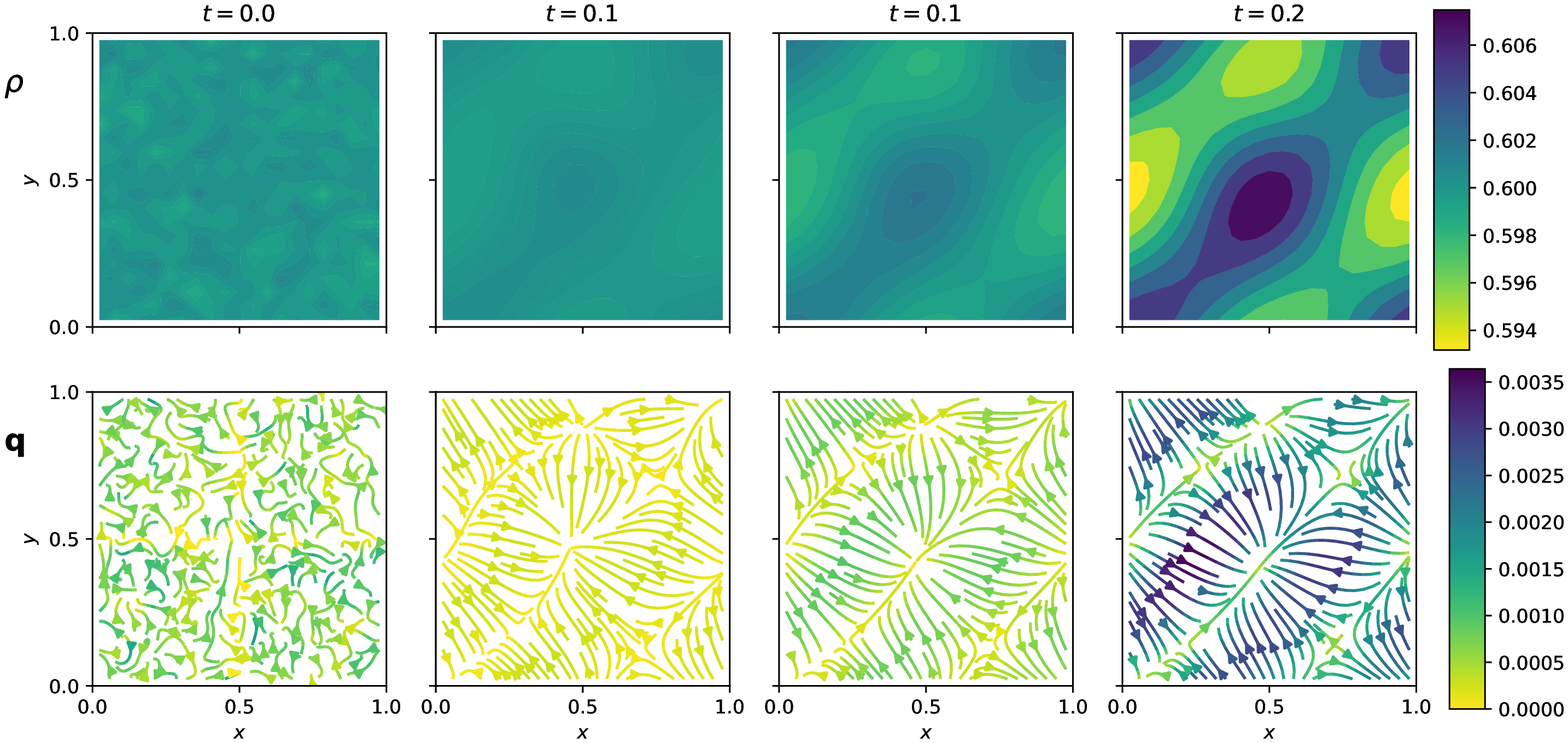}
\end{center}
  \caption{Time-evolution of Model 2 \eqref{model2_c} with $\phi = 0.6$, $\Pe = 40$ and final time $T = 0.2$.}
 \label{fig:model2_Tf=0.2_phi=0.6_v0=40.0}
\end{figure}

\begin{figure}[htb]
\begin{center}
\includegraphics[width = \textwidth]{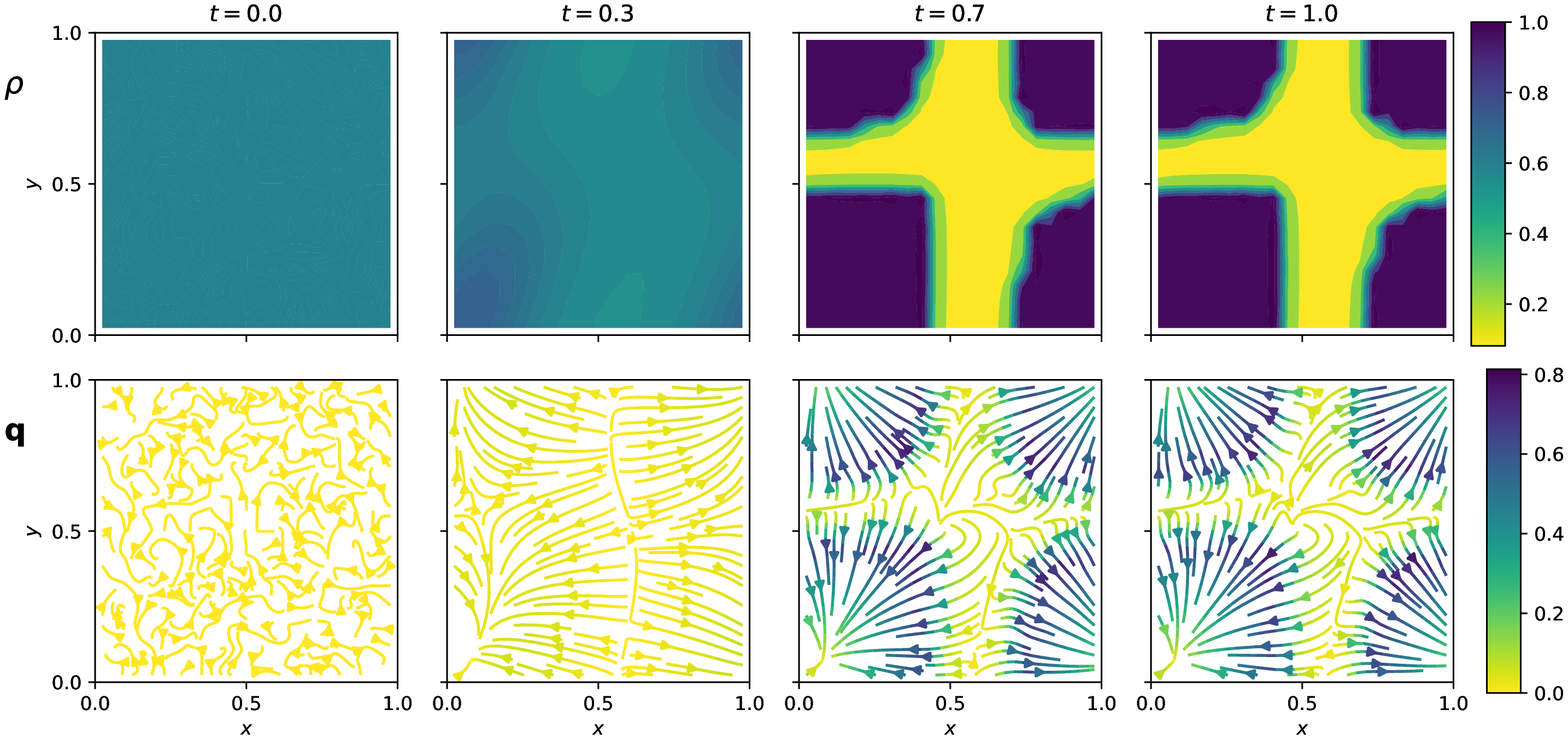}
\end{center}
  \caption{Time-evolution of Model 2 \eqref{model2_c} with $\phi = 0.6$, $\Pe = 40$ and final time $T = 1$.}
 \label{fig:model2_Tf=1.0_phi=0.6_v0=40.0}
\end{figure}

\begin{figure}[htb]
\begin{center}
\includegraphics[width = \textwidth]{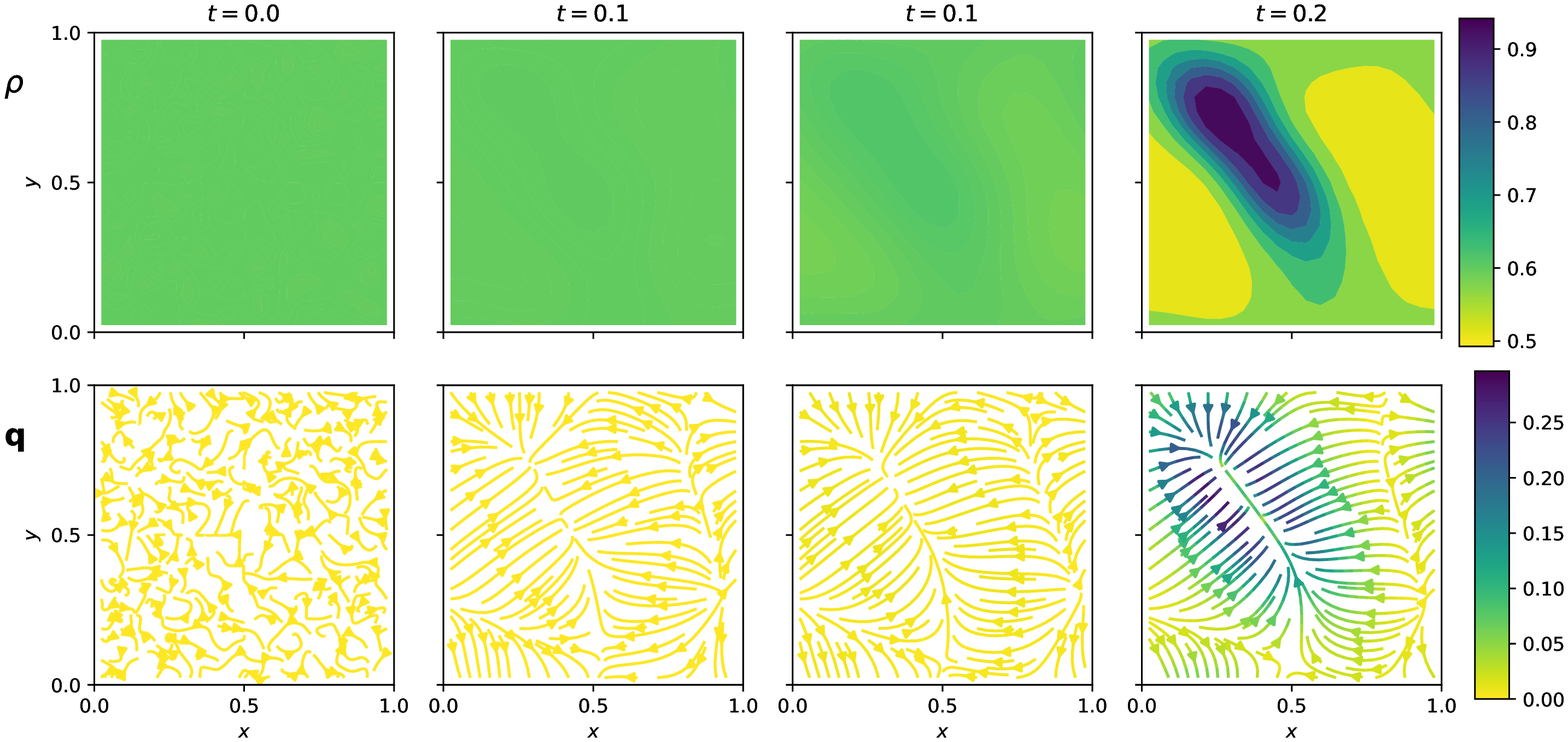}
\end{center}
  \caption{Time-evolution of Model 2 \eqref{model2_c} with $\phi = 0.6$, $\Pe = 60$ and final time $T = 0.2$.}
 \label{fig:model2_Tf=0.2_phi=0.6_v0=60.0}
\end{figure}


\begin{figure}[htb]
\begin{center}
\includegraphics[width = \textwidth]{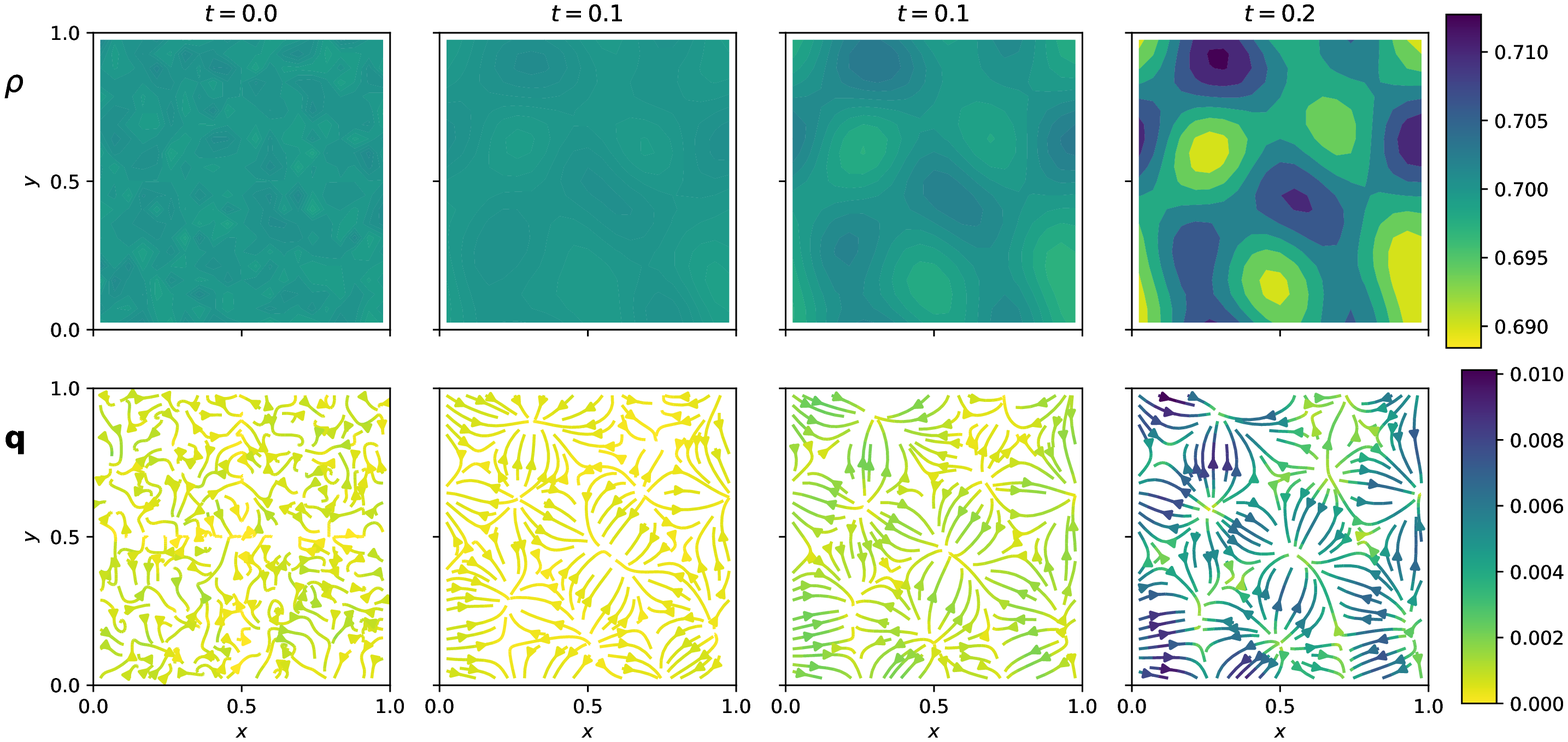}
\end{center}
  \caption{Time-evolution of Model 2 \eqref{model2_c} with $\phi = 0.7$, $\Pe = 20$ and final time $T = 0.2$.}
 \label{fig:model2_Tf=0.2_phi=0.7_v0=20.0}
\end{figure}

\begin{figure}[htb]
\begin{center}
\includegraphics[width = \textwidth]{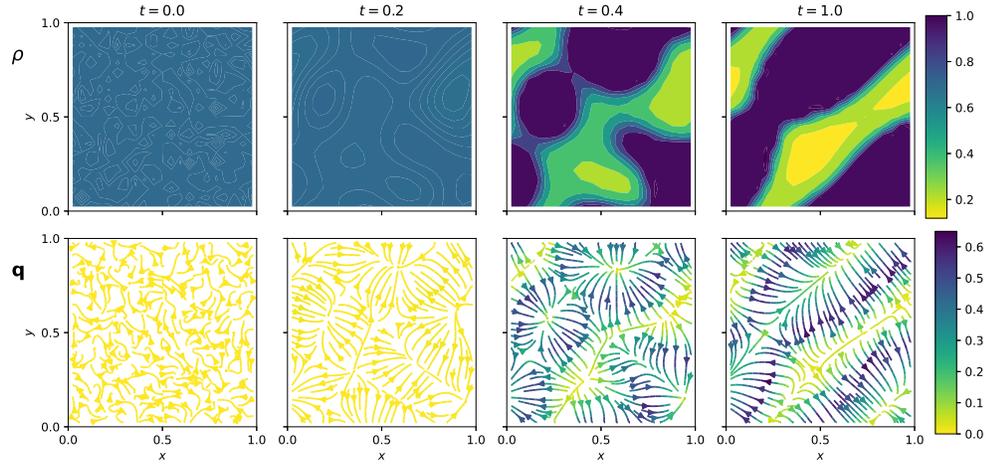}
\end{center}
  \caption{Time-evolution of Model 2 \eqref{model2_c} with $\phi = 0.7$, $\Pe = 20$ and final time $T = 1$.}
 \label{fig:model2_Tf=1.0_phi=0.7_v0=20.0}
\end{figure}



\begin{figure}[htb]
\begin{center}
\includegraphics[width = \textwidth]{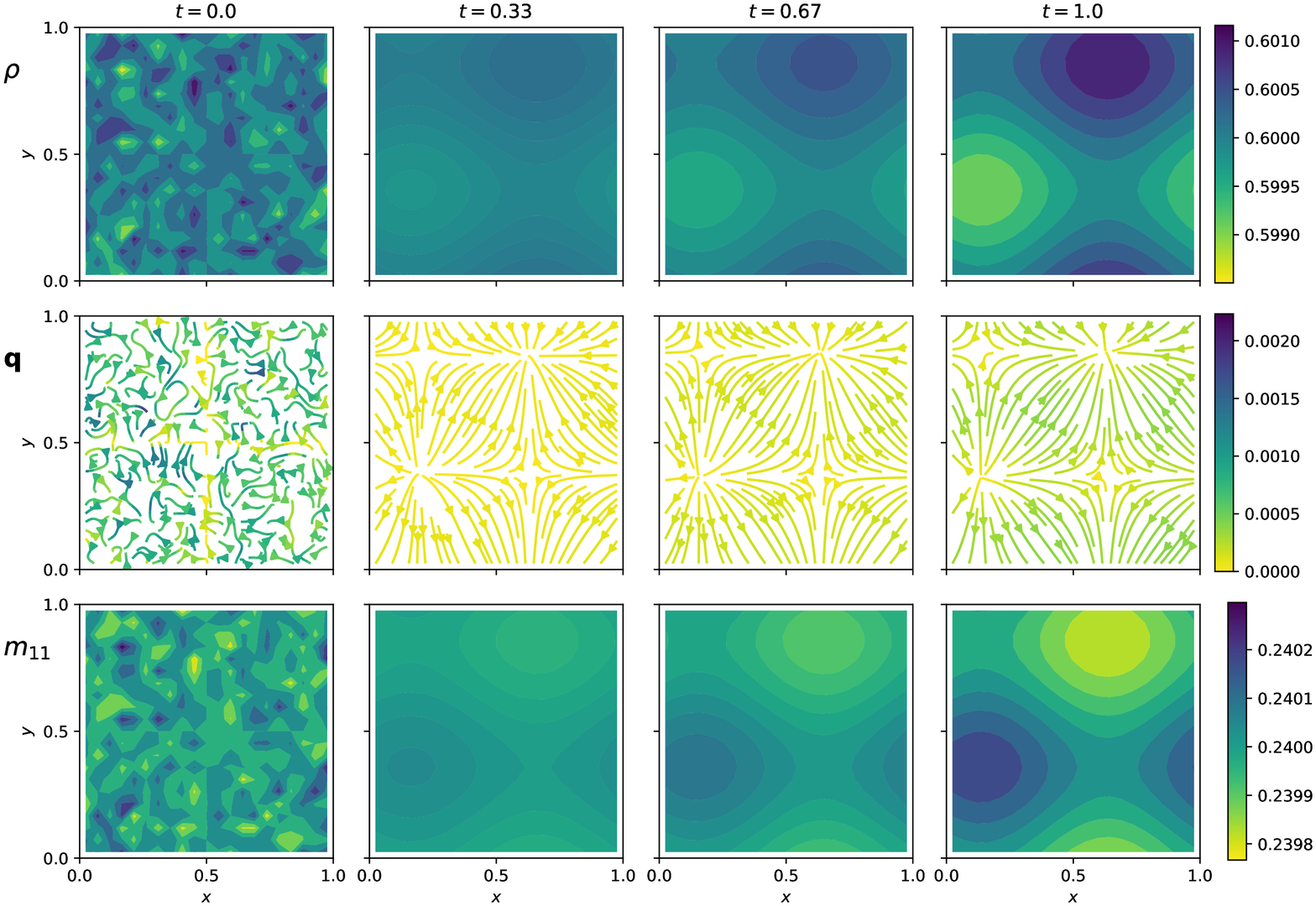}
\end{center}
  \caption{Time-evolution of Model 3 \eqref{model3_c} with $\phi = 0.6$, $\Pe = 40$ and final time $T = 1$.}
 \label{fig:model3_Tf=1.0_phi=0.6_v0=40.0}
\end{figure}

\begin{figure}[htb]
\begin{center}
\includegraphics[width = \textwidth]{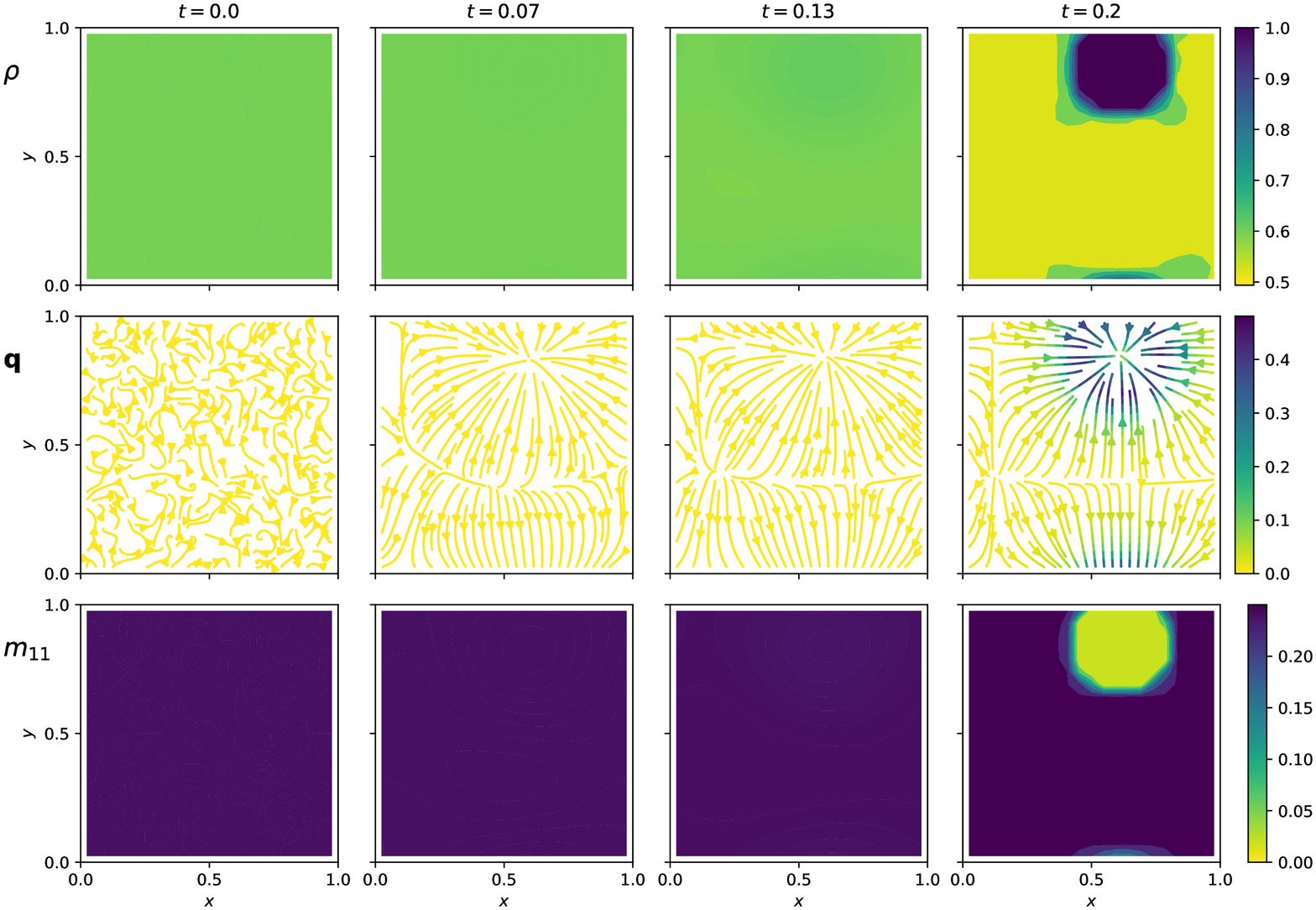}
\end{center}
  \caption{Time-evolution of Model 3 \eqref{model3_c} with $\phi = 0.6$, $\Pe = 60$ and final time $T = 0.2$.}
 \label{fig:model3_Tf=0.2_phi=0.6_v0=60.0}
\end{figure}

\begin{figure}[htb]
\begin{center}
\includegraphics[width = \textwidth]{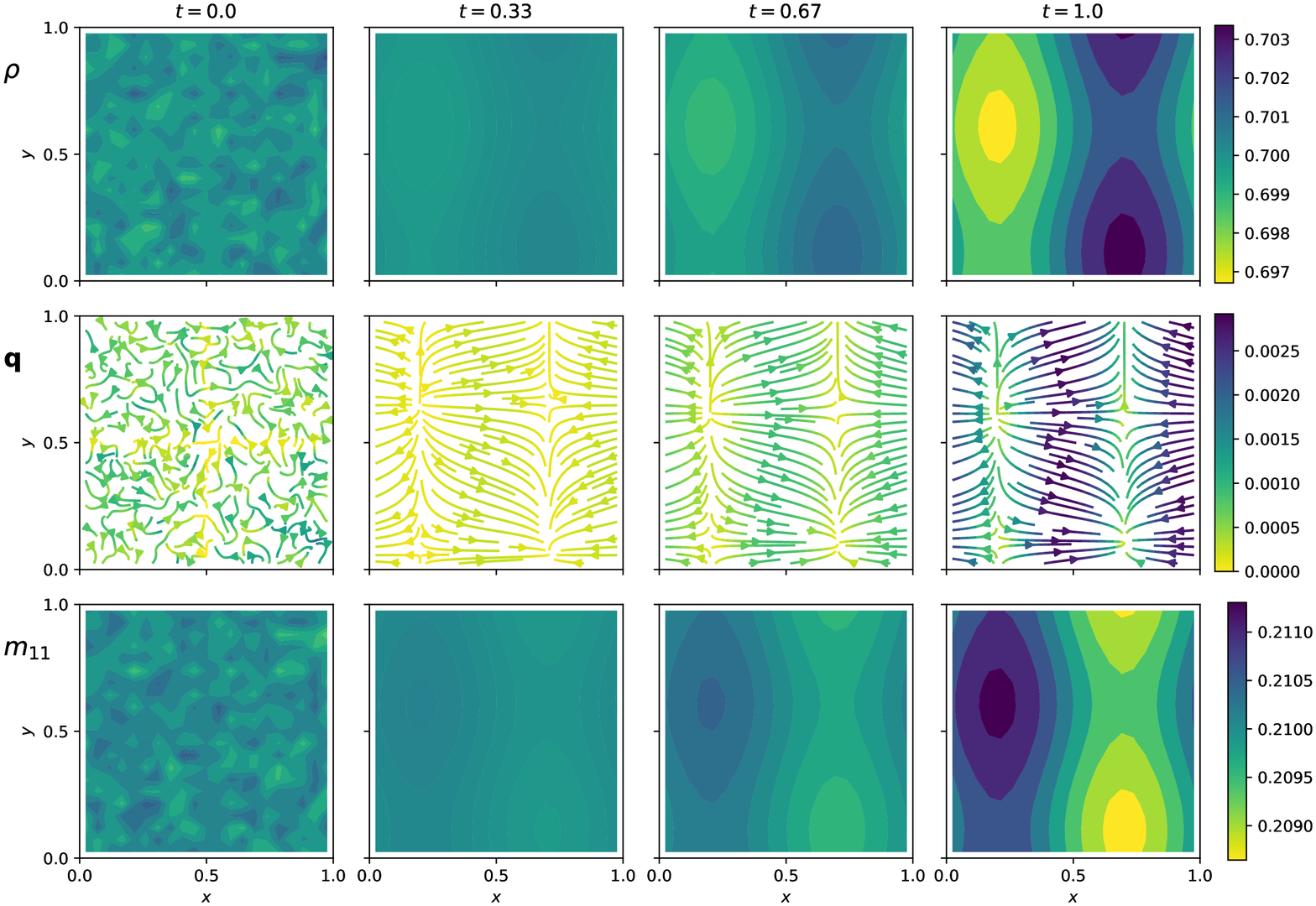}
\end{center}
  \caption{Time-evolution of Model 3 \eqref{model3_c} with $\phi = 0.7$, $\Pe = 20$ and final time $T = 1$.}
 \label{fig:model3_Tf=1.0_phi=0.7_v0=20.0}
\end{figure}

\begin{figure}[htb]
\begin{center}
\includegraphics[width = \textwidth]{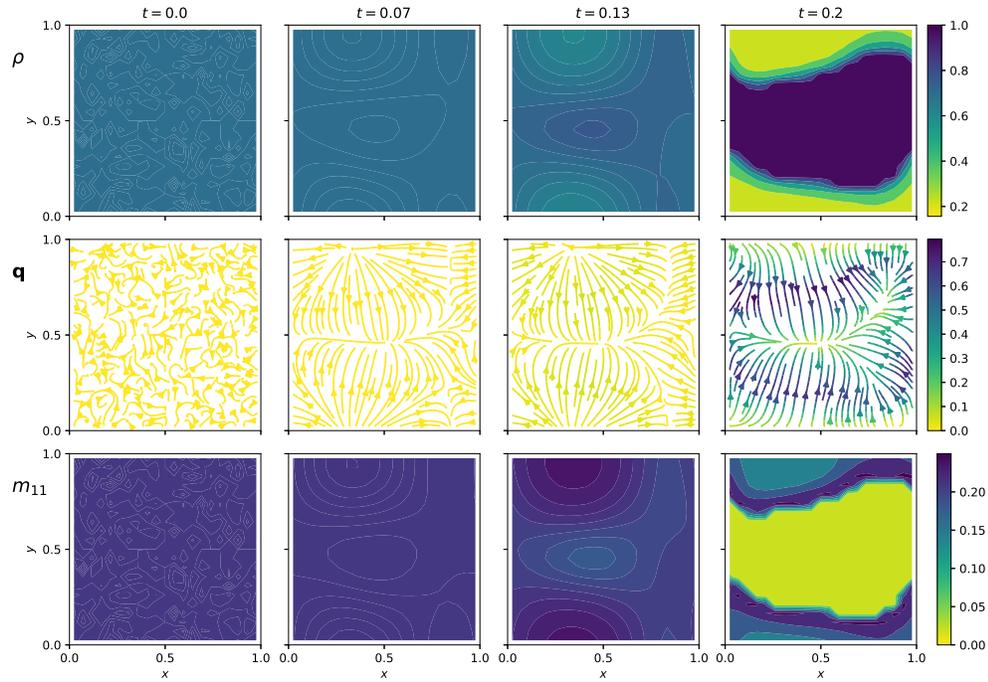}
\end{center}
  \caption{Time-evolution of Model 3 \eqref{model3_c} with $\phi = 0.7$, $\Pe = 40$ and final time $T = 0.2$.}
 \label{fig:model3_Tf=0.2_phi=0.7_v0=40.0}
\end{figure}

\begin{figure}[htb]
\begin{center}
\includegraphics[width = \textwidth]{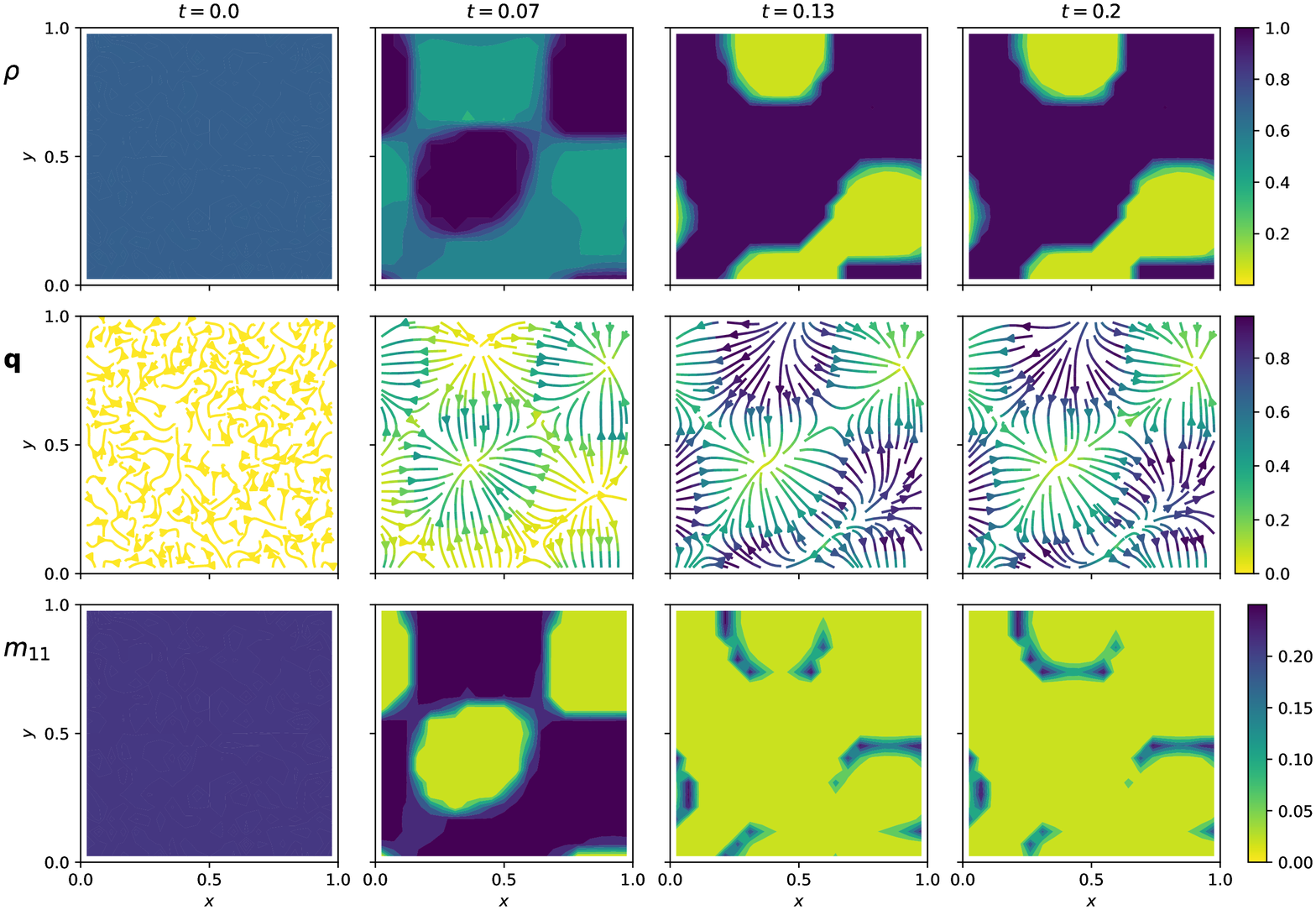}
\end{center}
  \caption{Time-evolution of Model 3 \eqref{model3_c} with $\phi = 0.7$, $\Pe = 60$ and final time $T = 0.2$.}
 \label{fig:model3_Tf=0.2_phi=0.7_v0=60.0}
\end{figure}


\begin{figure}[htb]
\begin{center}
\includegraphics[width = \textwidth]{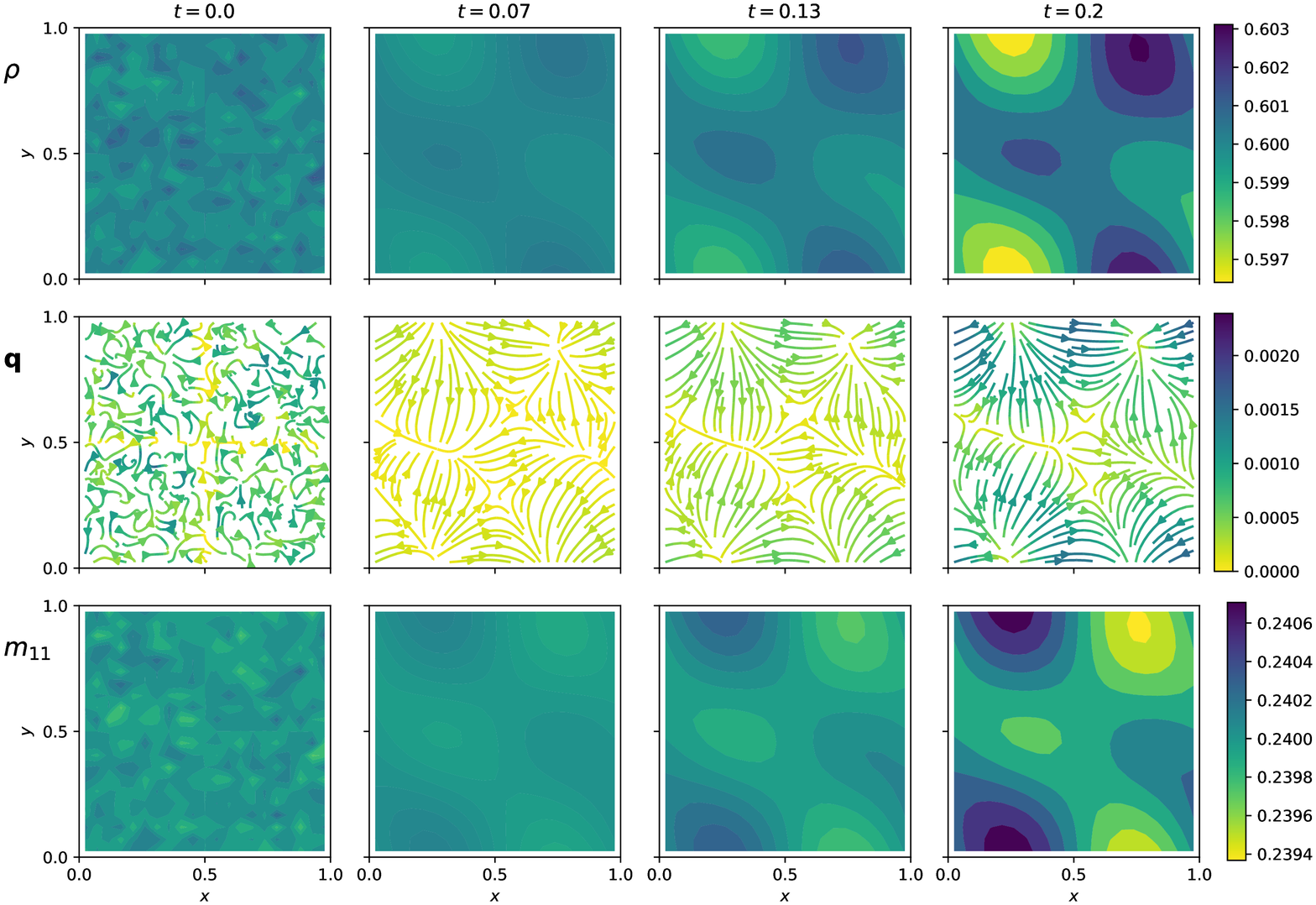}
\end{center}
  \caption{Time-evolution of Model 4 \eqref{model4_c} with $\phi = 0.6$, $\Pe = 60$ and final time $T = 0.2$.}
 \label{fig:model4_Tf=0.2_phi=0.6_v0=60.0}
\end{figure}

\begin{figure}[htb]
\begin{center}
\includegraphics[width = \textwidth]{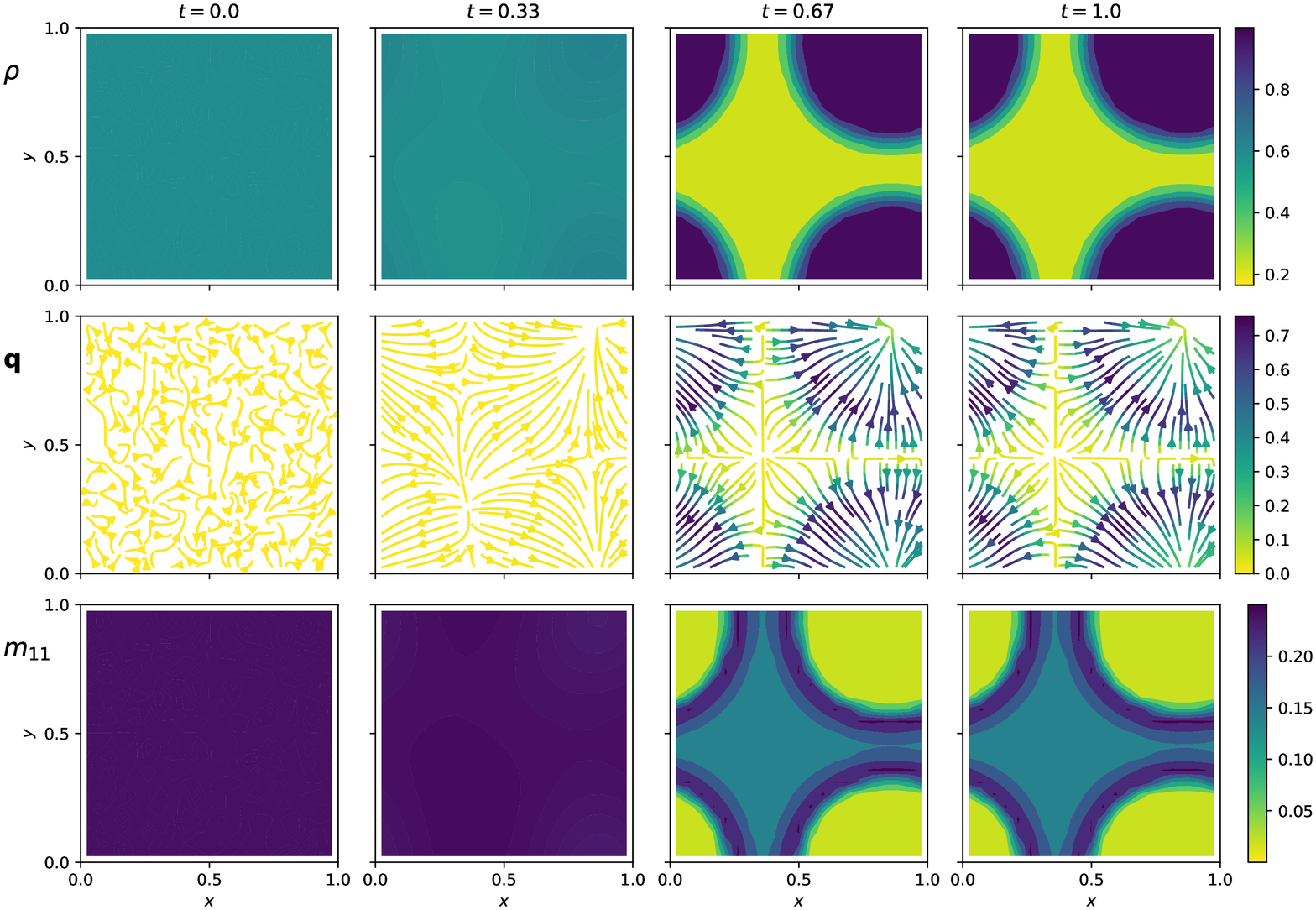}
\end{center}
  \caption{Time-evolution of Model 4 \eqref{model4_c} with $\phi = 0.6$, $\Pe = 60$ and final time $T = 1.0$.}
 \label{fig:model4_Tf=1.0_phi=0.6_v0=60.0}
\end{figure}

\begin{figure}[htb]
\begin{center}
\includegraphics[width = \textwidth]{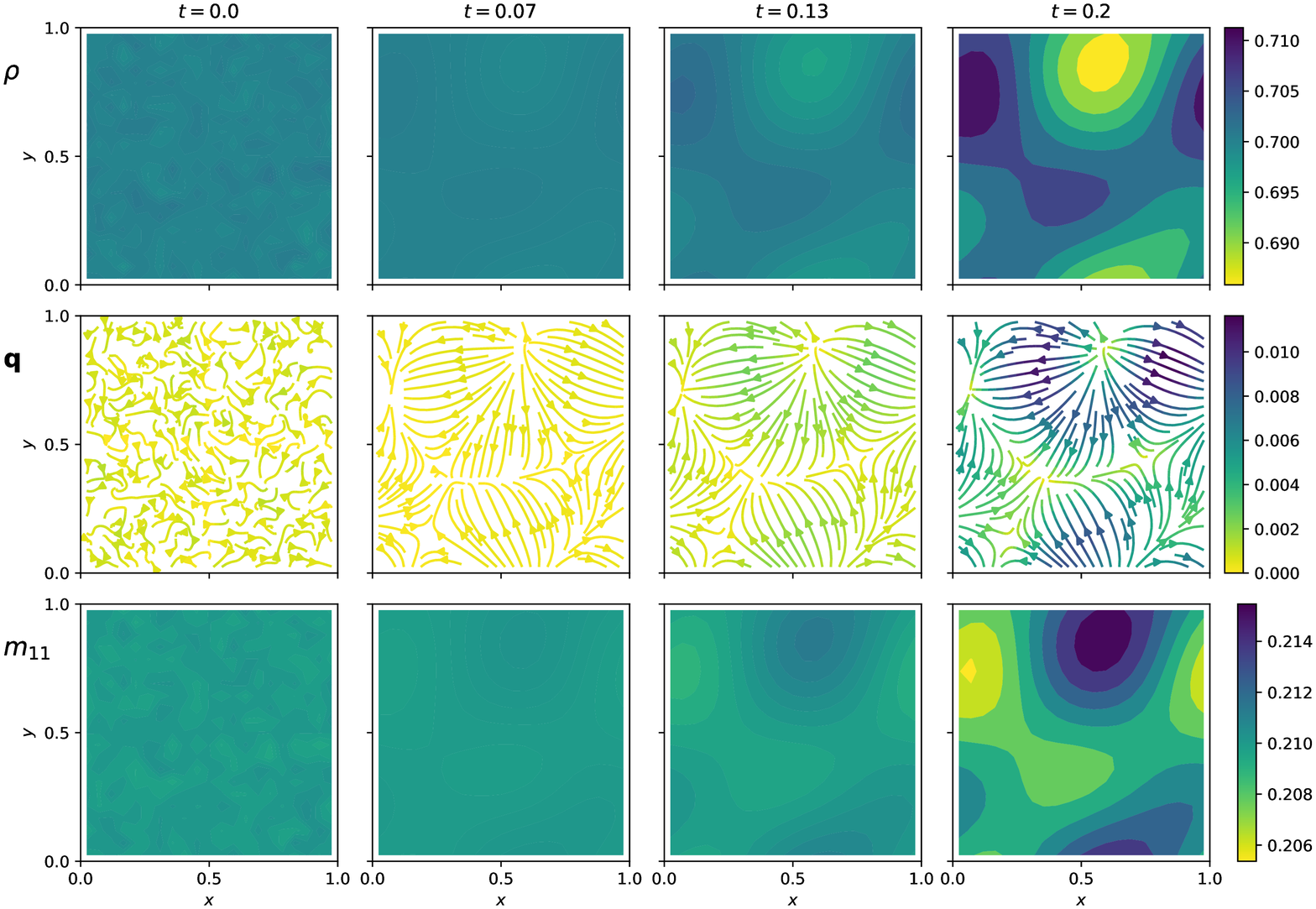}
\end{center}
  \caption{Time-evolution of Model 4 \eqref{model4_c} with $\phi = 0.7$, $\Pe = 40$ and final time $T = 0.2$.}
 \label{fig:model4_Tf=0.2_phi=0.7_v0=40.0}
\end{figure}

\begin{figure}[htb]
\begin{center}
\includegraphics[width = \textwidth]{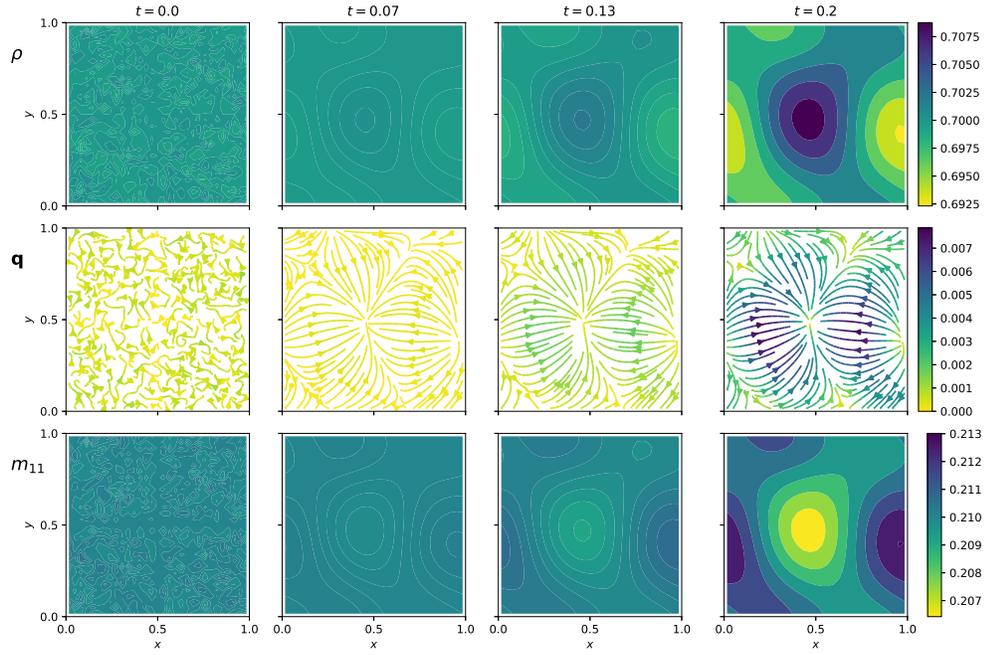}
\end{center}
  \caption{Time-evolution of Model 4 \eqref{model4_c} with $\phi = 0.7$, $\Pe = 40$ and final time $T = 0.2$. Finer grid $N_x = N_y = N_\theta = 31$.}
 \label{fig:model4_Tf=0.2_phi=0.7_v0=40.0_fine}
\end{figure}

\begin{figure}[htb]
\begin{center}
\includegraphics[width = \textwidth]{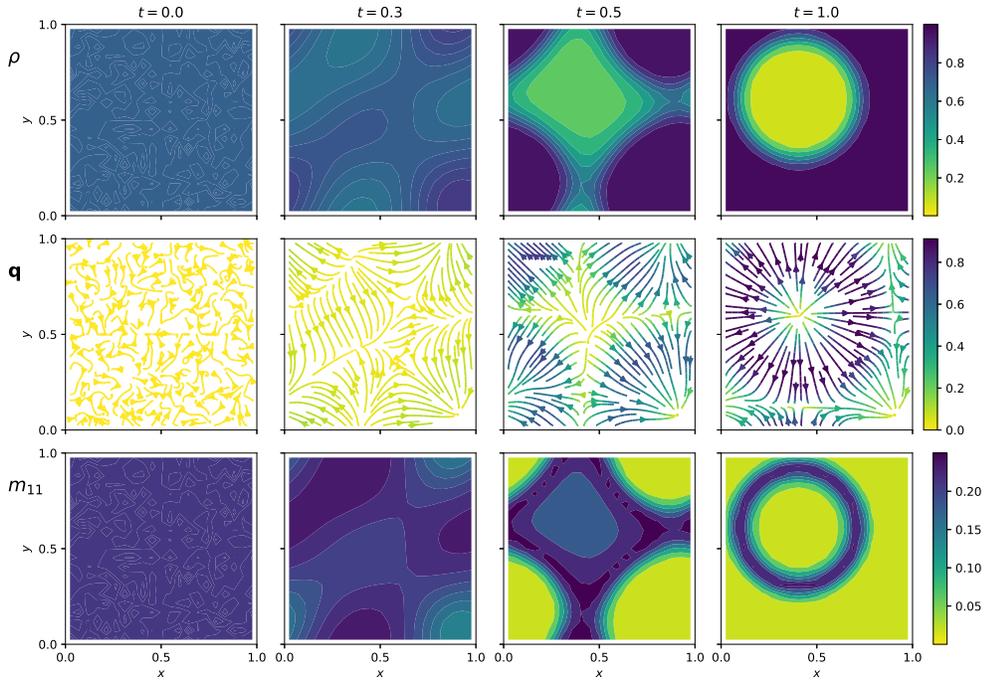}
\end{center}
  \caption{Time-evolution of Model 4 \eqref{model4_c} with $\phi = 0.7$, $\Pe = 40$ and final time $T = 1.0$.}
 \label{fig:model4_Tf=1.0_phi=0.7_v0=40.0}
\end{figure}

\begin{figure}[htb]
\begin{center}
\includegraphics[width = \textwidth]{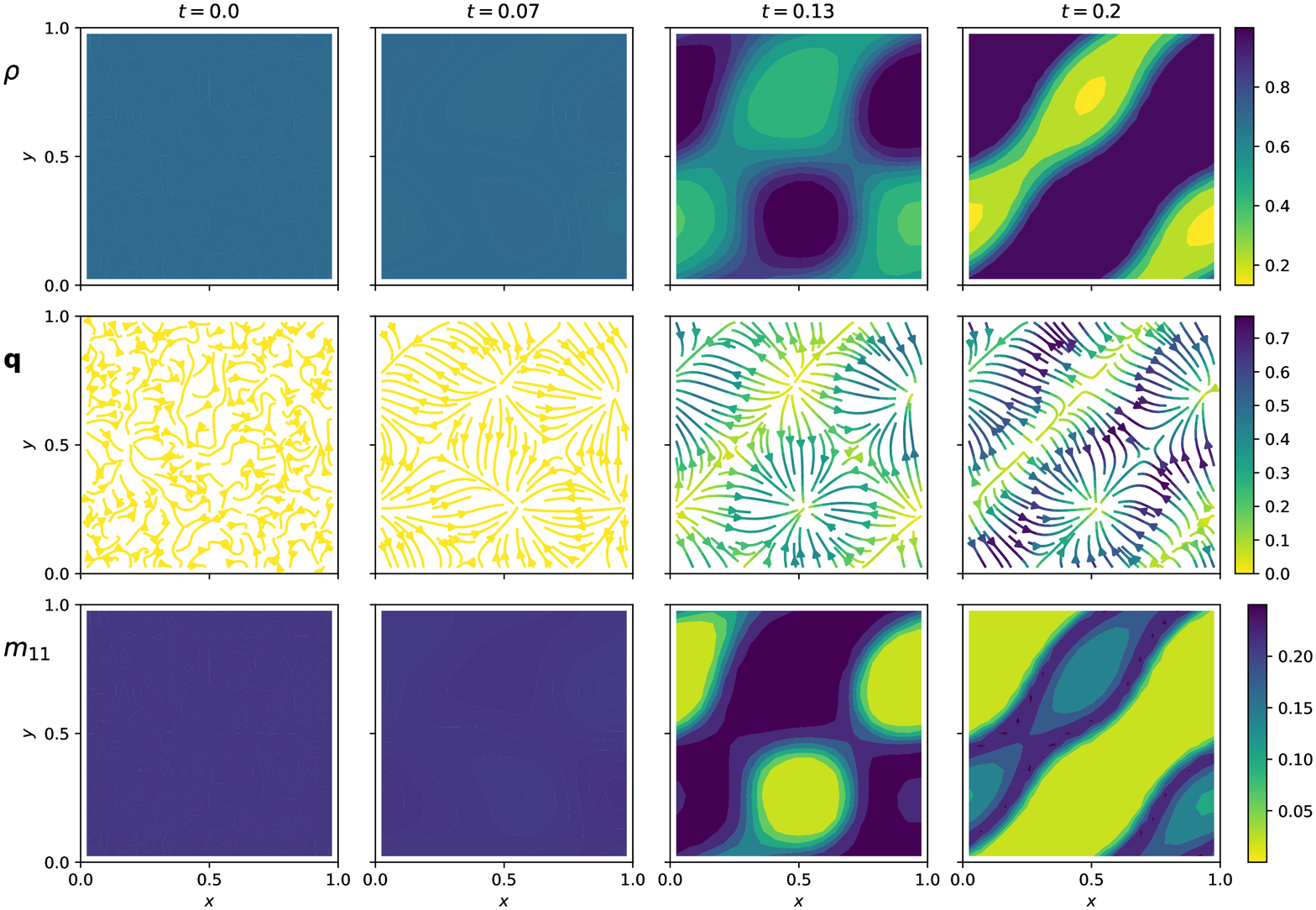}
\end{center}
  \caption{Time-evolution of Model 4 \eqref{model4_c} with $\phi = 0.7$, $\Pe = 60$ and final time $T = 0.2$.}
 \label{fig:model4_Tf=0.2_phi=0.7_v0=60.0}
\end{figure}

\begin{figure}[htb]
\begin{center}
\includegraphics[width = \textwidth]{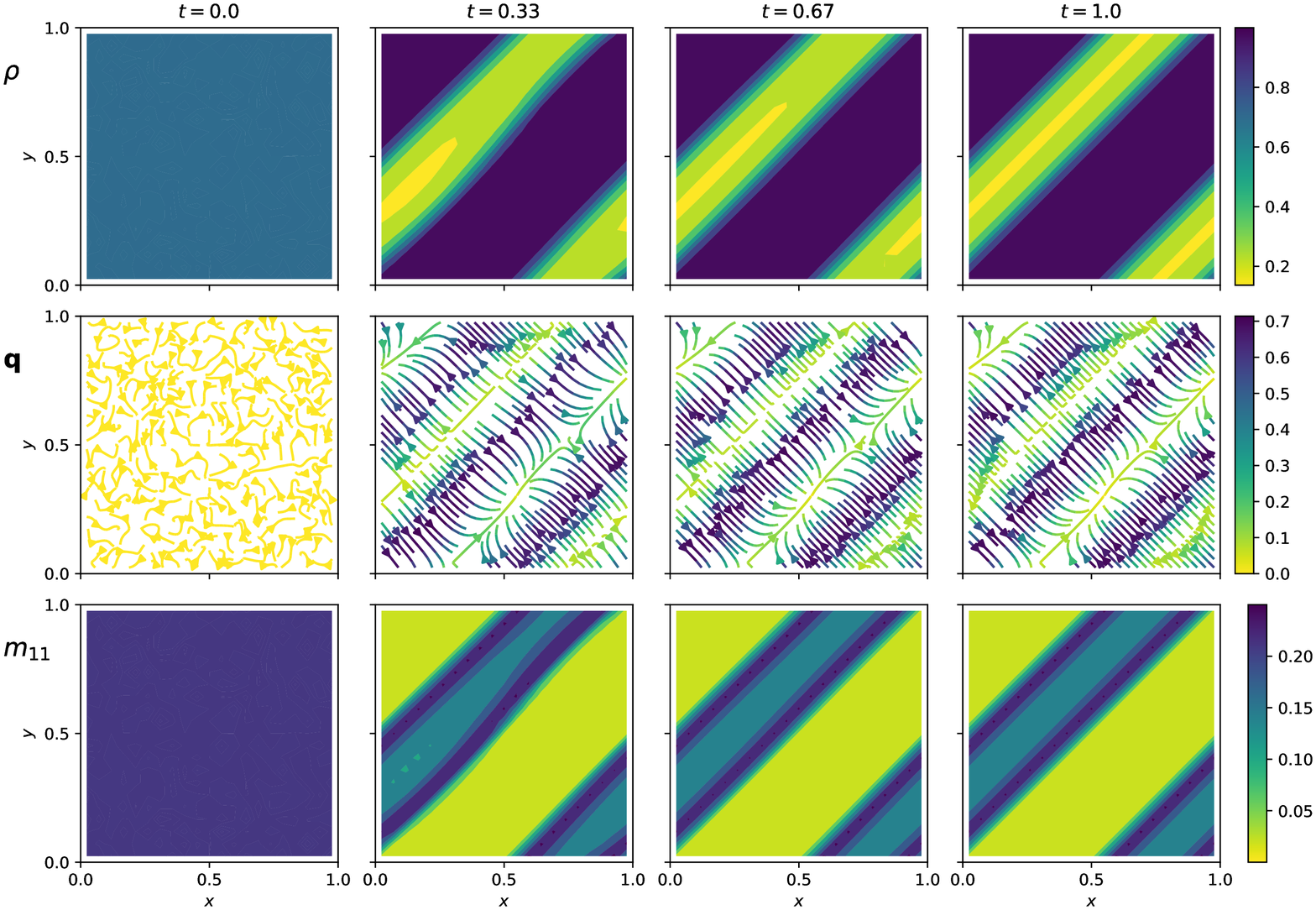}
\end{center}
  \caption{Time-evolution of Model 4 \eqref{model4_c} with $\phi = 0.7$, $\Pe = 60$ and final time $T = 1$.}
 \label{fig:model4_Tf=1.0_phi=0.7_v0=60.0}
\end{figure}

\section{Supplementary stochastic simulations of Model 4}
In this section, we show examples of the microscopic Model 4. In particular, we generate several runs of varying occupied fraction $\phi$ and effective speed $\Pe$, all with $N= 500$ particles starting that are initially independently and uniformly distributed in space and orientation (while satisfying the constraint of one particle per lattice site at most. The total number of lattice sites $M$ is chosen such that $M = N/\phi$. The system is evolved according to the active ASEP described in \cref{sec:model3} until $T=1$ using a fixed time-step $\Delta t = 10^{-4}$. \cref{fig:ABM4_Tf=1.0_v0=100.0} shows eight runs at fixed $\Pe = 100$ and increasing values of $\phi$. We observe no phase segregation for the lowest densities ($\phi = 0.1$ and $\phi=0.2$), and a sharp change for $\phi\ge 0.3$ with a dense and dilute phases forming. All the particles are contained in the dense phase for the three largest values of $\phi$. 
We define the local polarisation at a lattice point $\x$ by $\q(\x) = \sum_{\X_j \in \mathcal N_9(\x)} \exp(i \Theta_j)$, where the sum runs over all the particles with orientation $\Theta_j$ whose position $\X_j$ lies within the Moore neighbourhood of $\x$. The absolute value $|\q|$ is shown as a colormap in \cref{fig:ABM4_Tf=1.0_v0=100.0}, with white corresponding to $|\q| = 0$ and dark green to $|\q| = 1$. 
Another example, this time with fixed occupied fraction $\phi = 0.59$ and increasing values of $\Pe$ is shown in \cref{fig:ABM4_Tf=1.0_phi=0.6}.

\begin{figure}
\begin{center}
	\includegraphics[width = 0.40\textwidth]{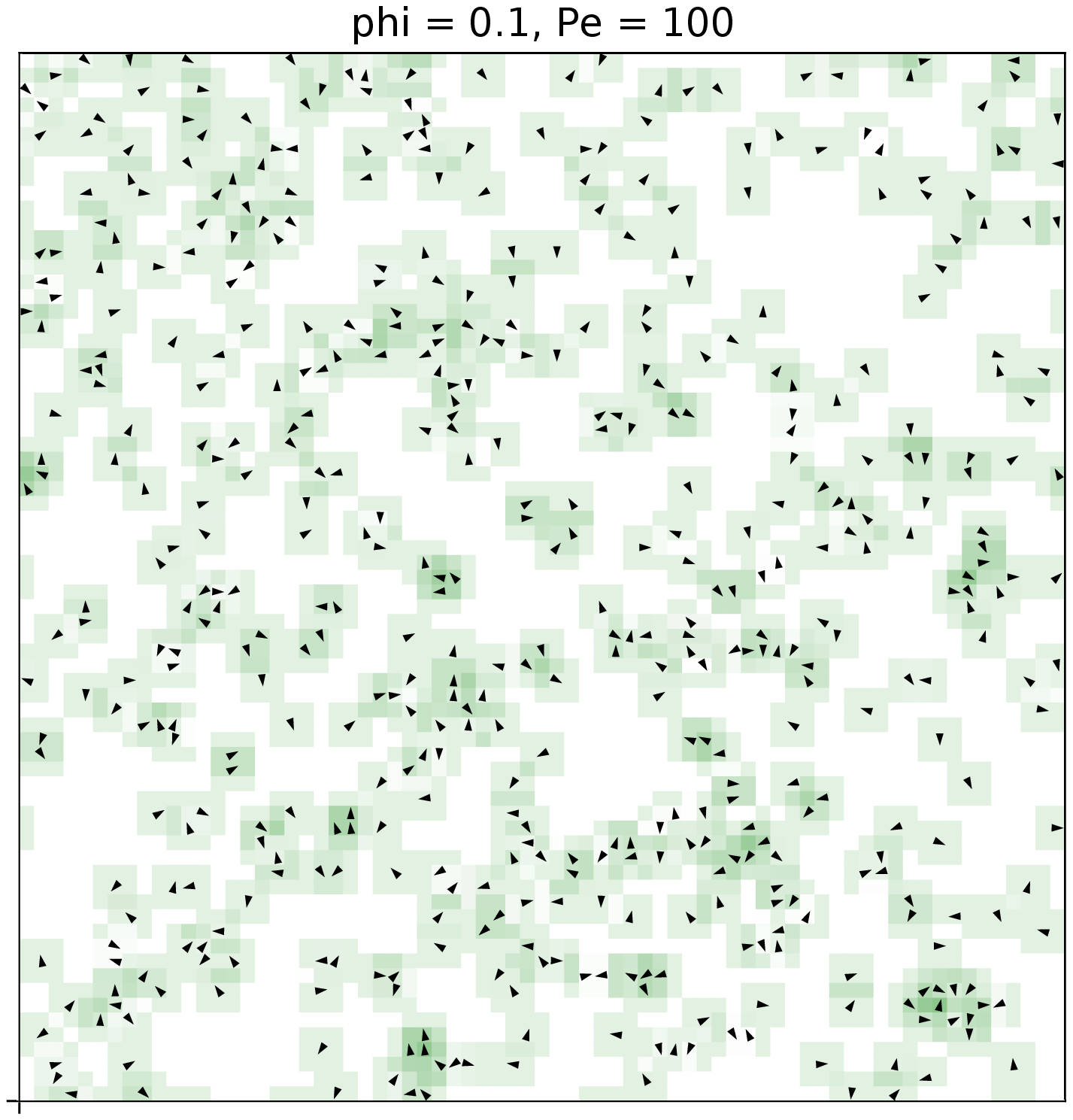}\includegraphics[width = 0.40\textwidth]{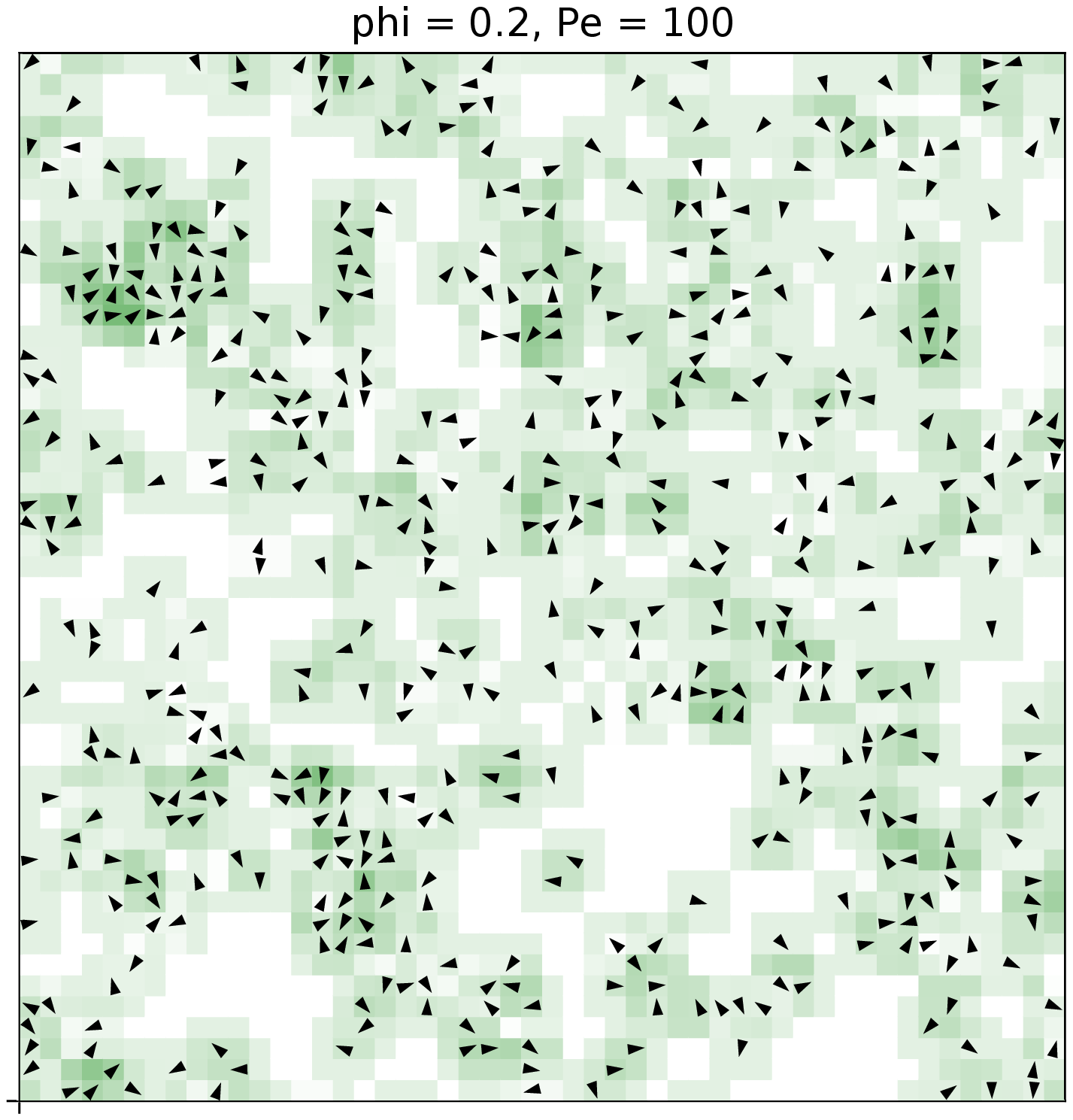} \\
	\includegraphics[width = 0.40\textwidth]{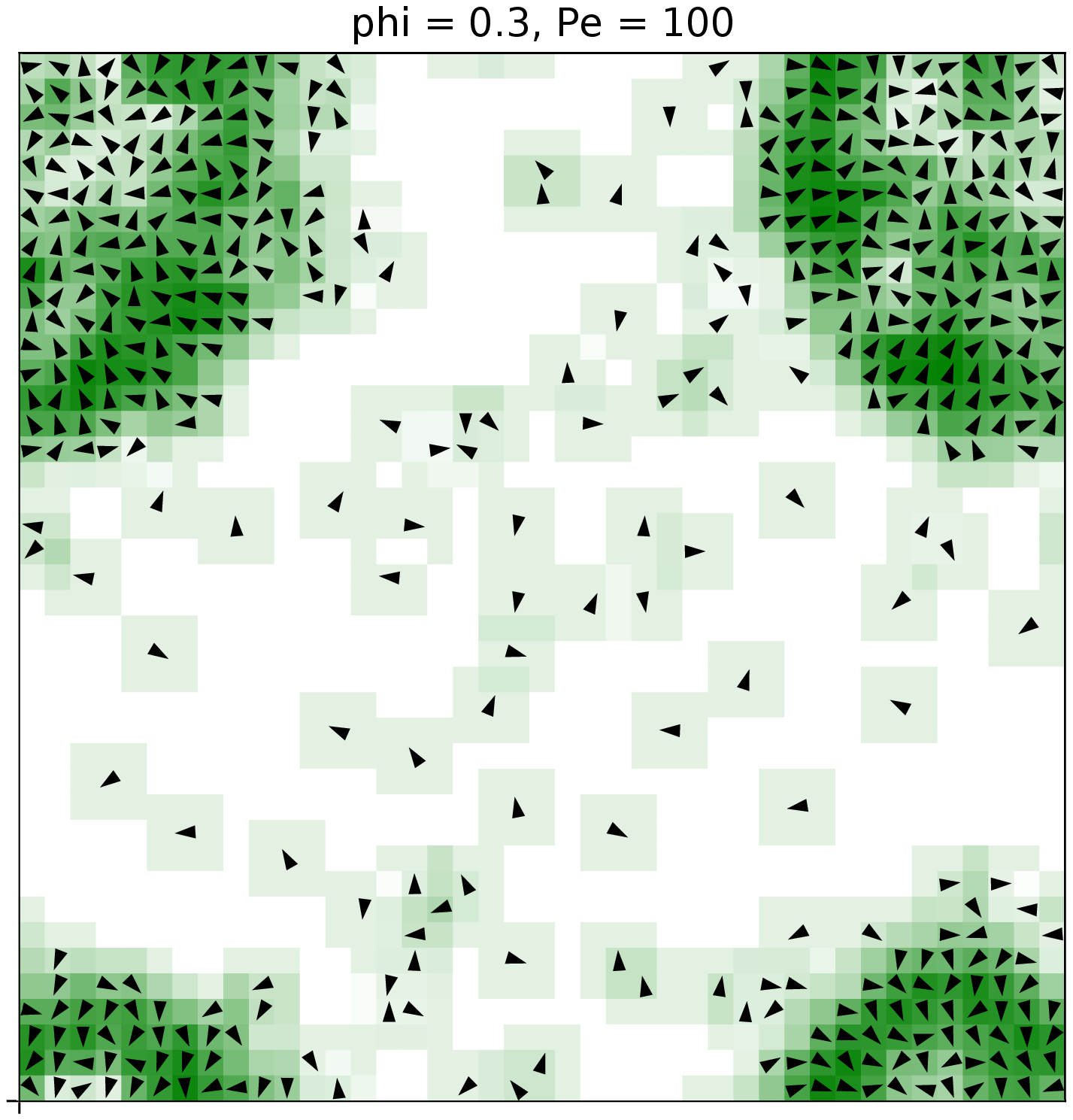} 
	\includegraphics[width = 0.40\textwidth]{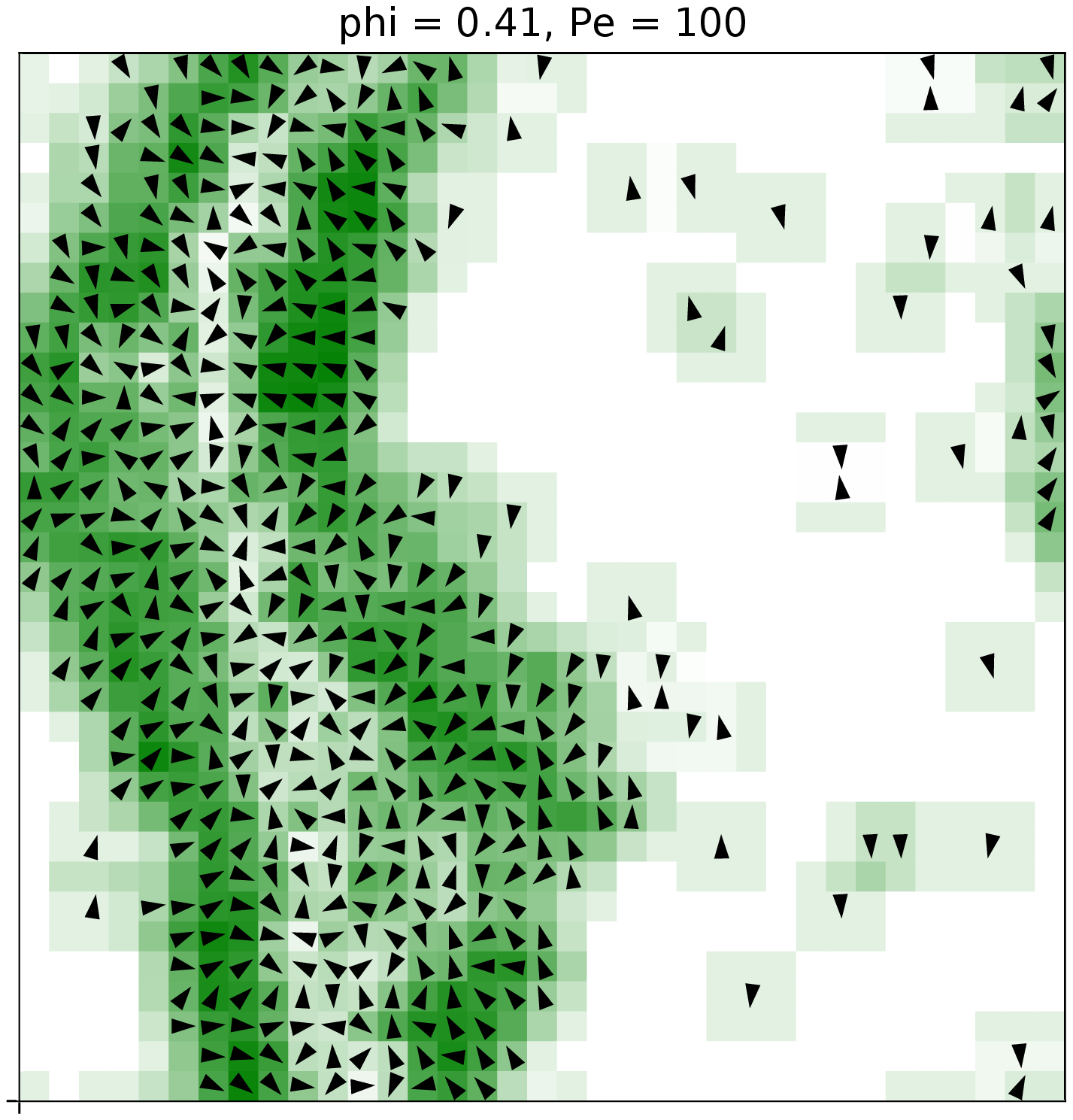} \\
	\includegraphics[width = 0.40\textwidth]{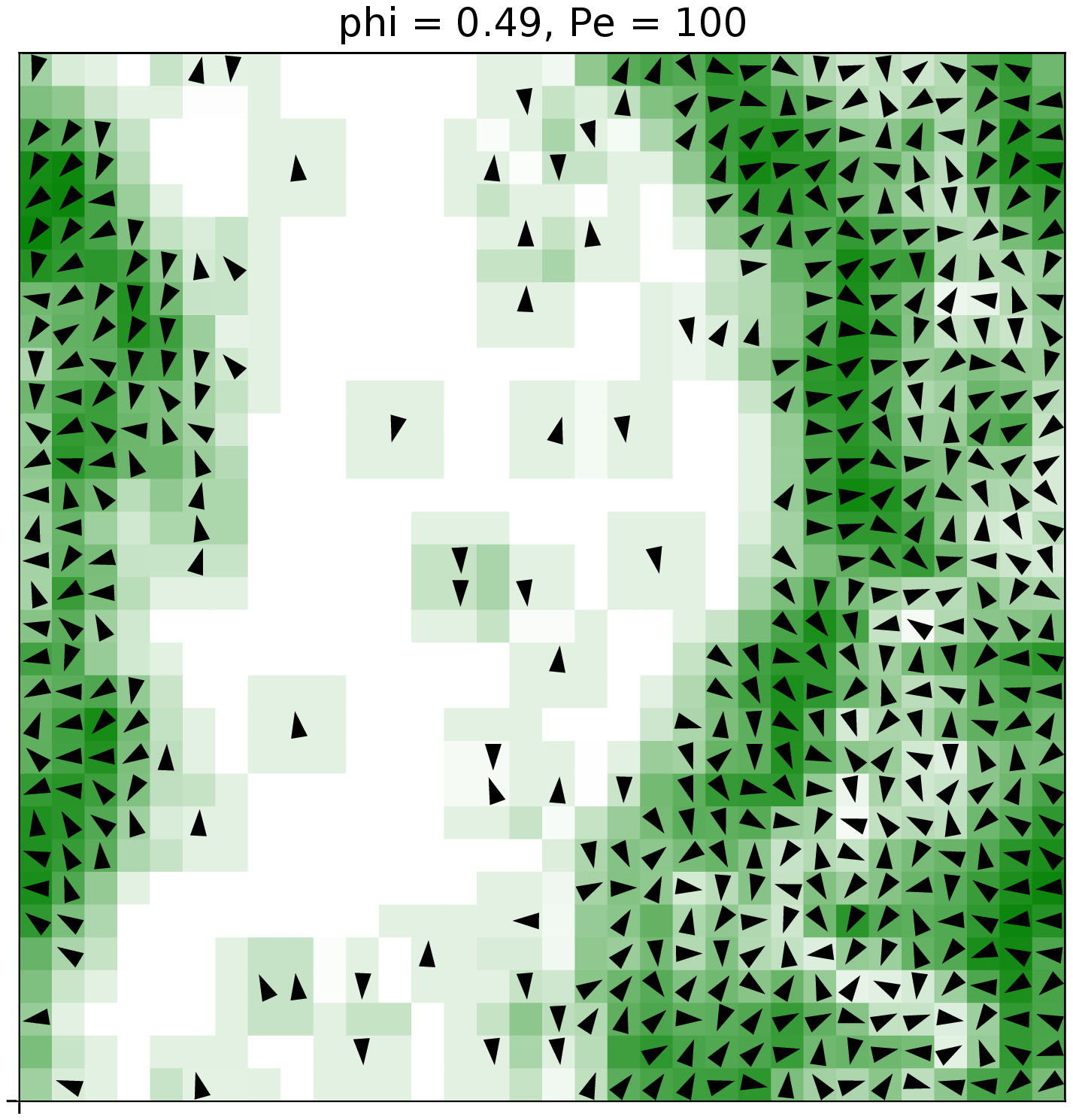}\includegraphics[width = 0.40\textwidth]{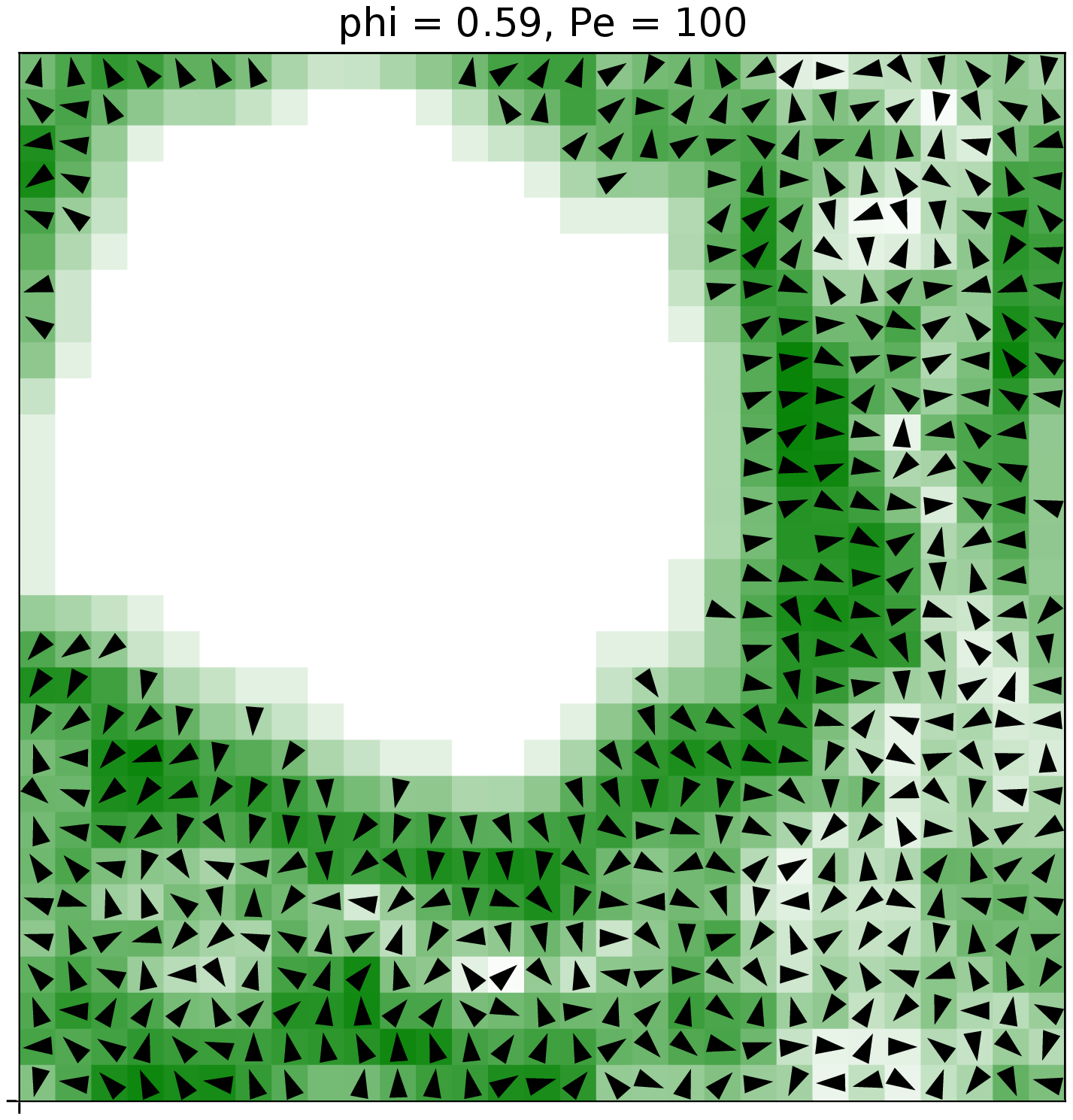} \\
	\includegraphics[width = 0.40\textwidth]{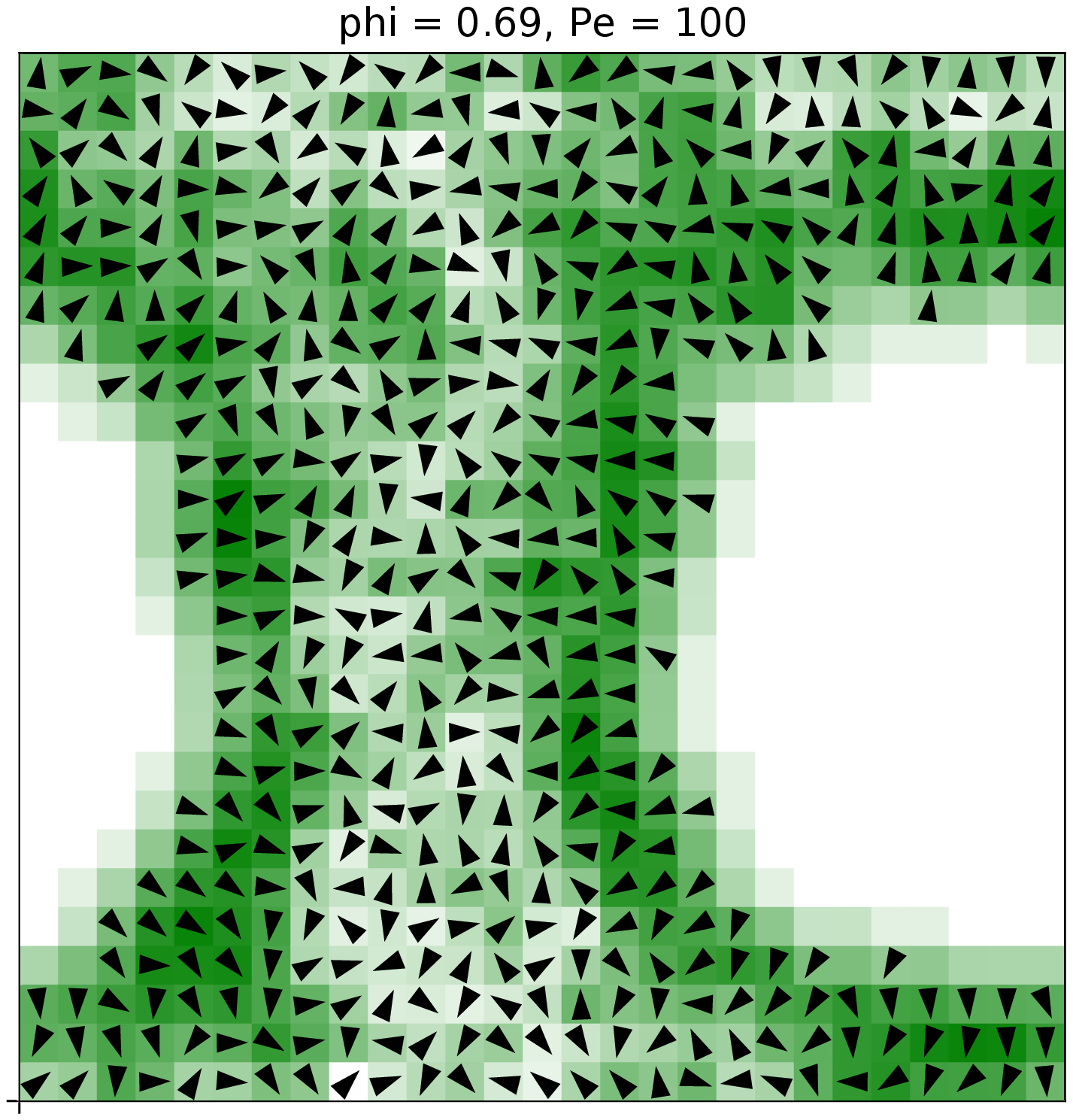}\includegraphics[width = 0.40\textwidth]{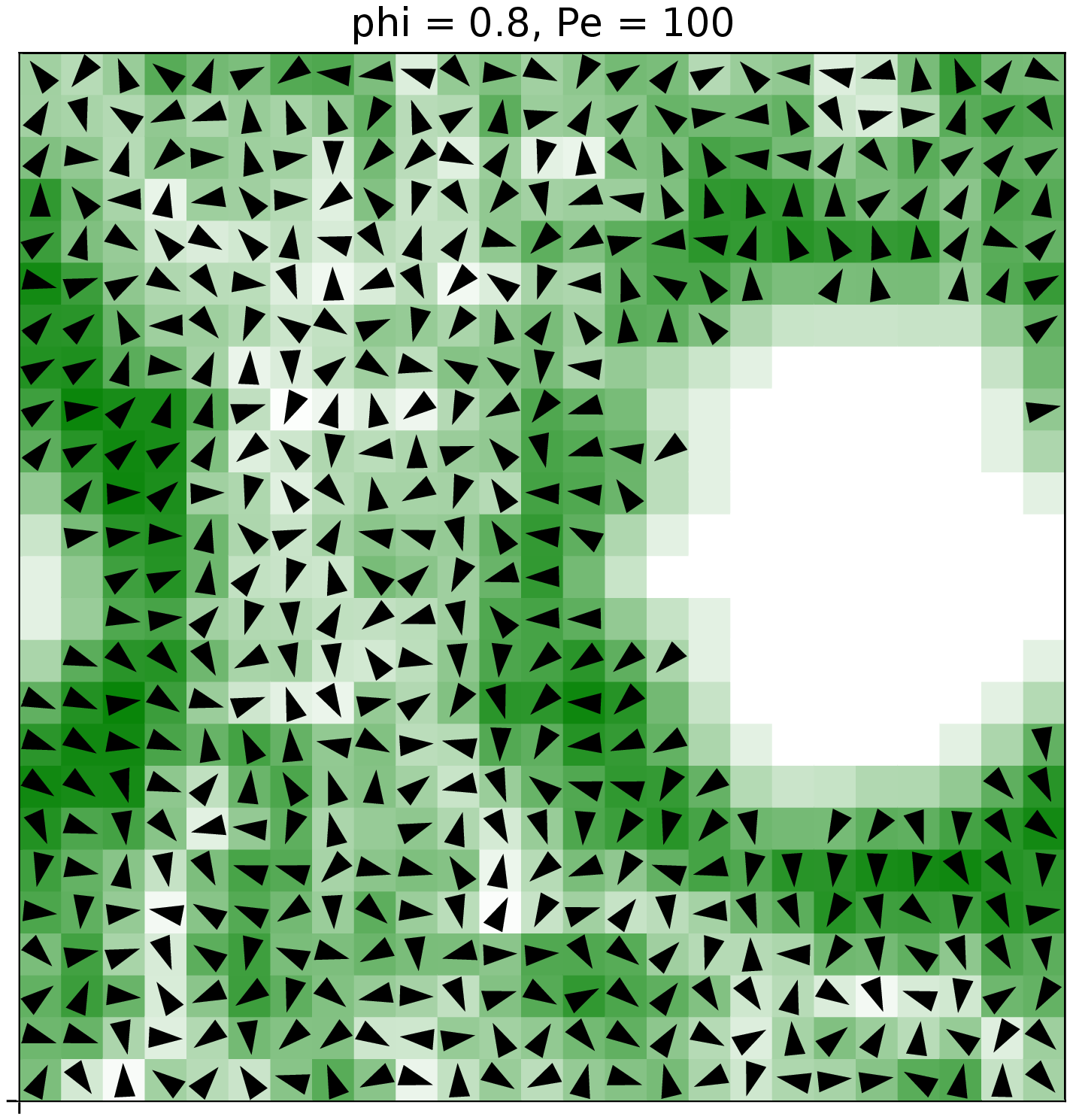} 
\end{center}	
  \caption{Snapshot of the microscopic lattice Model 4 at time $T = 1$ with $\Pe = 100$ and increasing values of $\phi = 0.1, 0.2, \dots, 0.8$.}
 \label{fig:ABM4_Tf=1.0_v0=100.0}
\end{figure}

\begin{figure}
\begin{center}
	\includegraphics[width = 0.40\textwidth]{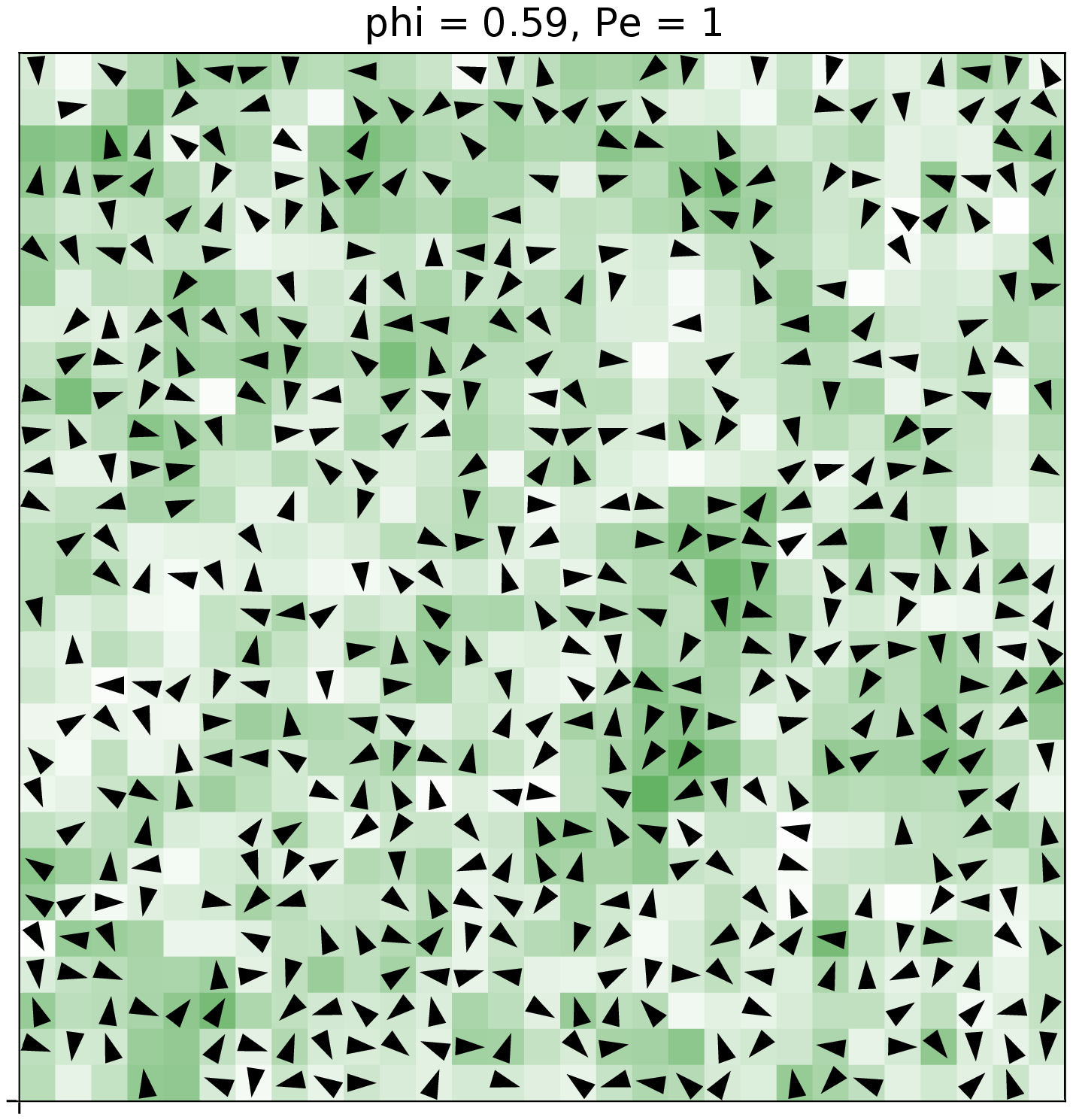}
	\includegraphics[width = 0.40\textwidth]{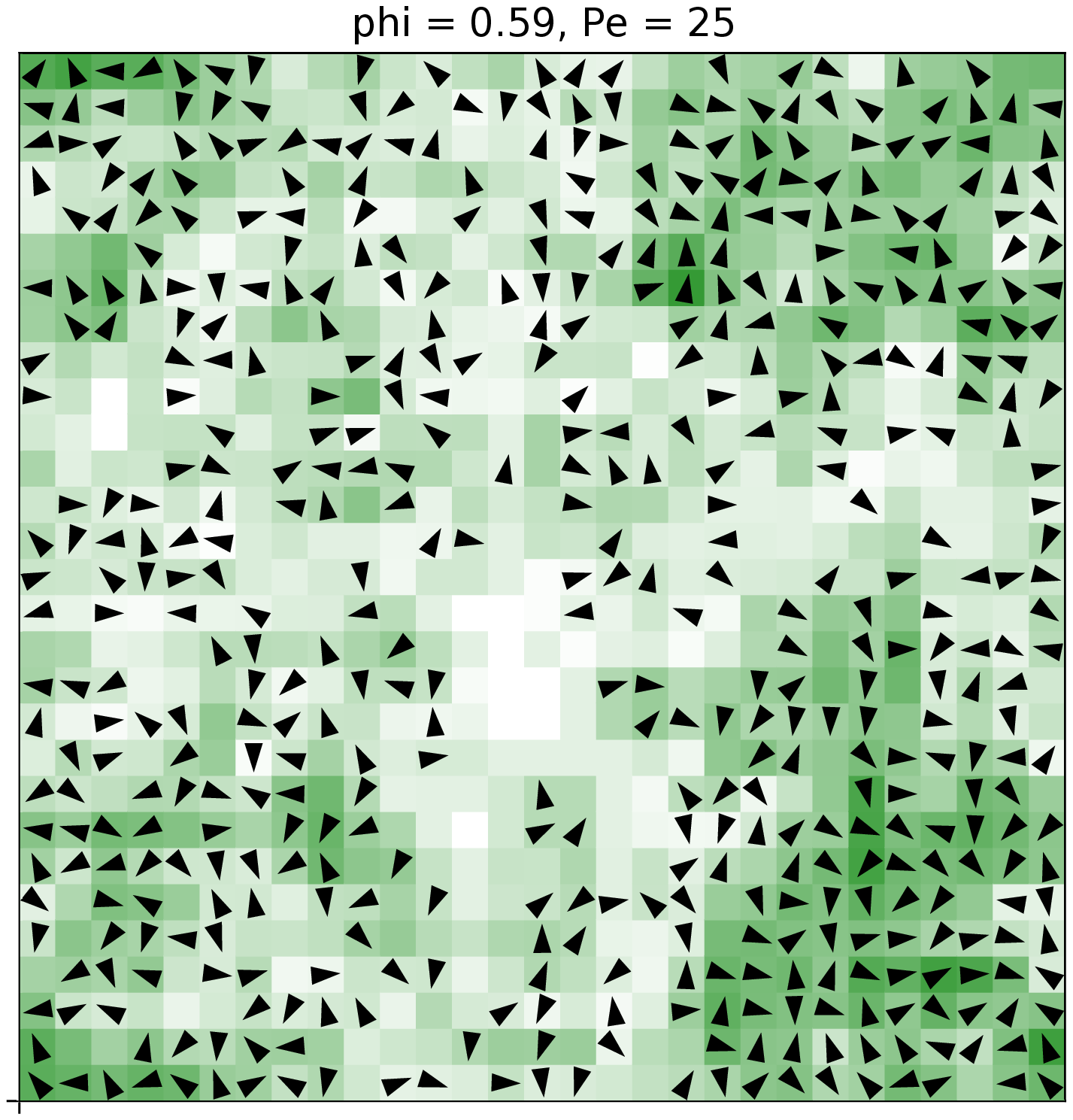} \\
	\includegraphics[width = 0.40\textwidth]{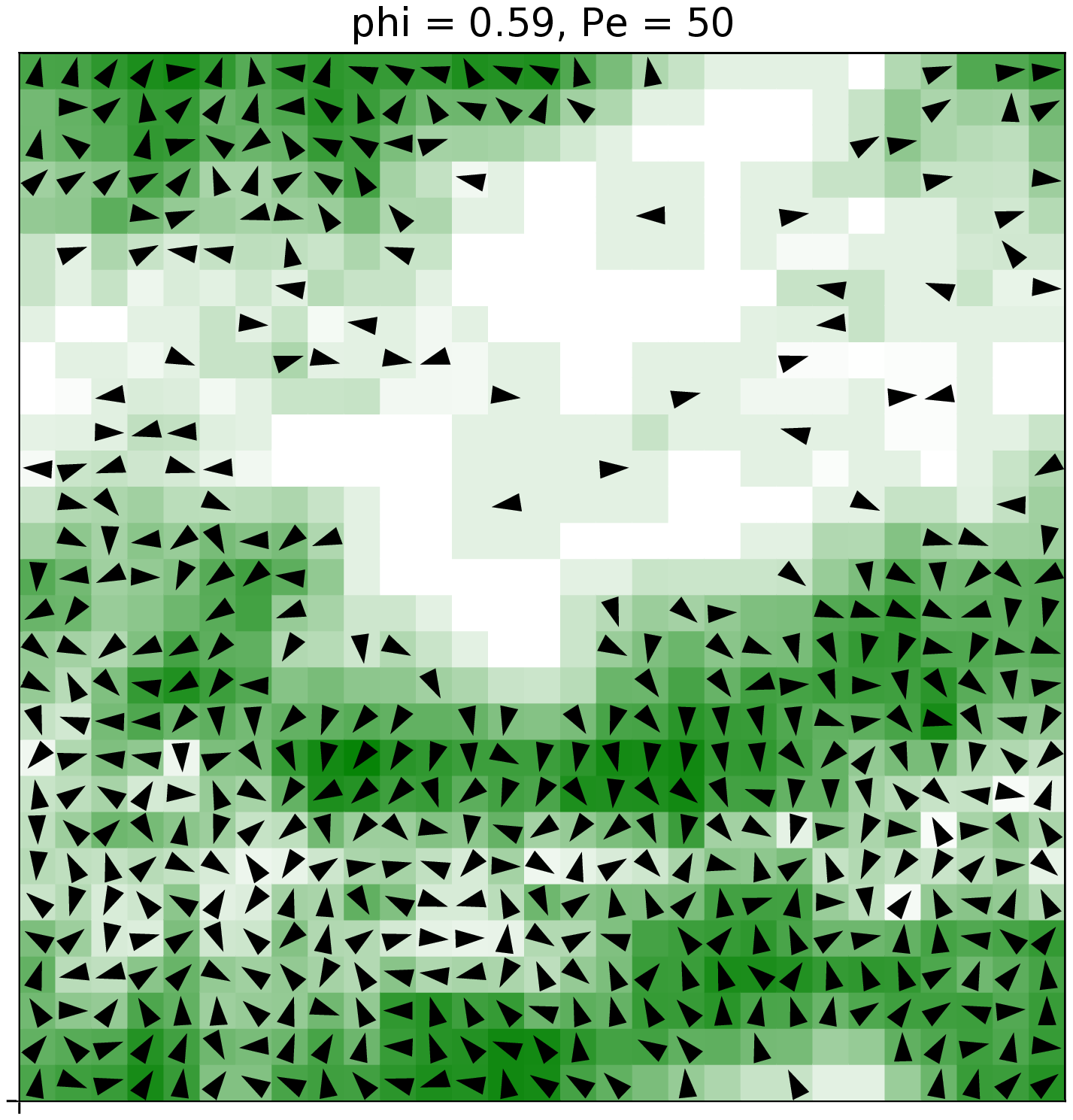}
	\includegraphics[width = 0.40\textwidth]{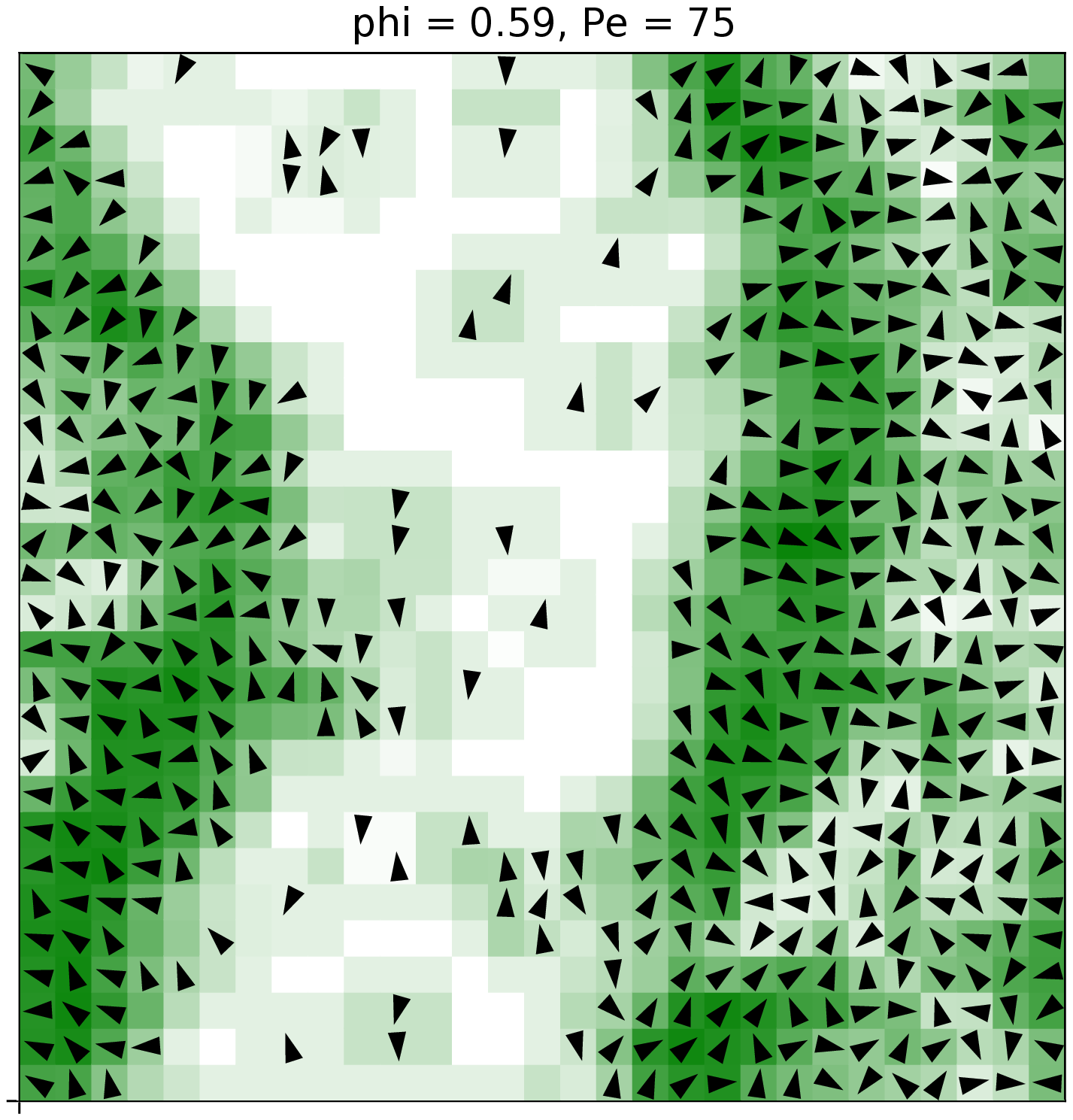}
\end{center}	
  \caption{Snapshot of the microscopic lattice Model 4 at time $T = 1$ with $\phi = 0.59$  and increasing values of $\Pe = 1, 25, 50, 75$.}
 \label{fig:ABM4_Tf=1.0_phi=0.6}
\end{figure}

We introduce $p_9(N, M)$ as the probability that, given $N$ uniformly distributed particles in a lattice with $M$ sites, a particle has its eight Moore  neighbouring sites occupied. For $M > N \ge 9$, this is given by
\begin{equation}
	\label{eq:p9unif}
	p_9(N, M) = \frac{(N-1)(N-2) \dots (N-8)}{(M-1)(M-2) \dots (M-8)}.
\end{equation}
In the main text, we use $\phi = N/M$ to write $P_{9,unif}(\phi) = P_{9}(\phi, 0) := p_9(N, N/\phi)$ while keeping the dependency on $N$ implicit (as we keep $N= 500$ fixed throughout the experiments). 
In the simulations, we measure $P_9$ by adding all the particles $\X_i$ have full Moore neighbourhoods $\mathcal N_9(\X_i)$ and dividing by the total number of particles $N$. We use the difference between the measured value of $P_9(\phi, \Pe)$ for given values of $\phi$ and $\Pe$ at time $T=1$ and the theoretical value for $P_9(\phi, 0)$ using \cref{eq:p9unif} to quantify the degree of non-uniformity in the system. \cref{fig:model4_ABM_metrics_sup}(left) shows the difference $\Delta P_9 := P_9 - P_{9,unif}$ for a range of values of $\phi, \Pe$, whereas \cref{fig:model4_ABM_metrics_sup}(right) shows the same data in a different format and using the relative difference $\delta P_9:= \Delta P_9/P_9$ (which we set to zero whenever $P_9 = 0$), so that the boundary between close to uniform and phase-segregated cases is very clear (especially for low values of $\phi$). 
\begin{figure}
\begin{minipage}[c]{0.45\linewidth}
\centering
\includegraphics[width = \textwidth]{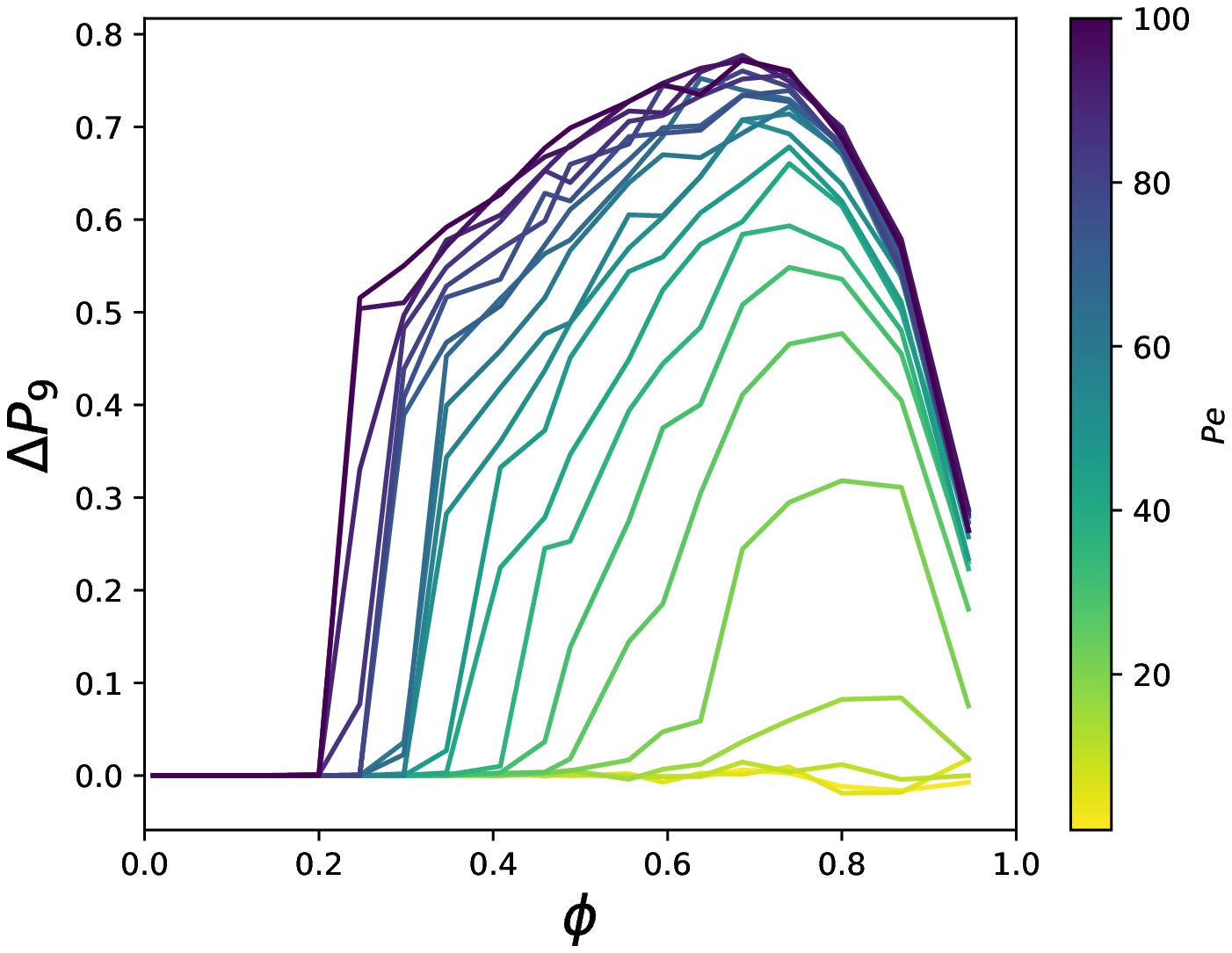}
\end{minipage}
\begin{minipage}[c]{0.55\linewidth}
\centering
\includegraphics[width =\textwidth]{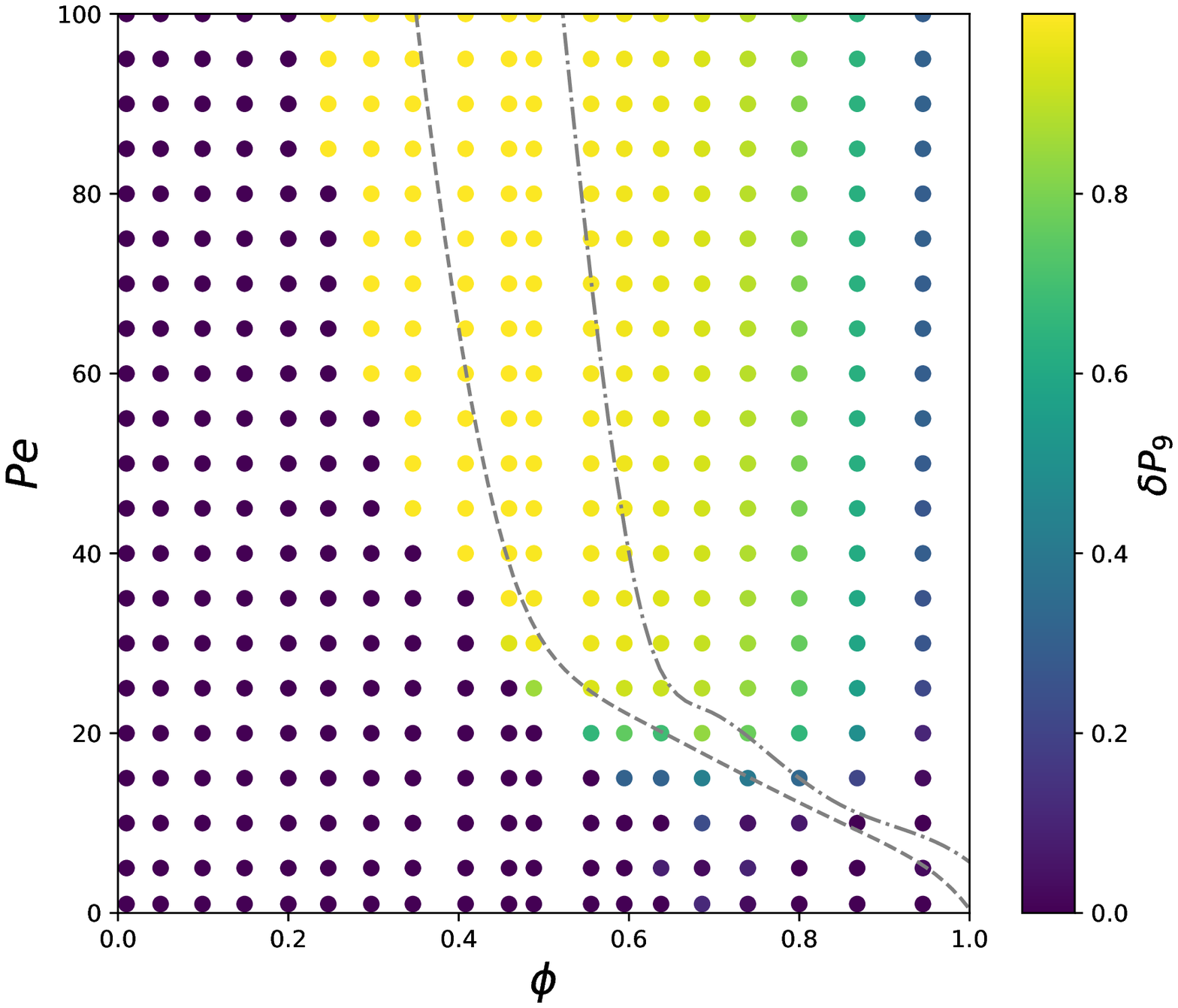}
\end{minipage}
\caption{Cluster fraction $P_{9}(\phi, \Pe)$ in the stochastic Model 4 for varying values of $\phi$ and $\Pe$ relative to $P_{9,unif}$, which is the cluster fraction when particles are uniformly distributed \cref{eq:p9unif}. Left: curves $\Delta P_9 := P_9 - P_{9,unif}$ for various $\Pe$ computed from simulations of Model 4. Right: colormap of the relative difference $\delta P_9:= \Delta P_9/P_9$. The two grey lines correspond to the boundaries between linear and Lyapunov stability and Lyapunov stability and instability from the PDE model \cref{model4_c} (as shown in \cref{fig:fulldispersion}). }
 \label{fig:model4_ABM_metrics_sup}
\end{figure}

\end{document}